\def\confversion{0}
\def\ifconf{\ifnum\confversion=1}
\def\ifnotconf{\ifnum\confversion=0}

\documentclass[11pt]{article}

\def\showauthornotes{0}
\def\showkeys{0}
\def\showdraftbox{0}

\usepackage[utf8]{inputenc} 
\usepackage[T1]{fontenc} 
\usepackage{ae} 
\usepackage{aecompl} 

\usepackage{amsmath} 
\usepackage{amsthm} 
\usepackage{amsfonts} 
\usepackage{amssymb} 
\usepackage{stmaryrd} 
\usepackage{mathtools} 
\usepackage{mathdots} 
\usepackage{bm} 
\usepackage{bbm} 

\usepackage[inline]{enumitem} 
\usepackage{array} 

\usepackage{tikz} 
\usetikzlibrary{cd} 

\usepackage[in]{fullpage} 
\usepackage{setspace} 
\usepackage{indentfirst} 

\usepackage{boththeorems} 

\def\({\left(}
\def\){\right)}
\def\[{\left[}
\def\]{\right]}
\def\<{\left\langle}
\def\>{\right\rangle}

\def\iff{\Longleftrightarrow}
\def\implies{\Longrightarrow}

\newcommand{\floor}[1]{\ensuremath{\left\lfloor#1\right\rfloor}}

\newcommand{\symdiff}{\mathbin{\triangle}}
\newcommand{\comp}{\mathbin{\circ}}

\newcommand{\place}{{{}\cdot{}}}

\let\geq\geqslant
\let\leq\leqslant

\let\succeq\succcurlyeq
\let\preceq\preccurlyeq

\let\:\colon

\let\epsilon\varepsilon

\def\CC{\ensuremath{\mathbb{C}}}

\def\FF{\ensuremath{\mathbb{F}}}

\def\HH{\ensuremath{\mathbb{H}}}

\def\NN{\ensuremath{\mathbb{N}}}

\def\PP{\ensuremath{\mathbb{P}}}

\def\RR{\ensuremath{\mathbb{R}}}

\def\ZZ{\ensuremath{\mathbb{Z}}}

\def\One{\ensuremath{\mathbbm{1}}}

\def\cA{\ensuremath{\mathcal{A}}}

\def\cD{\ensuremath{\mathcal{D}}}

\def\cF{\ensuremath{\mathcal{F}}}

\def\cL{\ensuremath{\mathcal{L}}}

\def\cS{\ensuremath{\mathcal{S}}}

\usepackage{xspace,xcolor}
\usepackage{amsmath,amssymb}
\usepackage{amsthm}
\usepackage[toc,page]{appendix}
\usepackage{thmtools}
\usepackage{thm-restate}
\usepackage{color,graphicx}
\usepackage{boxedminipage}
\usepackage{makecell}
\usepackage{tabularx}

\ifnum\showkeys=1
\usepackage[color]{showkeys}
\fi

\definecolor{darkred}{rgb}{0.5,0,0}
\definecolor{darkgreen}{rgb}{0,0.35,0}
\definecolor{darkblue}{rgb}{0,0,0.55}

\usepackage[pdfstartview=FitH,pdfpagemode=UseNone,colorlinks,linkcolor=darkblue,filecolor=darkred,citecolor=darkgreen,urlcolor=darkred,pagebackref]{hyperref}

\usepackage[capitalise,nameinlink]{cleveref}
\usepackage[T1]{fontenc}
\usepackage{mathtools,dsfont,bbm}
\usepackage{mathpazo}

\ifnum\showauthornotes=1
\newcommand{\Authornote}[2]{{\sf\small\color{red}{[#1: #2]}}}
\newcommand{\Authorcomment}[2]{{\sf \small\color{gray}{[#1: #2]}}}
\newcommand{\Authorfnote}[2]{\footnote{\color{red}{#1: #2}}}
\else
\newcommand{\Authornote}[2]{}
\newcommand{\Authorcomment}[2]{}
\newcommand{\Authorfnote}[2]{}
\fi

\ifnum\showdraftbox=1
\newcommand{\draftbox}{\begin{center}
  \fbox{%
    \begin{minipage}{2in}%
      \begin{center}%
        \begin{Large}%
          \textsc{Working Draft}%
        \end{Large}\\
        Please do not distribute%
      \end{center}%
    \end{minipage}%
  }%
\end{center}
\vspace{0.2cm}}
\else
\newcommand{\draftbox}{}
\fi

\def\to{\rightarrow}

\def\epsilon{\varepsilon}

\def\phi{\varphi}
\def\cal{\mathcal}

\def\implies{\Rightarrow}

\def\psdgeq{\succeq} 
\newcommand{\defeq}{:=}

\newcommand{\etal}{et al.\xspace}

\newcommand{\R}{{\mathbb R}}

\newcommand{\N}{{\mathbb{N}}}

\newcommand{\F}{{\mathbb F}}

\usepackage{nicefrac}

\newcommand{\abs}[1]{\ensuremath{\left\lvert #1 \right\rvert}}

\newcommand{\ip}[2] {\ensuremath{\left\langle #1 , #2 \right\rangle}}

\newcommand{\one}{{\mathbf{1}}}

\newcommand{\Esymb}{\mathbb{E}}

\DeclareMathOperator*{\ExpOp}{\Esymb}

\makeatletter

\def\ProbabilityRender#1#2{
  \@ifnextchar\bgroup%
  {\renderwithdist{#1}{#2}}
   {\singlervrender{#1}{#2}}
}
\def\singlervrender#1#2{%
   \ensuremath{\mathchoice
       {{#1}\left[ #2 \right]}
       {{#1}[ #2 ]}
       {{#1}[ #2 ]}
       {{#1}[ #2 ]}
   }
}
\def\renderwithdist#1#2#3{%
   \@ifnextchar\bgroup
   {\superfancyrender{#1}{#2}{#3}}
   {\ensuremath{\mathchoice
      {\underset{#2}{#1}\left[ #3 \right]}
      {{#1}_{#2}[ #3 ]}
      {{#1}_{#2}[ #3 ]}
      {{#1}_{#2}[ #3 ]}
     }
   }
}
\def\superfancyrender#1#2#3#4#5{
   \ensuremath{\mathchoice
      {\underset{#1}{{#1}}\left#4 #3 \right#5}
      {{#1}_{#2}#4 #3 #5}
      {{#1}_{#2}#4 #3 #5}
      {{#1}_{#2}#4 #3 #5}
   }
}
\makeatother

\newfont{\inhead}{eufm10 scaled\magstep1}

\newcommand{\calP}{{\cal P}}

\newcommand{\poly}{{\mathrm{poly}}}

\DeclareMathOperator\supp{Supp}

\DeclareMathOperator{\tr}{\operatorname {tr}}
\DeclareMathOperator{\Span}{\operatorname {span}}

\DeclareMathOperator{\im}{\operatorname{im}}

\newcommand{\el}{\ensuremath{\ell} }

\renewcommand{\iff}{\ensuremath{\Leftrightarrow}}

\newcommand{\val}{{\sf val}}

\newcommand{\Lovasz}{Lov\'asz\xspace}

\newcommand{\KLP}{\textup{KrawtchoukLP}}
\newcommand{\DLP}{\textup{DelsarteLP}}

\newtheoremstyle{plain}
  {\medskipamount}
  {\smallskipamount}
  {\slshape}
  {0pt}
  {\bfseries}
  {.}
  { }
  {\thmname{#1}\thmnumber{ #2}{\normalfont\thmnote{ (#3)}}}

\theoremstyle{plain}

\newboththeorems{theorem}{Theorem}[section]
\newboththeorems{lemma}[theorem]{Lemma}
\newboththeorems{proposition}[theorem]{Proposition}
\newboththeorems{corollary}[theorem]{Corollary}
\newboththeorems{afirmation}[theorem]{Afirmation}
\newboththeorems{fact}[theorem]{Fact}
\newboththeorems{claim}[theorem]{Claim}
\newboththeorems{conjecture}[theorem]{Conjecture}
\newboththeorems{problem}[theorem]{Problem}
\newboththeorems{question}[theorem]{Question}
\newboththeorems{observation}[theorem]{Observation}
\newboththeorems{remark}[theorem]{Remark}
\newboththeorems{definition}[theorem]{Definition}
\newboththeorems{notation}[theorem]{Notation}
\newboththeorems{example}[theorem]{Example}
\newboththeorems{exercise}[theorem]{Exercise}

{\end{proof}\endgroup}

\newenvironment{proofsketch}[1][Proof (sketch)]{%
\begingroup
\begin{proof}[#1]}%
{\end{proof}\endgroup}

{\end{proofsketch}}

{\end{proof}\endgroup}

{\end{proof}\endgroup}

{\end{enumerate}\medskip\endgroup}

\usepackage{thm-restate} 
\usepackage{float}
\usepackage{etoolbox}
\AtBeginEnvironment{quote}{\par\singlespacing\small}
\usepackage{xfrac}

\setlist[enumerate]{label={\roman*.}, ref={(\roman*)}}

\def\rn{\bm}
\newcommand{\df}{\stackrel{\text{def}}{=}}
\let\emph\textit
\def\supp{\operatorname{supp}}

\newcommand{\config}{\textup{\textsf{Config}}}
\newcommand{\forbconfig}{\textup{\textsf{ForbConfig}}}
\DeclareMathOperator{\Aut}{Aut}

\DeclareMathOperator{\Ann}{Ann}
\DeclareMathOperator{\id}{id}
\DeclareMathOperator{\Stab}{Stab}

\newcommand{\unif}{\in_{\text{R}}}
\newcommand{\TransSDP}{\textup{TranslationSDP}}
\newcommand{\FLP}{\textup{FourierLP}}

\newcommand{\Lin}{\textup{Lin}}

\DeclareMathOperator{\GL}{GL}

\definecolor{orange}{HTML}{FF7F00}

\makeatletter
\def\st@r#1#2{\m@th\ooalign{$\hfil#1*\hfil$\cr$#1#2$}}
\def\starprod{\mathop{\mathpalette\st@r\prod}}
\makeatother

\title{A Complete Linear Programming Hierarchy for Linear Codes}

\author{Leonardo Nagami Coregliano\thanks{{\tt IAS}. {\tt lenacore@ias.edu}. This material is based upon work supported by the National Science Foundation, and by the IAS School of Mathematics.} \and
        Fernando Granha Jeronimo\thanks{{\tt IAS}. {\tt granha@ias.edu}. This material is based upon work supported by the National Science Foundation under Grant No. CCF-1900460.} \and
        Chris Jones \thanks{{\tt UChicago}. {\tt csj@uchicago.edu}. This material is based upon work supported by the National Science Foundation under Grant No. CCF-2008920.}}

\begin{document}
\maketitle
\pagenumbering{roman}

\draftbox

\begin{abstract}  
  A longstanding open problem in coding theory is to determine the
  best (asymptotic) rate $R_2(\delta)$ of binary codes with minimum
  constant (relative) distance $\delta$. An existential lower bound
  was given by Gilbert and Varshamov in the 1950s. On the
  impossibility side, in the 1970s McEliece, Rodemich, Rumsey and
  Welch (MRRW) proved an upper bound by analyzing Delsarte's linear
  programs. To date these results remain the best known lower and
  upper bounds on $R_2(\delta)$ with no improvement even for the
  important class of \emph{linear} codes.  Asymptotically, these
  bounds differ by an exponential factor in the blocklength.

  In this work, we introduce a new hierarchy of linear programs (LPs)
  that converges to the true size $A^{\textup{Lin}}_2(n,d)$ of an
  optimum \emph{linear} binary code (in fact, over any finite field)
  of a given blocklength $n$ and distance $d$.
  This hierarchy has several notable features:
  \begin{enumerate}
     \item It is a natural generalization of the Delsarte LPs used in the first MRRW bound.

     \item It is a hierarchy of linear programs rather than semi-definite programs potentially making
           it more amenable to theoretical analysis.

    \item It is \emph{complete} in the sense that the optimum code size can be retrieved
          from level $O(n^2)$.
                        
    \item It provides an answer in the form of a hierarchy (in larger dimensional spaces) to the question of
          how to cut Delsarte's LP polytopes to approximate the true size of \emph{linear} codes.
  \end{enumerate}
  We obtain our hierarchy by generalizing the Krawtchouk polynomials
  and MacWilliams inequalities to a suitable ``higher-order'' version
  taking into account interactions of $\ell$ words. Our method also
  generalizes to translation schemes under mild assumptions.
\end{abstract}


\thispagestyle{empty}

\newpage
\tableofcontents
\clearpage

\pagenumbering{arabic}
\setcounter{page}{1}

\section{Introduction}

A fundamental question in coding theory is the maximum size of a
binary code given a blocklength parameter $n$ and a minimum distance
parameter $d_n$. This value is typically denoted by $A_2(n,d_n)$. A
particularly important regime occurs when $\lim_{n\to\infty} d_n/n =
\delta$ for some absolute constant $\delta \in (0,1/2)$. In this
regime, $A_2(n,d_n)$ is known to grow exponentially in $n$. However,
the precise rate of this exponential growth remains an elusive major
open problem. It is often convenient to denote the asymptotic basis of
this growth as $2^{R_2(\delta)}$, where the rate $R_2(\delta)$ is defined
as
\begin{align*}
  R_2(\delta) \coloneqq \limsup_{n\to \infty} \frac{1}{n} \log_2\left(A_2(n,\lfloor \delta n \rfloor)\right).
\end{align*}

An equivalent way of defining $A_2(n,d)$ is as the independence number
of the graph $H_{n,d}$ whose vertex set is $V(H_{n,d})\coloneqq\F_2^n$ and two
vertices $x,y \in V(H_{n,d})$ are adjacent if and only if their
Hamming distance $\Delta(x,y)$ lies in $\{1,\ldots,d-1\}$. Note that
there is a one-to-one correspondence between independent sets in this
graph and binary codes of blocklength $n$ and minimum distance
$d$. Most of the literature about $A_2(n,d)$ takes advantage of this
graph-theoretic interpretation.

A lower bound on $A_2(n,d)$ follows from the trivial degree bound on
the independence number of a graph, namely, $\alpha(H_{n,d}) \ge
\lvert V(H_{n,d}) \rvert/(\operatorname{deg}(H_{n,d})+1)$, which gives
$\alpha(H_{n,d}) \ge 2^{(1-h_2(d/n) + o(1))n}$ where $h_2$ is the
binary entropy function. First discovered
by Gilbert \cite{G52} and later generalized to linear codes by
Varshamov \cite{V57}, this existential bound is now known as the
Gilbert--Varshamov (GV) bound. Observe that the GV~bound readily
implies that $R_2(\delta) \ge 1-h_2(\delta)$. Despite its simplicity,
this bound remains the best (existential) lower bound on
$R_2(\delta)$.

The techniques to upper bound $A_2(n,d)$ are oftentimes more involved,
with the most prominent being the Delsarte linear programming method
that we now describe. A binary code ${C \subseteq \F_2^n}$ is
\emph{linear} if it is a subspace of $\F_2^n$ and its weight
distribution is the tuple $(a_0,a_1,\ldots,a_n) \in \mathbb{N}^{n+1}$,
where $a_i$ is the number of codewords of $C$ of Hamming weight
$i$. MacWilliams \cite{Macwilliams63} showed that the weight
distribution $(b_0,b_1,\ldots,b_n)$ of the dual code $C^{\perp}$ can
be obtained by applying a linear transformation to
$(a_0,a_1,\ldots,a_n)$. More precisely, the MacWilliams identities
establish that
\begin{align*}
b_j = \frac{1}{\abs{C}} \sum_{i=0}^n K_j(i) \cdot a_i,
\end{align*}
where the coefficients $K_j(i)$ are evaluations of the so-called
Krawtchouk (or Kravchuk) polynomial of degree $j$. The Krawtchouk
polynomials form a family of orthogonal polynomials under the measure
$\mu_n(i) = \binom{n}{i}/2^n$ and they play an important role in
coding theory~\cite[Chapter 1]{vanLint99}. Since the weight
distribution of the dual $C^{\perp}$ is non-negative, the
MacWilliams identities can be relaxed to inequalities
\begin{align*}
\sum_{i=0}^n K_j(i) \cdot a_i \ge 0
\end{align*}
for $j=0,\ldots,n$. This naturally leads to the following linear
program (LP) relaxation for $A_2(n,d)$ when $C$ is a linear code
(recall that for a linear code, having distance at least $d$ is equivalent
to having no words of Hamming weight 1 through $d-1$).
\begin{figure}[H]
  \begin{align*}
    \max \quad
    & \sum_{i=0}^n a_i
    \\
    \text{s.t.} \quad
    & a_0 = 1
    & &
    & & (\text{Normalization})
    \\
    & a_i = 0
    & & \text{for } i=1,\ldots,d-1
    & &  (\text{Distance constraints})
    \\
    & \sum_{j=1}^n K_i(j)\cdot a_j \geq 0
    & & \text{for } i=0,\ldots,n
    & & (\text{MacWilliams inequalities})
    \\
    & a_i \geq 0
    & & \text{for } i=0,\ldots,n
    & & (\text{Non-negativity}).
  \end{align*}
  \caption{Delsarte's linear program for $A_2(n,d)$.}\label{lp:delsarte}
\end{figure}
The $a_i$ can be suitably generalized to codes which are not
necessarily linear (by setting $a_i\coloneqq \abs{\{(x,y) \in C^2 \vert
  \Delta(x,y)=i\}}/\abs{ C}$). The MacWilliams inequalities
hold for these generalized $a_i$'s as proven by MacWilliams, Sloane
and Goethals~\cite{MSG72}. Therefore, the same linear program above
also bounds $A_2(n,d)$ for general codes. This family of linear
programs was first introduced by Delsarte
in~\cite{delsarte1973algebraic}, where it was obtained in greater
generality from the theory of association schemes. We refer to the
above linear program as Delsarte's linear program, or, more formally,
as $\DLP(n,d)$.

The best known upper bound on $R_2(\delta)$ for distances $\delta \in
(0.273, 1/2)$ is obtained by constructing solutions to the dual
program of Delsarte's linear program, as first done by McEliece,
Rodemich, Rumsey and Welch (MRRW) \cite{MRRW77} in their first linear
programming bound. In the same work, McEliece~\etal also gave the best
known bound for $\delta \in (0,0.273)$ via a second family of linear
programs. Since our techniques are more similar to their first linear
programming bound, we restrict our attention to it in this discussion.
In the first linear programming bound, they showed that $R_2(\delta) \le
h_2(1/2-\sqrt{\delta(1-\delta)})$ with a reasonably sophisticated
argument using properties of general orthogonal polynomials and also
particular properties of Krawtchouk polynomials. Simpler perspectives
on the first LP bound analysis were found by Navon and
Samorodnitsky~\cite{NS05} and by
Samorodnitsky~\cite{Samorodnitsky2021proof}.

Instead of linear programming, one can use more powerful techniques
based on semi-definite programming (SDP) to upper bound
$A_2(n,d)$. For instance, the Sum-of-Squares/Lasserre SDP hierarchy
was suggested for this problem by Laurent~\cite{Laurent09}. The value
of the program equals $\alpha(H_{n,d})$ for a sufficiently high level
of the hierarchy, so in principle analyzing these programs could give
$A_2(n,d)$ exactly. Analyzing SDP methods to improve $R_2(\delta)$
seems challenging and we do not even know how to analyze the simplest
of them~\cite{Schrijver05}, which is weaker than degree-$4$
Sum-of-Squares (see related work below for more details on SDP
methods).

On the one hand, we have reasonably simple linear programs of Delsarte
already requiring a non-trivial theoretical analysis for proving upper
bounds on $R_2(\delta)$. On the other hand we have more sophisticated
SDP methods which are provably stronger than the Delsarte LP, but for
which no theoretical analyses are known.

\subsection{Our Contribution}

In this work, we refine the Delsarte linear programming method used in
the first LP bound for $A_2(n,d)$ by designing a \emph{hierarchy} of
linear programs.  For a parameter $\el \in \N_+$, the hierarchy is based
on specific higher-order versions of Krawtchouk polynomials and
MacWilliams inequalities that take advantage of $\ell$-point
interactions of words. We denote by $\KLP(n,d,\ell)$ the linear
programming relaxation for $A_2(n,d)$ at level $\ell$ of our hierarchy.

We define $A^\Lin_2(n,d)$ analogously to $A_2(n,d)$ as the
maximum size of a \emph{linear} binary code of blocklength $n$ and
minimum distance $d$. For linear codes, we impose additional
``semantic'' constraints on the programs in the hierarchy taking
advantage of the linear structure of the code. We denote by
$\KLP_\Lin(n,d,\ell)$ the linear program relaxation for
$A^\Lin_2(n,d)$ with these additional constraints. Both
$\KLP_\Lin(n,d,1)$ and $\KLP(n,d,1)$ coincide with
$\DLP(n,d)$ at the first level of our hierarchy.

There is a known gap between the value of Delsarte's linear programs
and the GV~bound. In particular when $\delta = 1/2-\epsilon$,
Delsarte's linear programs do not yield an upper bound tighter than
$R_2(1/2-\epsilon) \le \Theta(\epsilon^2 \log(1/\epsilon))$, as shown
by Navon and Samorodnitsky~\cite{NS05}, whereas the GV~bound
establishes a lower bound of $R_2(1/2-\epsilon) \ge
\Omega(\epsilon^2)$. There are no known improvements to these bounds
even for the important class of \emph{linear} codes. If the GV~bound
is indeed tight, then analyzing $\DLP$ is not sufficient to prove it.
The goal of our hierarchy is to give tighter and tighter upper bounds
on $A_2(n,d)$ as the level of the hierarchy increases.

We show that for \emph{linear} codes the hierarchy $\KLP_\Lin(n, d,
\el)$ is \emph{complete}, meaning that the value of the hierarchy
converges to $A^\Lin_2(n,d)$ as $\el$ grows larger.  We prove that
level roughly $\el = O(n^2)$ is enough to retrieve the correct value
of $A^\Lin_2(n,d)$. More generally, for linear codes over $\F_q$, we
have the following completeness theorem for $A^\Lin_q(n,d)$.

\begin{theorem}[Completeness - Informal version of \cref{theo:main:completeness:formal}]\label{theo:completeness_informal}
  For $\el \ge \Omega_{\epsilon,q}(n^2)$, we have
  \begin{align*}
    A^\Lin_q(n,d) ~\leq~ 
    \val(\KLP^{\FF_q}_\Lin(n, d, \el))^{1/\ell} ~\leq~ (1+\epsilon) \cdot A^\Lin_q(n,d).
  \end{align*}
\end{theorem}

We think that the $\KLP_\Lin(n, d, \el)$ hierarchy is an extremely
interesting object for the following reasons.
\begin{enumerate}
     \item It is takes advantage of  higher-order interactions of codewords by
           naturally computing Hamming weight statistics of subspaces spanned by $\ell$ codewords
          (see \cref{def:configuration}).
  
    \item It is a generalization of the Delsarte LP used in the first MRRW bound and the two
    share strong structural similarities (see \cref{sec:binary_warmup}).

    \item It is a hierarchy of linear programs rather than semi-definite programs (see \cref{def:holp} and \cref{sec:diagonalization}).

    \item It is a \emph{complete} hierarchy (see \cref{theo:main:completeness:formal}).
                        
    \item It provides an answer in the form of a hierarchy (in larger dimensional spaces) to the question of
          how to cut Delsarte's LP polytopes~\cite{NS05} to approximate the true size of \emph{linear} codes.
\end{enumerate}

We hope this hierarchy will fill an important gap in the coding theory
literature between Delsarte's LP, for which theoretical analyses are
known, and more powerful SDP methods, for which we seem to have no clue 
how to perform asymptotic analysis.

Not unexpectedly, the hierarchy $\KLP(n, d, \el)$ corresponding to
general (not necessarily linear) codes does not improve on Delsarte's
linear program. Without the extra structure of linearity, the number
of constraints we can add to our LP hierarchy is limited.  We prove
that solutions of $\DLP(n,d)$ (easily) lift to solutions of
$\KLP(n,d,\el)$ with the same value as follows.

\begin{proposition}[Hierarchy Collapse - Informal version of \cref{prop:main:lifting:formal}]\label{prop:lifting}
  For $\el \in \mathbb{N}$, we have
  \begin{align*}
    \val(\KLP(n, d, \el))^{1/\ell} ~=~ \DLP(n,d).
  \end{align*}
\end{proposition}

This contrast between the hierarchies $\KLP_\Lin(n, d, \el)$ and
$\KLP(n, d, \el)$ reinvigorates the question of whether the maximum
sizes of general and linear codes are substantially different or not.

Though we give special attention to the binary case since it may be
the most important one, we prove completeness and lifting results more
generally in the language of association schemes. For example, a
suitable modification of the linear programming hierarchy also
converges to the maximum size of a $D$-code in the Hamming scheme over
any finite field (see \cref{rmk:completeness}); this in particular
covers other types of codes that may be of interest such as
$\epsilon$-balanced codes.

\subsection*{More on Related Work}

Although quantitatively the McEliece~\etal~\cite{MRRW77} upper bound
on $R_2(\delta)$ has not improved, our qualitative understanding of
this upper bound is now substantially better. Friedman and
Tillich~\cite{FT05} designed generalized Alon--Boppana theorems in
order to bound the size of linear binary codes. Inspired
by~\cite{FT05}, Navon and Samorodnitsky~\cite{NS05,NavonS09} rederived
the McEliece~\etal bound on $R_2(\delta)$ for general codes using a
more intuitive proof based on Fourier analysis. Despite a seemingly
different language, the proof in~\cite{NS05} also yields feasible
solutions to the dual of Delsarte's LP as in MRRW. More recently,
Samorodnitsky~\cite{Samorodnitsky2021proof} gave yet a new
interpretation of the McEliece~\etal upper bound and conjectured
interesting hypercontractivity inequalities towards improving the
upper bound on $R_2(\delta)$.

Schrijver~\cite{Schrijver79} showed that the seemingly artificial
Delsarte LP has the same value\footnote{In fact, by a symmetrization
  of the $\vartheta'$ SDP on $H_{n,d}$ using its graph automorphisms,
  one obtains $\DLP(n,d)$ exactly, see \cref{sec:sos}.} as the \Lovasz $\vartheta'$
relaxation for $\alpha(H_{n,d})$, which is also essentially the
degree-2 Sum-of-Squares/Lasserre relaxation of $\alpha(H_{n,d})$ (with
additional non-negativity constraints on the entries of the matrix).
Schrijver showed that this holds generally for commutative association
schemes, a connection that allows us to also express $\KLP$ as
$\vartheta'$ of a certain graph.

A line of work (similar in motivation to the current work) is to strengthen a convex relaxation of
$A_2(n,d)$. In Delsarte's approach, only the distance between pairs of points is taken into account in the
optimization. For this reason, Delsarte's approach is classified as a $2$-point bound~\cite{V19}. Nonetheless,
there is no reason to restrict oneself to just $2$-point interactions.  Schrijver~\cite{Schrijver05}
constructed a family of semi-definite programs (SDPs) for $A_2(n,d)$ designed to take into account the
$3$-point interactions. Extending Schrijver's result to a $4$-point interaction bound, Gijswijt, Mittelmann
and~Schrijver~\cite{GMS12} gave another tighter family of SDPs for $A_2(n,d)$ (they also give a description of their
hierarchy for arbitrary $\el)$. A complete SDP hierarchy for $\alpha(H_{n,d})$ is the Sum-of-Squares/Lasserre
hierarchy, which was proposed for code upper bounds by Laurent~\cite{L07}, building on de
Klerk~\etal~\cite{dKPS07}.

Since the Sum-of-Squares hierarchy is guaranteed to find the correct
value of $\alpha({H}_{n,d})$ when the level is sufficiently high
(precisely, level $2\alpha({H}_{n,d})$), in principle it would be
enough to analyze this SDP to compute $A_2(n,d)$. Unfortunately,
studying the performance of SDPs on a fixed instance is a notoriously
difficult task.  In particular, the global positive semi-definiteness
constraint is nontrivial.  Unfortunately, no theoretical analysis is
known for ``genuine'' SDP methods even for the simplest of them, the
$3$-point bound of Schrijver~\cite{Schrijver05} mentioned above.

In summary, the state of affairs on upper bounding $A_2(n,d)$ or
$A_2^\Lin(n,d)$ is as follows. On one hand, we have a
thorough theoretical understanding of techniques based on Delsarte's
LP, but if the true value of $A_2(n,d)$ or $A_2^\Lin(n,d)$
is closer to the GV~bound, then these techniques fall short of
providing tight bounds. On the other hand, we have $\ell$-point bounds
from SDP techniques capable of yielding the correct value of
$A_2(n,d)$, but (apparently) no clue how to theoretically analyze them
to bound $R_2(\delta)$ for general codes or linear codes. We hope that
our hierarchy will open a new angle of attack on this elusive problem
for the important class of \emph{linear} binary codes.

\subsection{Outline of the Paper}

\cref{sec:prelim} contains some notation and basic facts.

\cref{sec:binary_warmup} shows the construction of the LP hierarchy for the binary code case.  In this
section, we introduce the notion of an \emph{$\el$-configuration}, which roughly capture the Hamming weights
of all words in the subspace spanned by the $\el$ points. In analogy with the usual the Delsarte LP, we then
analyze statistics of codes called \emph{$\ell$-configuration profiles}, which capture the number of
$\ell$-tuples in each possible $\el$-configuration. In the rest of the section we construct higher-order
Krawtchouk polynomials, show MacWilliams identities, define the LP hierarchy and prove that its restrictions
can be computed in $O(n^{2^{\ell+1}-2})$ time (for $\ell\in\NN_+$ fixed).

\cref{sec:sos} shows how the LP hierarchy admits several other interpretations. The LP hierarchy is a
\emph{symmetrization} of an exponential-size hierarchy, which has a natural interpretation either as checking
non-negativity of Fourier coefficients of the code, or as $\vartheta'(G)$ for a large graph $G$.

In \cref{sec:association}, we study our construction in more generality through the theory of
\emph{association schemes}.  Our construction can be seen as adding extra constraints to the $\el$-fold tensor
product of the Delsarte LP. More specifically, the underlying association scheme is a \emph{refinement} of the
$\el$-fold tensor product scheme in which ``semantic'' constraints can be added due to linearity of the code
in the original translation scheme. We study this type of refinement, giving conditions under which it is
still a bona fide translation scheme. The other sections may be read mostly independently of this section.

In \cref{sec:main-results}, we show the main results: that the LP hierarchy is complete for linear codes, and
no better than the Delsarte LP in the general (not necessarily linear) case.

We conclude in \cref{sec:concl} with some open problems.


\section{Preliminaries}
\label{sec:prelim}

A binary code $C$ of block length $n$ is a subset of $\F_2^n$. For a word $x\in\FF_2^n$, we denote by
$\abs{x}\coloneqq\lvert\{i\in[n]\mid x_i\neq 0\}\rvert$ its \emph{Hamming weight}. Given two words $x,y \in
\F_2^n$, we denote by $\Delta(x,y) \coloneqq\lvert x - y\rvert$ their \emph{Hamming distance}. The
\emph{(minimum) distance} of $C$ is defined by $\Delta(C) \coloneqq \min \{\Delta(x,y) \mid x,y \in C\land x
\ne y\}$. The \emph{rate} of $C$ is defined by $r(C) \coloneqq \log_2(\abs{C})/n$. The maximum size of a code
of blocklength $n$ and minimum distance at least $d$ is defined as
\begin{align*}
  A_2(n,d) \coloneqq \max \{ \abs{C} ~\vert~ C \subseteq \F_2^n, \Delta(C) \ge d\}.
\end{align*}
We denote the \emph{asymptotic rate} of codes of relative distance at least $\delta$ and alphabet size $q$ as
\begin{align*}
  R_2(\delta) \coloneqq \limsup_{n\to \infty} \frac{1}{n} \log_2\left(A_2(n,\lfloor \delta n \rfloor)\right).
\end{align*}
We define $A_2^\Lin(n,d)$ and $R_2^\Lin(\delta)$ for linear codes in an analogous way, by further requiring
the code $C$ to be \emph{linear} (i.e., an $\FF_2$-linear subspace of $\F_2^n$).

Note that a code of distance at least $d$ can alternatively be viewed as an independent set in the
\emph{Hamming cube graph of distance less than $d$}, $H_{n,d}$, whose vertex set is $V(H_{n,d})\coloneqq
\FF_2^n$ and whose edge set is $E(H_{n,d})\coloneqq \{\{x,y\}\in\binom{\FF_2^n}{2} \mid \Delta(x,y)\leq
d-1\}$.

Let $f,g : \F_2^n \to \mathbb{R}$. We denote by $\ip{f}{g}\coloneqq\ExpOp_{x\unif\F_2^n}[f(x)g(x)]$ the
\emph{inner product} of $f$ and $g$ under the uniform measure and we denote by $f * g$ their convolution given
by $(f*g)(x) \coloneqq \ExpOp_{y \unif \F_2^n} [f(y)\cdot g(x-y)]$ ($x\in\FF_2^n$). The \emph{Fourier
  transform} $\widehat{f}$ of $f$ is given by $\widehat{f}(x) \coloneqq \ip{f}{\chi_x} = \ExpOp_{y \unif
  \F_2^n} [f(y) \cdot \chi_x(y)]$, where $\chi_x(y)\coloneqq (-1)^{\ip{x}{y}}$. The (simple) Plancherel
identity will be used in our computations.
\begin{fact}[Plancherel]\label{fact:plancherel}
  Let $f,g \colon \F_2^n \to \mathbb{R}$. Then $\ip{f}{g} = \sum_{x \in \F_2^n} \widehat{f}(x) \cdot
  \widehat{g}(x)$.
\end{fact}

Given a linear code $C\subseteq\FF_2^n$, the \emph{dual code} of $C$ is defined as $C^\perp \coloneqq
\{x\in\FF_2^n \mid \forall y\in C,\chi_x(y) = 1\}$.  The Fourier transform of the indicator of a linear code
maps it to a multiple of the indicator of its dual code in the following way.
\begin{fact}\label{fact:fourier_dual_linear_code}
  If $C \subseteq \F_2^n$ is a linear code and $\One_C$ is its indicator function, then $\widehat{\One_C} =
  \lvert C\rvert\cdot \One_{C^{\perp}}/2^n = \One_{C^\perp} / \lvert C^\perp\rvert$.
\end{fact}


\section{Krawtchouk Hierarchies for Binary Codes}\label{sec:binary_warmup}

In this section, we describe the LP hierarchy for the standard case of binary codes. We opt for an ad hoc
derivation from boolean Fourier analysis to show how the higher-order Krawtchouk polynomials nicely
parallel the usual Krawtchouk polynomials. Any omitted proofs in this section can be found in
\cref{sec:binaryproofs}. In \cref{sec:association}, we will generalize the construction using the language of
association schemes.

\subsection{Higher-order Krawtchouk polynomials}
\label{subsec:highkrawtchouk}

As we alluded to previously, we want to incorporate $\ell$-point interactions in our optimization problem for
$A_2(n,d)$ similar in spirit to the Sum-of-Squares semi-definite programming hierarchy for the independence
number of a graph but in the simpler setting of \emph{linear} programming. To accomplish this goal, we measure
the profile of ``configurations'' of $\ell$-tuples of codewords from the code.

We start with the definition of symmetric difference configurations. In plain English, the symmetric
difference configuration of an $\ell$-tuple $(z_1,\ldots,z_\ell)\in(\FF_2^n)^\ell$ of words captures all
information of $(z_1,\ldots,z_\ell)$ corresponding to Hamming weights of linear combinations of the words.

\begin{definition}\label{def:configuration}
  The \emph{symmetric difference configuration} of the $\ell$-tuple $(z_1,\ldots,z_\ell)\in(\FF_2^n)^\ell$ is
  the function $\config_{n,\ell}^\Delta(z_1,\ldots,z_\ell)\colon 2^{[\ell]}\to\RR$ defined by
  \begin{align*}
    \config_{n,\ell}^\Delta(z_1,\ldots,z_\ell)(J)
    & \coloneqq
    \abs{\sum_{j\in J} z_j},
  \end{align*}
  for every $J\subseteq[\ell]$, that is, the value of the function at $J\subseteq[\ell]$ is the Hamming weight
  of the linear combination $\sum_{j\in J} z_j$.

  By viewing $\config_{n,\ell}^\Delta$ as a function $(\FF_2^n)^\ell\to\RR^{2^{[\ell]}}$ (i.e., a function
  from the space of $\ell$-tuples of words to the space of functions $2^{[\ell]}\to\RR$), the set of (valid)
  symmetric difference configurations of $\ell$-tuples of elements of $\FF_2^n$ is captured by its image
  $\im(\config_{n,\ell}^\Delta)$.

  Given a symmetric difference configuration $g\in\im(\config_{n,\ell}^\Delta)$, we will also abuse notation
  and write $(z_1,\ldots,z_\ell)\in g$ to mean that $(z_1,\ldots,z_\ell)\in(\FF_2^n)^\ell$ has configuration
  $\config_{n,\ell}^\Delta(z_1,\ldots,z_\ell) = g$. In other words, this abuse of notation consists of
  thinking of a configuration as the set of all $\ell$-tuples of words that have this configuration
  (see also \cref{lem:configorbits} below). We also let $\lvert g\rvert$ be the size of this set, i.e., the
  number of $\ell$-tuples whose configuration is $g$.

  The \emph{trivial symmetric difference configuration} is the constant $0$ function (denoted by $0$), which
  is the symmetric configuration of the tuple $(0,\ldots,0)\in(\FF_2^n)^\ell$.
\end{definition}

\begin{remark}
A configuration measures the Hamming weights of vectors
in the subspace of $\F_2^n$ spanned by $z_1, \dots, z_\el$. However, \cref{def:configuration}
depends on the choice of basis for the subspace.
With more technical difficulty one 
can define configurations in a basis-independent way; see the 
discussion at the end of \cref{sec:sos:symmetrization}.
\end{remark}

Even though the space $(\FF_2^n)^\ell$ has exponential size in $n$ (for a fixed $\ell$), the next lemma says
that the number of configurations is polynomial in $n$ (for a fixed $\ell$).

\begin{restatable}{lemma}{configcount}\label{lem:configcount}
  We have
  \begin{align*}
    \abs{\im(\config_{n,\ell}^\Delta)} = \binom{n + 2^\ell - 1}{2^\ell - 1}.
  \end{align*}
\end{restatable}

The next lemma provides an alternative way of viewing configurations: for each symmetric difference
configuration $g\in\im(\config_{n,\ell}^\Delta)$, the set of $\ell$-tuples with a certain configuration $g$ is
precisely one of the orbits of the natural (diagonal) right action of the symmetric group $S_n$ on $n$ points
on $(\FF_2^n)^\ell$.

\begin{restatable}{lemma}{configorbits}\label{lem:configorbits}
  Let $n,\ell\in\NN_+$ and consider the natural (diagonal) right action of $S_n$ on $(\FF_2^n)^\ell$ given by
  $(x_1,\ldots,x_\ell)\cdot\sigma\coloneqq (y_1,\ldots,y_\ell)$, where $(y_j)_i\coloneqq (x_j)_{\sigma(i)}$
  ($(x_1,\ldots,x_\ell),(y_1,\ldots,y_\ell)\in(\FF_2^n)^\ell$, $\sigma\in S_n$, $j\in[\ell]$, $i\in[n]$).

  The following are equivalent for $(x_1,\ldots,x_\ell),(y_1,\ldots,y_\ell)\in(\FF_2^n)^\ell$.
  \begin{enumerate}
  \item $(x_1,\ldots,x_\ell)$ and $(y_1,\ldots,y_\ell)$ are in the same $S_n$-orbit.
    \label{lem:configorbits:orbit}
  \item $\config_{n,\ell}^\Delta(x_1,\ldots,x_\ell) = \config_{n,\ell}^\Delta(y_1,\ldots,y_\ell)$.
    \label{lem:configorbits:SD}
  \end{enumerate}
\end{restatable}

Similarly to the weight profile of a code, we can define a higher-order configuration profile.

\begin{definition}
  The \emph{$\ell$-configuration profile} of a code $C\subseteq\FF_2^n$ is the sequence
  $(a^C_g)_{g \in\im(\config_{n,\ell}^\Delta)}$ defined by
  \begin{align*}
    a^C_g
    & \coloneqq
    \frac{1}{\lvert C\rvert^\ell}
    \Bigl\lvert\Bigl\{
    \bigl((x_1,\ldots,x_\ell),(y_1,\ldots,y_\ell)\bigr)\in C^\ell\times C^\ell
    \;\Big\vert\;
    \config_{n,\ell}^\Delta(x_1-y_1,\ldots,x_\ell-y_\ell) = g
    \Bigr\}\Bigr\rvert.
  \end{align*}
\end{definition}

\begin{remark}\label{rmk:linearprofile}
  Note that if $C$ is linear, $a^C_g$ can alternatively be computed as
  \begin{align*}
    a^C_g
    =
    \lvert\{(z_1,\ldots,z_\ell) \in C^\ell \mid \config_{n,\ell}^\Delta(z_1,\ldots,z_\ell)=g\}\rvert.
  \end{align*}  
\end{remark}

\medskip

Recall that the (usual) Krawtchouk polynomial $K_i$ of degree $i$ is defined by
\begin{align*}
  K_i(t)
  & \coloneqq
  2^n \ExpOp_{x \in \mathbb{F}_2^n}[\one_{W_i}(x)\cdot \chi_y(x)]
  \\
  & =
  \sum_{x\in W_i} \chi_y(x),
\end{align*}
where $W_i\subseteq\FF_2^n$ is the set of all words of Hamming weight $i$, $\one_{W_i} \colon \mathbb{F}_2^n
\to \{0,1\}$ is its indicator function and $y\in W_t$ is any element with of Hamming weight $t$.

\begin{definition}[Higher-order Krawtchouk]
  Let $h \in\im(\config_{n,\ell}^\Delta)$ be a symmetric difference configuration. The \emph{higher-order Krawtchouk
    polynomial} indexed by $h$ is the function $K_h\colon\im(\config_{n,\ell}^\Delta)\to\RR$ defined by
  \begin{equation}\label{eq:highkrawtchouk}
    \begin{aligned}
      K_h(g)
      & \coloneqq
      2^{\ell n}
      \ExpOp_{(y_1,\ldots,y_\ell) \in (\mathbb{F}_2^n)^\ell}
      \left[\One_h(y_1,\ldots,y_\ell)\cdot \prod_{j=1}^\ell \chi_{x_j}(y_j)\right]
      \\
      & =
      \sum_{(y_1,\ldots,y_\ell)\in h} \prod_{j=1}^\ell \chi_{x_j}(y_j),
    \end{aligned}
  \end{equation}
  for every symmetric difference configuration $g\in\im(\config_{n,\ell}^\Delta)$, where
  $(x_1,\ldots,x_\ell)\in g$ is any $\ell$-tuple of words with symmetric difference configuration $g$ and
  $\One_h$ is the indicator function of the set of $\ell$-tuples whose symmetric difference configuration is
  $h$ (\cref{lem:explicitkrawtchouk} shows this is well-defined).
\end{definition}

\begin{remark}\label{rmk:highkrawtchouk}
  Another way to see $K_h$ above is as the unique function (see \cref{lem:explicitkrawtchouk} below) such that
  \begin{align*}
    \widehat{\One_h}
    & =
    \frac{K_h\comp\config_{n,\ell}^\Delta}{2^{n\ell}}.
  \end{align*}
\end{remark}

Note that when $\ell=1$, a symmetric difference configuration $\config_{n,1}^\Delta(x)$ of a word
$x\in\FF_2^n$ only tracks the Hamming weight $\config_{n,1}^\Delta(x)(\{1\}) = \lvert x\rvert$ of $x$ (as
$\config_{n,1}^\Delta(x)(\varnothing)$ is always equal to $0$) thus we recover the univariate Krawtchouk
polynomials.

For explicit computations of the higher-order Krawtchouk polynomials, the formula~\eqref{eq:highkrawtchouk} is
quite inconvenient as it involves a sum of $2^{n\ell}$ terms. We will provide an alternative formula in
\cref{subsec:prophighkrawtchouk}.

\subsection{Higher-order MacWilliams Identities and Inequalities}

In this section, we show a higher-order analogue of MacWilliams identities and inequalities using only basic
Fourier analysis. Later we are going to define a suitable family of association schemes from which MacWilliams
identities and inequalities follow from the general theory of association schemes of
Delsarte~\cite{delsarte1973algebraic,DL98}.

The MacWilliams identities show a surprising combinatorial fact: the weight profile of the dual $C^\perp$ of a
linear code $C\subseteq\FF_2^n$ is completely determined by the weight profile of $C$. The higher-order
MacWilliams identities generalize this fact to $\ell$-configuration profiles.

\begin{lemma}[Higher-order MacWilliams identities]
  Let $C\in\FF_2^n$ be a linear code and let $h\in\im(\config_{n,\ell}^\Delta)$ be a symmetric difference
  configuration. Then
  \begin{align*}
    a^{C^\perp}_h
    & =
    \frac{1}{\lvert C\rvert^\ell}\sum_{g\in\im(\config_{n,\ell}^\Delta)} K_h(g)\cdot a^C_g.
  \end{align*}
\end{lemma}

\begin{proof}
  By \cref{rmk:linearprofile}, we have
  \begin{align*}
    a^{C^\perp}_h
    & =
    2^{n\ell}\ip{\One_h}{\One_{(C^\perp)^\ell}}
    =
    2^{n\ell}
    \sum_{x\in(\FF_2^n)^\ell} \widehat{\One_h}(x)\widehat{\One_{(C^\perp)^\ell}}(x)
    \\
    & =
    \frac{1}{\lvert C\rvert^\ell}
    \sum_{x\in(\FF_2^n)^\ell} K_h(\config_{n,\ell}^\Delta(x))\cdot\One_{C^\ell}(x)
    =
    \frac{1}{\lvert C\rvert^\ell}\sum_{g\in\im(\config_{n,\ell}^\Delta)} K_h(g)\cdot a^C_g,
  \end{align*}
  where the second equality follows from \cref{fact:plancherel} and the third equality follows from
  \cref{rmk:highkrawtchouk,fact:fourier_dual_linear_code}.
\end{proof}

Just as the usual MacWilliams inequalities hold for arbitrary codes, we can prove that the same transformation at
least yields non-negative numbers.

\begin{lemma}[Higher-order MacWilliams inequalities]\label{lem:highMacWilliamsineq}
  Let $C\in\FF_2^n$ be an arbitrary code. For $h\in\im(\config_{n,\ell}^\Delta)$, we have
  \begin{align*}
    \sum_{g\in\im(\config_{n,\ell}^\Delta)} K_h(g)\cdot a^C_g & \geq 0.
  \end{align*}
\end{lemma}

\begin{proof}
  Note that \cref{rmk:highkrawtchouk} implies that the Fourier transform of $K_h\comp\config_{n,\ell}^\Delta$
  is $\One_h$, so we have
  \begin{align*}
    \sum_{g\in\im(\config_{n,\ell}^\Delta)} K_h(g)\cdot a^C_g
    & =
    \sum_{x,y\in(\FF_2^n)^\ell}
    K_h(\config_{n,\ell}^\Delta(x-y))\cdot \One_{C^\ell}(x)\One_{C^\ell}(y)
    \\
    & =
    2^{2n\ell}
    \ip{K_h\comp\config_{n,\ell}^\Delta}{\One_{C^\ell} * \One_{C^\ell}}
    \\
    & =
    2^{2n\ell}
    \sum_{x\in(\FF_2^n)^\ell} \One_h(x)\widehat{\One_{C^\ell}}^2
    \geq
    0,
  \end{align*}
  where the third equality follows from \cref{fact:plancherel}.
\end{proof}

\subsection{Higher-order Delsarte's Linear Programs}

Now we have all the elements to define a hierarchy of linear programs for $A_2(n,d)$ parameterized by the size
of the interactions $\ell\in\NN_+$ in analogy to $\DLP(n,d)$.

\begin{definition}\label{def:holp}
  For $n,\ell\in\NN_+$ and $d\in\{0,1,\ldots,n\}$, we let $\KLP(n,d,\ell)$ be the following linear program.
  \begin{align*}
    \max \quad
    & \sum_{g\in\im(\config_{n,\ell}^\Delta)} a_g
    \\
    \text{s.t.} \quad
    & a_0 = 1
    & &
    & & (\text{Normalization})
    \\
    & a_g = 0
    & & \forall g\in\forbconfig(n,d,\ell)
    & &  (\text{Distance constraints})
    \\
    & \sum_{g\in\im(\config_{n,\ell}^\Delta)} K_h(g)\cdot a_g \geq 0
    & & \forall h\in\im(\config_{n,\ell}^\Delta)
    & & (\text{MacWilliams inequalities})
    \\
    & a_g \geq 0
    & & \forall g\in\im(\config_{n,\ell}^\Delta)
    & & (\text{Non-negativity}),
  \end{align*}
  where the variables are $(a_g)_{g\in\im(\config_{n,\ell}^\Delta)}$ and
  \begin{align*}
    \forbconfig(n,d,\ell)
    & \coloneqq
    \{g\in\im(\config_{n,\ell}^\Delta)
    \mid
    \exists j\in[\ell], g(\{j\})\in\{1,\ldots,d-1\}
    \}.
  \end{align*}

  We also define $\KLP_\Lin(n,d,\ell)$ as the linear program obtained by replacing the set $\forbconfig(n,d,\ell)$
  with
  \begin{align*}
    \forbconfig_\Lin(n,d,\ell)
    & \coloneqq
    \{g\in\im(\config_{n,\ell}^\Delta)
    \mid
    \exists J\subseteq[\ell], g(J)\in\{1,\ldots,d-1\}\}.
  \end{align*}
\end{definition}

\begin{proposition}
  The linear programs $\KLP(n,d,\ell)$ and $\KLP_\Lin(n,d,\ell)$ are sound, that is, we have
  \begin{align*}
    \val(\KLP(n,d,\ell))^{1/\ell} & \geq A_2(n,d),\\
    \val(\KLP_\Lin(n,d,\ell))^{1/\ell} & \geq A^\Lin_2(n,d).
  \end{align*}
\end{proposition}

\begin{proof}
  Recall that for $C\subseteq\FF_2^n$, we have
  \begin{align*}
    a^C_g
    & \coloneqq
    \frac{1}{\lvert C\rvert^\ell}
    \lvert\{
    (x_1,\ldots,x_\ell),(y_1,\ldots,y_\ell)\in C^\ell\times C^\ell
    \mid
    \config_{n,\ell}^\Delta(x_1-y_1,\ldots,x_\ell-y_\ell) = g
    \}\rvert.
  \end{align*}

  If $C$ is an arbitrary code of distance at least $d$, then \cref{lem:highMacWilliamsineq} implies that the
  $\ell$-configuration profile $a^C$ satisfies the MacWilliams inequalities. On the other hand, if
  $g\in\forbconfig(n,d,\ell)$, that is, we have $g(\{j\})\in\{1,\ldots,d-1\}$ for some $j\in[\ell]$, then
  clearly no pair of $\ell$-tuples of codewords $(x_1,\ldots,x_\ell),(y_1,\ldots,y_\ell)\in C^\ell$ can satisfy
  $\config_{n,\ell}^\Delta(x_1-y_1,\ldots,x_\ell-y_\ell) = g$ as it would imply $\lvert x_j - y_j\rvert =
  g(\{j\})\in\{1,\ldots,d-1\}$, thus the distance constraints are also satisfied.

  All other restrictions follow trivially from the definition of $a^C$, thus $a^C$ is a feasible solution of
  $\KLP(n,d,\ell)$. Since the objective value of $a^C$ is $\sum_{g\in\im(\config_{n,\ell}^\Delta)} a^C_g =
  \lvert C\rvert^\ell$, it follows that $\val(\KLP(n,d,\ell))^{1/\ell}\geq A_2(n,d)$.

  \medskip

  If we further assume that $C$ is linear and $g\in\forbconfig_\Lin(n,d,\ell)$ is such that $g(J)\in [d-1]$
  for some $J\subseteq[\ell]$, then no tuple $(z_1,\ldots,z_\ell)\in C^\ell$ can satisfy
  $\config_{n,\ell}^\Delta(z_1,\ldots,z_\ell) = g$ as it would imply $\lvert\sum_{j\in J} z_j\rvert =
  g(J)\in\{1,\ldots,d-1\}$. By \cref{rmk:linearprofile}, we get $a^C_g = 0$, so $a^C$ is also a feasible
  solution of $\KLP_\Lin(n,d,\ell)$ and thus $\val(\KLP_\Lin(n,d,\ell))^{1/\ell}\geq A^\Lin_2(n,d)$.
\end{proof}

\subsection{Properties of higher-order Krawtchouk polynomials}
\label{subsec:prophighkrawtchouk}

In this section, we explore more properties of the higher-order Krawtchouk polynomials in order to show that
the objective and restrictions of the linear programs $\KLP(n,d,\ell)$ and $\KLP_\Lin(n,d,\ell)$ can
be algorithmically computed in $O(n^{2^{\ell+1}-2})$ time for a fixed $\ell\in\NN_+$ (see
\cref{prop:complexity}).

Even though symmetric difference configurations are more natural from the point of view of linear codes, for
computations and properties with the higher-order Krawtchouk polynomials, it is more convenient to work with
Venn diagram configurations defined below. In plain English, each word $z\in\FF_2^n$ induces a partition of
$[n]$ into its support $\supp(z)\coloneqq\{i\in[n]\mid z_i\neq 0\}$ and its complement $[n]\setminus\supp(z)$;
the Venn diagram configuration of a tuple $(z_1,\ldots,z_\ell)\in(\FF_2^n)^\ell$ then encodes the information
about the sizes of each of the cells of the Venn diagram of the coarsest common refinement of the partitions
induced by the $z_i$.

\begin{definition}
  The \emph{Venn diagram configuration} of the $\ell$-tuple $(z_1,\ldots,z_\ell)\in(\FF_2^n)^\ell$ is the
  function $\config_{n,\ell}^V(z_1,\ldots,z_\ell)\colon 2^{[\ell]}\to\RR$ defined by
  \begin{align*}
    \config_{n,\ell}^V(z_1,\ldots,z_\ell)(J)
    & \coloneqq
    \left\lvert
    \bigcap_{j\in J}\supp(z_i)\cap\bigcap_{j\in[\ell]\setminus J}([n]\setminus\supp(z_j))
    \right\rvert
    \\
    & =
    \Bigl\lvert\Bigl\{
    i\in[n]
    \;\Big\vert\;
    \{j\in[\ell] \mid (z_j)_i = 1\} = J
    \Bigr\}\Bigr\rvert,
  \end{align*}
  for every $J\subseteq[\ell]$.

  By viewing $\config_{n,\ell}^V$ as a function $(\FF_2^n)^\ell\to\RR^{2^{[\ell]}}$, the set of (valid) Venn
  diagram configurations of $\ell$-tuples of elements of $\FF_2^n$ is $\im(\config_{n,\ell}^V)$.
\end{definition}

The next lemma gives an easy description of the set of Venn diagram configurations of $\ell$-tuples of
elements of $\FF_2^n$ as the set of all functions $2^{[\ell]}\to\RR$ whose values are non-negative integers
that add up to $n$. Combining it with \cref{lem:configconversion} below gives an explicit description of the
set of symmetric difference configurations.

\begin{restatable}{lemma}{validVDconfigs}\label{lem:validVDconfigs}
  For every $n,\ell\in\NN_+$, we have
  \begin{align}\label{eq:validVDconfigs}
    \im(\config_{n,\ell}^V)
    & =
    \left\{
    g\colon 2^{[\ell]}\to\RR
    \;\middle\vert\;
    \sum_{J\subseteq[\ell]} g(J) = n\land
    \forall J\subseteq [\ell], g(J)\in\NN
    \right\}.
  \end{align}
\end{restatable}

The next lemma provides a pair of linear transformations that transform a symmetric difference configuration
into a Venn diagram configuration and vice-versa.

\begin{restatable}{lemma}{configconversion}\label{lem:configconversion}
  Let $n,\ell\in\NN_+$, let
  \begin{align*}
    S_{n,\ell} & \df \left\{g\in\RR^{2^{[\ell]}} \;\middle\vert\; \sum_{J\subseteq[\ell]} g(J) = n\right\}, &
    Z_{n,\ell} & \df \{g\in\RR^{2^{[\ell]}} \mid g(\varnothing) = 0\}
  \end{align*}
  and let $V_{n,\ell}\colon Z_{n,\ell}\to S_{n,\ell}$ and $D_{n,\ell}\colon S_{n,\ell}\to Z_{n,\ell}$ be given by
  \begin{align}
    D_{n,\ell}(g)(J)
    & \df
    \sum_{\substack{T\subseteq[\ell]\\\lvert T\cap J\rvert\text{ odd}}} g(T),
    \label{eq:D}
    \\
    V_{n,\ell}(g)(J)
    & \df
    n\cdot\One[J =\varnothing] + 2^{1-\ell}\sum_{T\subseteq[\ell]} (-1)^{\lvert T\cap J\rvert-1} g(T),
    \label{eq:V}
  \end{align}
  for every $J\subseteq[\ell]$.

  Then $V_{n,\ell}$ and $D_{n,\ell}$ are inverses of each other and $\config_{n,\ell}^\Delta =
  D_{n,\ell}\comp\config_{n,\ell}^V$ and $\config_{n,\ell}^V = V_{n,\ell}\comp\config_{n,\ell}^\Delta$.
\end{restatable}

Making use of Venn diagram configurations, we can also easily compute the number of $\ell$-tuples with a given
configuration as a multinomial.

\begin{restatable}{lemma}{configsize}\label{lem:configsize}
  For a symmetric difference configuration $g\in\im(\config_{n,\ell}^\Delta)$, we have
  \begin{align*}
    \lvert g\rvert
    & =
    K_g(0)
    =
    \binom{n}{V_{n,\ell}(g)}
    =
    \frac{n!}{\prod_{J\subseteq[\ell]} V_{n,\ell}(g)(J)!},
  \end{align*}
  where $V_{n,\ell}$ is given by~\eqref{eq:V}.
\end{restatable}

\medskip

The following lemma says that, similarly to the univariate case, the higher-order Krawtchouk polynomials are
orthogonal with respect to the natural discrete measure on symmetric configurations in which each
$g\in\im(\config_{n,\ell}^\Delta)$ has measure $\lvert g\rvert = \binom{n}{V_{n,\ell}(g)}$ (see
\cref{lem:configsize}), i.e., the number of $\ell$-tuples with configuration $g$.

\begin{restatable}{lemma}{orthogonality}[Orthogonality]
  For $n,\ell\in\NN_+$ and $h,h'\in\im(\config_{n,\ell}^\Delta)$, we have
  \begin{align*}
    \sum_{g\in\im(\config_{n,\ell}^\Delta)}
    \lvert g\rvert\cdot K_h(g)\cdot K_{h'}(g)
    & =
    2^{\ell n}\cdot\lvert h\rvert\cdot\One[h = h'].
  \end{align*}
\end{restatable}

Also similarly to the univariate case, the higher-order Krawtchouk polynomials satisfy the following
reflection property.

\begin{restatable}{lemma}{reflection}[Reflection]\label{lem:reflection}
  For $n,\ell\in\NN_+$ and $g,h\in\im(\config_{n,\ell}^\Delta)$, we have
  \begin{align*}
    \frac{K_h(g)}{\lvert h\rvert}
    =
    \frac{K_g(h)}{\lvert g\rvert}.
  \end{align*}
\end{restatable}

The next lemma provides an alternative formula for the higher-order Krawtchouk polynomial in which the sum
involves only $O(n^{2^{2\ell}})$ terms (as opposed to the $2^{\ell n}$ terms in~\eqref{eq:highkrawtchouk}).
\begin{restatable}{lemma}{explicitkrawtchouk}\label{lem:explicitkrawtchouk}
  For every $n,\ell\in\NN_+$ and every $g,h\in\im(\config_{n,\ell}^\Delta)$, we have
  \begin{align*}
    K_h(g)
    & =
    \sum_{F\in\cF}
    \prod_{J\subseteq[\ell]} \frac{V_{n,\ell}(g)(J)!}{\prod_{K\subseteq[\ell]} F(J,K)!}
    \cdot
    \prod_{j=1}^\ell \prod_{\substack{J,K\subseteq[\ell]\\ j\in J\cap K}} (-1)^{F(J,K)},
  \end{align*}
  where $\cF$ is the set of functions $F\colon 2^{[\ell]}\times 2^{[\ell]}\to\{0,1,\ldots,n\}$ such that
  \begin{align*}
    \forall J\subseteq [\ell], \sum_{K\subseteq[\ell]} F(J,K) = V_{n,\ell}(g)(J),\\
    \forall K\subseteq [\ell], \sum_{J\subseteq[\ell]} F(J,K) = V_{n,\ell}(h)(K),
  \end{align*}
  and $V_{n,\ell}$ is given by~\eqref{eq:V}.
\end{restatable}

The next lemma allows the computation of the Krawtchouk polynomials even faster via dynamic programming.

\begin{restatable}{lemma}{recursive}\label{lem:recursive}
  Let $n,\ell\in\NN_+$ with $n\geq 2$, let $g,h\in\im(\config_{n,\ell}^\Delta)$ be symmetric difference
  configurations and let $J_0\subseteq[\ell]$ be such that $V_{n,\ell}(g)(J_0) > 0$ for $V_{n,\ell}$ given
  by~\eqref{eq:V}. Then
  \begin{align}
    K_h(g)
    & =
    \sum_{\substack{K_0\subseteq[\ell]\\ V_{n,\ell}(h)(K_0) > 0}}
    (-1)^{\lvert J_0\cap K_0\rvert}\cdot K_{h\ominus K_0}(g\ominus J_0),
    \label{eq:recursive1}
    \\
    K_h(g)
    & =
    -
    \sum_{\substack{K_0\subseteq[\ell]\\ V(h)(K_0) > 0\\K_0\neq\varnothing}}
    K_{h\oplus\varnothing\ominus K_0}(g)
    +
    \sum_{\substack{K_0\subseteq[\ell]\\ V(h)(K_0) > 0}}
    (-1)^{\lvert J_0\cap K_0\rvert}\cdot K_{h\oplus\varnothing\ominus K_0}(g\oplus\varnothing\ominus J_0),
    \label{eq:recursive2}
  \end{align}
  where
  \begin{align*}
    h\ominus K_0 & \coloneqq D_{n-1,\ell}(V_{n,\ell}(h) - \One_{\{K_0\}}), &
    g\ominus J_0 & \coloneqq D_{n-1,\ell}(V_{n,\ell}(g) - \One_{\{J_0\}}),
    \\
    h\oplus\varnothing
    & \coloneqq
    D_{n+1,\ell}(V_{n,\ell}(h) + \One_{\{\varnothing\}}),
    &
    g\oplus\varnothing
    & \coloneqq
    D_{n+1,\ell}(V_{n,\ell}(g) + \One_{\{\varnothing\}}),
  \end{align*}
  and $D_{n-1,\ell}$ and $D_{n+1,\ell}$ are given by~\eqref{eq:D}.
\end{restatable}

\begin{proposition}\label{prop:complexity}
  The objective and restrictions of the linear programs $\KLP(n,d,\ell)$ and $\KLP_\Lin(n,d,\ell)$
  can be algorithmically computed in $O(n^{2^{\ell+1}-2})$ time for a fixed $\ell\in\NN_+$.
\end{proposition}

\begin{proof}
  The number of variables and restrictions of these linear programs is the number of configurations at level
  $\ell$, which is $O(n^{2^\ell-1})$ by \cref{lem:configcount}. Furthermore, converting between symmetric
  difference configurations and Venn diagram configurations using \cref{lem:configconversion} can be done in
  time $O(2^\ell)=O(1)$ and using \cref{lem:explicitkrawtchouk} and~\eqref{eq:recursive2} in
  \cref{lem:recursive}, we can compute all values of all Krawtchouk polynomials of level $\ell$ in time
  $O((n^{2^\ell-1})^2) = O(n^{2^{\ell+1}-2})$.
\end{proof}


\section{Unsymmetrized Formulations of the Krawtchouk Hierarchies}\label{sec:sos}

In this section we give other formulations for $\KLP$.  These formulations are \emph{unsymmetrized} versions
of the same hierarchy.  Working with the unsymmetrized hierarchy can be easier, since it avoids the technical
definitions of the Krawtchouk polynomials $K_h(g)$, but computationally the number of variables and
constraints of these hierarchies is huge.

\subsection{The Hierarchy as Checking Non-negativity of Fourier Coefficients}
\label{sec:sos:symmetrization}

The LP hierarchy for linear codes can be simply described as checking non-negativity of products
of Fourier coefficients.
Define the linear programming hierarchy $\FLP_\Lin(n,d,\el)$ with the variables
$a_{x} \;{(x \in (\F_2^n)^\el)}$:
\begin{align*}
  \max \quad
  & \sum_{x \in (\F_2^n)^\el} a_{x}
  \\
  \text{s.t.} \quad
  & a_{0} = 1
  & &
  & & (\text{Normalization})
  \\
  & a_{(x_1, \dots, x_\el)} = 0
  & & \exists w \in \Span(x_1, \dots, x_\el), \abs{w} \in \{1, \dots, d-1\}
  & & (\text{Distance constraints})
  \\
  & \sum_{x \in (\F_2^n)^\el} a_x\chi_\alpha(x) \geq 0
  & & \forall \alpha \in (\F_2^n)^\el
  & & (\text{Fourier coefficients})
  \\
  & a_x \geq 0
  & & \forall x \in (\F_2^n)^\el
  & & (\text{Non-negativity}).
\end{align*}

\begin{proposition}
  For every $n,\ell\in\NN_+$ and $d\in\{0,1,\ldots,n\}$, $\val(\FLP_\Lin(n,d,\el)) \geq
  A^\Lin_2(n,d)^\el$.
\end{proposition}

\begin{proof}
  Given a linear code $C$ with distance $d$, a feasible solution with value $\abs{C}^\el$ is $a_{(x_1, \dots,
    x_\el)} \coloneqq \prod_{i=1}^\el\One[x_i \in C]$.  The Fourier coefficient constraints are satisfied
  because
  \begin{align*}
    \sum_{x \in (\F_2^n)^\el} \prod_{i=1}^\el\One[x_i \in C] \chi_{\alpha_i}(x_i)
    & =
    2^{n\el} \prod_{i=1}^\el \widehat{\one_C}(\alpha_i),
  \end{align*}
  which are nonnegative by \cref{fact:fourier_dual_linear_code}.
\end{proof}

The corresponding hierarchy for non-linear codes $\FLP(n,d,\el)$ is defined over the variables $a_x \;{(x \in
  (\F_2^n)^\el)}$ as:
\begin{align*}
  \max \quad
  & \sum_{x \in (\F_2^n)^\el} a_{x}
  \\
  \text{s.t.} \quad
  & a_{0} = 1
  & &
  & & (\text{Normalization})
  \\
  & a_{(x_1, \dots, x_\el)} = 0
  & & \exists i \in [\el], \abs{x_i} \in \{1,\dots, d-1\}
  & & (\text{Distance constraints})
  \\
  & \sum_{x \in (\F_2^n)^\el} a_x\chi_\alpha(x) \geq 0
  & & \forall \alpha \in (\F_2^n)^\el
  & & (\text{Fourier coefficients})
  \\
  & a_x \geq 0
  & & \forall x \in (\F_2^n)^\el
  & & (\text{Non-negativity}).
\end{align*}

It turns out that $\KLP$ is a symmetrization of $\FLP$ (and likewise for the programs $\KLP_\Lin$ and
$\FLP_\Lin$). We will briefly describe the technique of symmetrizing convex programs, which is also described
in the survey article by Vallentin~\cite{V19}.  The proof that $\KLP$ and $\FLP$ are equivalent continues at
\cref{prop:flp_klp}.

The technique exploits the fact that convex relaxations for the independence number $\alpha(H_{n,d})$ of the
Hamming cube graph $H_{n,d}$ of distance less than $d$ are highly symmetric, that is, programs that are
invariant under large permutation groups as defined below.

\begin{definition}[Program invariance]\label{def:invariance}
  Let $\calP$ be a linear program with variables $(a_x)_{x \in X}$ for some set $X$. We say that $\calP$ is
  invariant under a permutation $\sigma$ of $X$ if for all feasible solutions $(a_x)$, the point
  $a\cdot\sigma$ defined by $(a\cdot\sigma)_x \coloneqq a_{\sigma(x)}$ is also feasible, and the objective
  value is the same.

  Similarly, a semi-definite program $\calP$ with variable $M\in \R^{X \times X}$ is invariant under $\sigma$
  if for all feasible $M$, the matrix $M\cdot\sigma$ defined by $(M\cdot\sigma)[x,y] \defeq M[\sigma(x),
    \sigma(y)]$ is also feasible, and the objective value is the same.

  The group of permutations of $X$ under which $\calP$ is invariant is called the \emph{automorphism group} of
  $\calP$ and is denoted $\Aut(\calP)$.
\end{definition}

If the input of a program $\calP$ is a graph $G$ and the program only depends on the isomorphism class of $G$,
then the program is invariant under the automorphism group $\Aut(G)$ of the graph $G$. For convex relaxations
such as the Lov{\'a}sz $\vartheta$-function or the Sum-of-Squares hierarchy, the variables of the program are
indexed by tuples of vertices from $G$, and thus a case of interest is when $\Aut(G)$ acts diagonally on
tuples of vertices.

By symmetrizing solutions, i.e., by averaging the values of the variables over the automorphism group
$\Aut(\calP)$, we may assume that the solution has the same symmetry:
\begin{fact}\label{fact:invariantsolution}
  For any $H\subseteq\Aut(\calP)$, the value $\val(\calP)$ equals the value of $\calP$ with the additional
  constraints $\forall\sigma\in H,\forall x \in X, a_x = a_{\sigma(x)}$ (or $\forall\sigma\in H,\forall x, y
  \in X, M[x,y] = M[\sigma(x), \sigma(y)]$ for an SDP).
\end{fact}
A symmetrized solution is constant on each orbit of the group action on $X$ or $X^2$.
Therefore, the ``effective'' number of variables in the convex program is only the number
of orbits, which may be significantly smaller than even $\abs{V(G)}$.

For example, the graph $H_{n,d}$ has a large symmetry group:
\begin{fact}
  For $1 < d < n$, $\Aut(H_{n,d})$ is the hyperoctahedral group, which is the semidirect product
  $\FF_2^n\rtimes S_n$ in which $S_n$ permutes the coordinates and $\FF_2^n$ applies a bit flip.
\end{fact}
Even though the hypercube has size $2^n$ and thus $\lvert V(H_{n,\ell})^\ell\rvert = 2^{n\ell}$, the number of
orbits of the diagonal action of $\Aut(H_{n,d})$ on $\ell$-tuples is only $\poly(n)$ for constant $\ell$. For
example, for $\el = 4$, viewing the hypercube momentarily as $\{-1,+1\}^n$, the orbit of $(x_1, x_2, x_3,
x_4)$ essentially only depends on the angles between the vectors: it is determined by the seven numbers
\begin{align}\label{eq:sevenangles}
  \ip{x_1}{x_2}, &&
  \ip{x_1}{x_3}, &&
  \ip{x_1}{x_4}, &&
  \ip{x_2}{x_3}, &&
  \ip{x_2}{x_4}, &&
  \ip{x_3}{x_4}, &&
  \sum_{i=1}^n x_{1,i} x_{2,i} x_{3,i} x_{4,i}.
\end{align}
Equivalently, it is determined by $\config_{n,\ell}^{\Delta}(x_2-x_1,x_3-x_1,x_4-x_1)$ (see
\cref{lem:configorbits}).

Since each of the numbers in~\eqref{eq:sevenangles} takes at most $n+1$ values, the effective number of
variables in the degree-4 Sum-of-Squares relaxation for $\alpha(H_{n,d})$ is at most $O(n^{7})$. Thus, the
search for an upper bound on an exponential-size object is reduced to a polynomial-size convex program! Of
course, to actually run this in polynomial time, one also needs to show that this polynomial-size convex
program can be computed in polynomial time (which rules out explicitly computing the original program then
taking a quotient).

We use the symmetrization technique to show that $\KLP$ and $\FLP$ are equivalent.
\begin{proposition}\label{prop:flp_klp}
  For every $n,\ell\in\NN_+$ and every $d\in\{0,1\ldots,n\}$, we have
  \begin{align*}
    \val(\FLP(n,d,\el)) & = \val(\KLP(n,d,\el)),\\
    \val(\FLP_\Lin(n,d,\el)) & = \val(\KLP_\Lin(n,d,\el)).
  \end{align*}
\end{proposition}

\begin{proof}
  Recall that the natural right action of $S_n$ on $\FF_2^n$ is given by $(x\cdot\sigma)_i\coloneqq x_{\sigma(i)}$
  ($x\in\FF_2^n$, $\sigma\in S_n$, $i\in[n]$) and consider the diagonal action of $S_n$ on $(\FF_2^n)^\ell$
  given by
  \begin{align*}
    (x_1,\ldots,x_\ell)\cdot\sigma & \coloneqq (x_1\cdot\sigma,\ldots,x_\ell\cdot\sigma)
    \quad ((x_1,\ldots,x_\ell)\in (\FF_2^n)^\ell, \sigma\in S_n).
  \end{align*}
  It is straightforward to check that $\FLP(n,d,\el)$ is invariant under this diagonal action.

  By \cref{fact:invariantsolution} we may consider only solution to the LP that are symmetrized over $S_n$,
  that is, we have $a_x = a_y$ for each $x,y \in (\F_2^n)^\el$ in the same orbit of the $S_n$-action.

  Recall from \cref{lem:configorbits} that $\ell$-tuples of words are in the same $S_n$-orbit if and only if
  they have the same symmetric difference configuration. We claim that the correspondence between the program
  $\KLP$ with variables $(a'_g)_{g\in\im(\config_{n,\ell}^\Delta)}$ and $\FLP$ is
  \begin{align*}
    a'_g = \lvert g\rvert \cdot a_{x_1, \dots, x_\el} \qquad \text{for any }(x_1, \dots, x_\el) \in g.
  \end{align*}

  It is straightforward to check that the objective function, normalization, distance, and non-negativity
  constraints for $\FLP(n,d,\ell)$ (under the assumption of an $S_n$-invariant solution) match exactly those
  of $\KLP(n,d,\ell)$. For the MacWilliams inequalities, note that for every
  $h\in\im(\config_{n,\ell}^\Delta)$ and every $\alpha\in h$, we have
  \begin{align*}
    & \!\!\!\!\!\!
    \sum_{g \in \config(n,\el)} a'_g K_h(g) \geq 0  
    \\
    & \iff
    \sum_{g \in \config(n, \el)} a'_g \frac{\abs{h}}{\abs{g}} K_g(c) \geq 0
    & (\text{Reflection, \cref{lem:reflection}})
    \\
    & \iff
    \sum_{g \in \config(n, \el)} \frac{a'_g}{\lvert g\rvert} K_g(h) \geq 0
    \\
    & \iff
    \sum_{x \in (\F_2^n)^\el} a_x \chi_\alpha(x) \geq 0
    & (\text{Definition of }K_g),
  \end{align*}
  where the third equivalence follows since $a_g/\lvert g\rvert=a_x$ for every $x \in g$.
    
  The same proof goes through for $\FLP_\Lin$ and $\KLP_\Lin$.
\end{proof}

\begin{remark}
  The linear programs $\FLP$ and $\FLP_\Lin$ are \emph{not} invariant under the other automorphisms of the
  hypercube of the form $x \mapsto x + z$ ($z \in \F_2^n$), because of the normalization constraint and the
  distance constraints. It makes more sense to view the underlying space as $\F_2^n$ instead of the hypercube,
  which does not have the $\FF_2^n$ automorphism because the origin is treated specially.
\end{remark}

There is actually more symmetry in the programs than just $S_n$.  In the case of the program for non-linear
codes, there is a symmetry under the right action of $S_\el$ on $(\FF_2^n)^\ell$ that permutes the
\emph{words} $x_1, \dots, x_\el$, that is, we have $(x_1,\ldots,x_\ell)\cdot\tau\coloneqq
(x_{\tau(1)},\ldots,x_{\tau(\ell)})$ ($(x_1,\ldots,x_\ell)\in(\FF_2^n)^\ell$, $\tau\in S_\ell$).  In the case
of the program for linear codes, we have symmetry under the action of $\GL_\ell(\F_2)$ that applies a basis
change to $(x_1, \dots, x_\el)$, that is, it is given by
\begin{align*}
  (A\cdot x)_i & \coloneqq \sum_{j\in[\ell]} A[i,j]\cdot x_j \in \FF_2^n
\end{align*}
for every $A\in\GL_\ell(\FF_2)$, every $x\in(\FF_2^n)^\ell$ and every $i\in[\ell]$. The distance constraints
are evidently invariant under this action as it does not change the linear subspace spanned by
$(x_1,\ldots,x_\ell)$. The Fourier constraints are invariant since
\begin{align*}
  \chi_\alpha(A\cdot x)
  & =
  \chi_{A^\top\cdot\alpha}(x)
\end{align*}
for every $x,\alpha\in\FF_2^\ell$.

Note that the actions of $\GL_\ell(\FF)$ and $S_n$ commute with each other and thus induce an action of the
direct product $\GL_\ell(\FF)\times S_n$. Another reasonable definition of the higher-order Krawtchouk
polynomials and linear program symmetrizes under this larger group action of $\GL_\ell(\FF)\times S_n$. There
is one Krawtchouk polynomial and one free variable for each orbit of this action.
\begin{definition}[Fully symmetrized higher-order Krawtchouks]
  Let $O\coloneqq (\FF_2^n)^\el/(\GL_\ell(\FF_2)\times S_n)$ be the set of orbits of the
  $(\GL_\ell(\FF_2)\times S_n)$-action as above. For each $h\in O$ we define the higher-order Krawtchouk
  polynomial $K_h\colon O\to\RR$ by
  \begin{align*}
    K_h(g)
    & \coloneqq
    \sum_{(\alpha_1, \dots, \alpha_\el) \in h} \prod_{j=1}^\el\chi_{\alpha_j}(x_j),
  \end{align*}
  where $(x_1, \dots, x_\el)$ is any element in the orbit $g\in O$.
\end{definition}

Since the symmetry group is larger and the number of orbits is smaller, the size of the resulting LP is
smaller. However, since $\lvert\GL_\ell(\FF_2)\rvert = \prod_{t=0}^{\ell-1} (2^\ell - 2^t) = O_\el(1)$, for a constant
$\ell$, this would only decrease the size of $\KLP$ by a constant factor. For practical computations,
constant factors make a difference and this symmetrization should likely be performed. We chose our definition
of Krawtchouks in \cref{sec:binary_warmup} because the orbits are simpler to describe (being captured by
explicit combinatorial objects, configuration functions) and we can compute the set of orbits and the
Krawtchouk polynomials efficiently (see \cref{prop:complexity}).

There is an equivalent interpretation of $(\GL_\ell(\FF_2)\times S_n)$-orbits as ``subspace weight profiles'' as
follows. The right action of $S_n$ naturally induces an action over linear subspaces of $\FF_2^n$ given by
\begin{align*}
  W\cdot\sigma & \coloneqq \{w\cdot\sigma \mid w\in W\} \quad (W\leq\FF_2^n, \sigma\in S_n).
\end{align*}
It is straightforward to see that two $\ell$-tuples $(x_1,\ldots,x_\ell),(y_1,\ldots,y_\ell)\in(\FF_2^n)^\ell$
are in the same $(\GL_\ell(\FF_2)\times S_n)$-orbit if and only if $\Span\{x_1,\ldots,x_\ell\}$ and
$\Span\{y_1,\ldots,y_\ell\}$ are in the same $S_n$-orbit, which in turn is equivalent to saying that both
spaces have the same dimension, say $k$, and there are ordered bases $b^x = (b^x_1,\ldots,b^x_k)$ and $b^y =
(b^y_1,\ldots,b^y_k)$ of these spaces respectively such that $\config_{n,k}^\Delta(b^x) =
\config_{n,k}^\Delta(b^y)$. Thus, the hierarchy corresponding to the $(\GL_\ell(\FF_2)\times S_n)$-action has
an interesting interpretation as measuring weight statistics of linear subspaces of the linear code of
dimension at most $\ell$.

\subsection{The Hierarchy as an SDP}\label{sec:diagonalization}

The LP hierarchy is also equivalent to an SDP relaxation with the harsh
constraint that the SDP matrix must be \emph{translation invariant}.

Define the semi-definite program $\TransSDP(n,d,\el)$ as
\begin{align*}
  \max \quad
  & \sum_{x \in (\F_2^n)^\el} M[0, x]
  \\
  \text{s.t.} \quad
  & M[0, 0] = 1
  & &
  & & (\text{Normalization})
  \\
  & M[0, (x_1, \dots, x_\el)] = 0
  & & \exists i \in [\el], \abs{x_i} \in \{1,\dots, d-1\}
  & & (\text{Distance constraints})
  \\
  & M[x, y] = M[0, y-x]
  & & \forall x, y \in (\F_2^n)^\el
  & & (\text{Translation symmetry})
  \\
  & M \psdgeq 0
  & &
  & & (\text{PSD-ness})
  \\
  & M[x, y] \geq 0
  & & \forall x, y \in (\F_2^n)^\el
  & & (\text{Non-negativity}),
\end{align*}
where the variable is $M \in \R^{(\F_2^n)^\el \times (\F_2^n)^\el}$.

To form $\TransSDP_\Lin(n,d,\el)$, replace the distance constraints by
\begin{align*}
  M[0, (x_1, \dots, x_\el)] & = 0
  \quad \exists w \in \Span(x_1, \dots, x_\el), \abs{w} \in \{1, \dots, d-1\}.
\end{align*}

The crucial translation symmetry property of $\TransSDP$ ensures $M$ lies in the \emph{commutative} matrix
algebra $\Span\{D_z \mid z \in (\F_2^n)^\el\}$, where
\begin{align*}
  D_z[x,y] & \coloneqq \One[y - x = z].
\end{align*}
The coefficient of $M$ on $D_z$ is $M[0,z]$.

Since the matrices $D_z$ commute, they are simultaneously diagonalizable. More specifically, their common
eigenvectors are the Fourier characters.
\begin{fact}
  The matrices $D_z$ are simultaneously diagonalized by $(\chi_\alpha \mid \alpha \in (\FF_2^n)^\el)$ with the
  eigenvalue of $D_z$ on $\chi_\alpha$ being $\chi_\alpha(z)$.
\end{fact}
Therefore, the PSD-ness constraint in $\TransSDP$ is particularly simple:
to check that $\lambda_z D_z \psdgeq 0$, it is equivalent to
check $\sum_{z \in (\F_2^n)^\el}\lambda_z \chi_\alpha(z) \geq 0$  for all $\alpha \in (\F_2^n)^\el$.
This is a linear constraint on the $\lambda_z$, and hence we can express the SDP
as an LP, giving yet another formulation of the hierarchy.

\begin{proposition}
  For every $n,\ell\in\NN_+$ and every $d\in\{0,1,\ldots,n\}$, we have
  \begin{align*}
    \val(\FLP(n,d,\el)) & = \val(\TransSDP(n,d,\el)),\\
    \val(\FLP_\Lin(n,d,\el)) & = \val(\TransSDP_\Lin(n,d,\el)).
  \end{align*}
\end{proposition}

\begin{proof}
    The formal correspondence of the variables is $M[0, x] = a_x$.
    The Fourier coefficient constraints in $\FLP$ are equivalent to PSD-ness
    as described above,
    and the other constraints also match up.
\end{proof}

Along with \cref{prop:flp_klp}, the above implies that $\TransSDP$ also has the same value as $\KLP$.

\begin{remark}
  In previous convex relaxations for $A_2(n,d)$, in order to implement the program efficiently, a key
  technical step has been finding an explicit block diagonalization of the SDP matrix (which
  reduces the program size). This step requires significant technical
  work~\cite{Schrijver05,GMS12,Gijswijt09}. An advantage of 
  the LP hierarchy is that complete diagonalization is trivial. 
\end{remark}

\subsection{\texorpdfstring{The Hierarchy as $\vartheta'$}{The Hierarchy as Theta'}}

The hierarchy can also be seen as computing the (modified) Lov\'{a}sz $\vartheta'$ function on progressively
larger graphs, whose definition is recalled below. In fact, this formulation of the hierarchy holds for any
\emph{association scheme} (see \cref{thm:vartheta'} below).

\begin{definition}[$\vartheta'$ Program]\label{def:vartheta'}
  The (modified) Lov\'{a}sz $\vartheta'$ function is defined as follows. For a graph $G$, $\vartheta'(G)$ is the
  optimum value of the semi-definite program $\cS(G)$ given by
  \begin{align*}
    \max \quad
    & \ip{J}{M}
    \\
    \text{s.t.} \quad
    & \tr M = 1
    & &
    & & (\text{Normalization})
    \\
    & M[u,v] = 0
    & & \forall \{u,v\} \in E(G)
    & & (\text{Independent set})
    \\
    & M \succeq 0
    & &
    & & (\text{PSD-ness})
    \\
    & M[u,v] \ge 0
    & & \forall u,v \in V(G)
    & & (\text{Non-negativity}),
  \end{align*}
  where the variable is $M\in\RR^{V\times V}$ symmetric, $J$ is the all ones matrix and
  $\ip{A}{B}\coloneqq\tr(A^\top B)$.

  By strong duality $\vartheta'(G)$ is also the optimum value of the dual semi-definite program $\cS'(G)$ given
  by
  \begin{align*}
    \min \quad
    & \beta
    \\
    \text{s.t.}\quad
    & \beta I - N\succeq 0
    & &
    & & (\text{PSD-ness})
    \\
    & N[u,v]\geq 1
    & & \forall u,v\in V(G)\text{ with }\{u,v\}\notin E
    & & (\text{Independent set}),
  \end{align*}
  where the variables are $N\in\RR^{V\times V}$ symmetric and $\beta\in\RR$.
\end{definition}

It is straightforward to see that $\vartheta'(G)$ is an upper bound for the independence number of the graph
$G$ since if $A\subseteq V(G)$ is an independent set, then $\One_A\One_A^\top/\lvert A\rvert$ is a feasible
solution of $\cS(G)$ with value $\lvert A\rvert$.

In the same way that a code $C\subseteq\FF_2^n$ of distance at least $d$ can be seen as an independent set in
the graph $H_{n,d}$, we can see $C^\ell$ as an independent set in exclusion graphs defined below based on the
sets $\forbconfig(n,d,\ell)$ and $\forbconfig_\Lin(n,d,\ell)$ of \cref{def:holp}.

\begin{definition}[Exclusion Graph]
  We define the exclusion graph $H_{n,d,\ell}$ to have vertex set $(\F_2^n)^\el$ and edge set
  \begin{align*}
    E(H_{n,d,\ell})
    & \coloneqq
    \left\{
    (x,y)\in\binom{(\FF_2^n)^\ell}{2}
    \;\middle\vert\;
    \config_{n,\ell}^\Delta(x-y)\in\forbconfig(n,d,\ell)
    \right\}.
  \end{align*}

  We define $H_{n,d,\ell}^\Lin$ analogously replacing $\forbconfig(n,d,\ell)$ with
  $\forbconfig_\Lin(n,d,\ell)$.
\end{definition}

\begin{lemma}\label{lem:theta-formulation}
  For every $n,\ell\in\NN_+$ and every $d\in\{0,1,\ldots,n\}$, we have
  \begin{align*}
    \val(\TransSDP(n,d,\ell)) & = \val(\vartheta'(H_{n,d,\ell})),\\
    \val(\TransSDP_\Lin(n,d,\ell)) & = \val(\vartheta'(H^\Lin_{n,d,\ell})).
  \end{align*}
\end{lemma}

\begin{proof}
  The program $\cS(H_{n,d,\ell})$ corresponding to $\vartheta'(H_{n,d,\ell})$ is invariant under $\Aut(G)$, so
  it is invariant in particular under the translation action of $\FF_2^n$ on itself.
    
  Therefore, by \cref{fact:invariantsolution}, we may consider only solutions of $\cS(H_{n,d,\ell})$ that are
  translation invariant. Now there is a correspondence between solutions $M$ for $\TransSDP$ and translation
  invariant solutions $M'$ for $\cS(H_{n,d,\ell})$ given by $M = 2^{n\el}\cdot M'$. The proof goes through
  similarly for the linear case.
\end{proof}


\section{Generalized Krawtchouk Hierarchies from Association Schemes}
\label{sec:association}

In this section, we recall some of the basic definitions and results of association scheme theory and show
that our construction generalizes nicely to translation schemes with an underlying left module structure over
some ring, producing a translation scheme ``refining'' a tensor power of the original scheme; once this is
shown, MacWilliams identities and inequalities follow from the theory of translation schemes. The general
theory of association schemes will also be used to show our completeness and lifting results
of~\cref{sec:completeness} and~\cref{sec:lifting}.

Further background on association scheme theory can be found in the survey article by Martin and
Tanaka~\cite{MT09}.

\subsection{Association Scheme Theory Review}

\begin{definition}[Association schemes]
  An \emph{association scheme} is a pair $(X,R)$ where $X$ is a finite set and $R\subseteq 2^{X\times X}$ is a
  collection of non-empty subsets of $X\times X$, called \emph{relations}, satisfying the following
  properties.
  \begin{enumerate}
  \item $R$ is a partition of $X\times X$ into non-empty subsets.
  \item The \emph{diagonal relation} $\cD_X\coloneqq\{(x,x)\mid x\in X\}$ is an element of $R$.
  \item For every $r\in R$, the \emph{transposed relation} $r^\top\coloneqq\{(y,x)\mid (x,y)\in r\}$ is an element
    of $R$.
  \item For every $r,s,t\in R$, there exists an \emph{intersection number} $p_{rs}^t\in\NN$ such that for
    every $(x,y)\in t$, we have
    \begin{align*}
      p_{rs}^t & = \lvert\{z\in X \mid (x,z)\in r \land (z,y)\in s\}\rvert.
    \end{align*}
  \end{enumerate}
  Furthermore, the association scheme $S$ is called:
  \begin{enumerate}[label={\arabic*}.]
  \item \emph{Commutative}, if $p_{rs}^t = p_{sr}^t$ for every $r,s,t\in R$.
  \item \emph{Symmetric}, if $r^\top = r$ for every $r\in R$.
  \end{enumerate}
\end{definition}

\begin{fact}
  A symmetric association scheme is also commutative.
\end{fact}

As the next definition describes, association schemes can be viewed as certain matrix algebras.
\begin{definition}[Bose--Mesner algebra]
  Given an association scheme $S=(X,R)$, for each $r\in R$, let $D_r\in\CC^{X\times X}$ be given by
  $D_r[x,y]\coloneqq\One[(x,y)\in r]$. The \emph{Bose--Mesner algebra of $S$} is the $\CC$-algebra $\cA_S$
  generated by $\{D_r \mid r\in R\}$.
\end{definition}

The key observation underlying the above definition is that the intersection numbers $p_{rs}^t$ in the
definition of an association scheme guarantee that the linear span of $\{D_r \mid r \in R\}$ is closed under
matrix multiplication and adjoints. Since
\begin{align*}
  D_r D_s & = \sum_{t \in R} p_{rs}^t D_t,
\end{align*}
it follows that $(D_r)_{r\in R}$ is also a $\CC$-vector space basis of $\cA_S$. Furthermore, the above also
implies that $S$ is commutative if and only if its Bose--Mesner algebra is commutative. Moreover, $S$ is
symmetric if and only if every matrix in $\cA_S$ is symmetric.

\begin{fact}
  The Bose--Mesner algebra $\cA_S$ of a \emph{commutative} association scheme $S$ has a unique (up to
  permutation of its elements) $\CC$-vector space basis of idempotent orthogonal matrices $(E_s)_{s\in R'}$,
  where $\lvert R'\rvert = \lvert R\rvert$. That is, we have $E_{s_1}E_{s_2} = \One[s_1 = s_2]E_{s_1}$ for
  every $s_1,s_2\in R'$; namely, each $E_s$ is the projection onto a maximal common eigenspace of the matrices
  $\{D_r \mid r\in R\}$.
\end{fact}

Since both $(D_r \mid r\in R)$ and $(E_s \mid s\in R')$ are bases of $\cA_S$, each of their elements can be
written as a linear combination of the elements of the other basis using the $p$ and $q$-functions defined
below.
\begin{definition}[$p$-functions and $q$-functions]
  The \emph{$p$-functions} $p_r\colon R'\to\CC$ ($r\in R$) and \emph{$q$-functions} $q_s\colon R\to\CC$ ($s\in
  R'$) of an association scheme $S$ are the unique functions such that
  \begin{align*}
    D_r & = \sum_{s\in R'} p_r(s) E_s, &
    E_s & = \sum_{r\in R} q_s(r) D_r.
  \end{align*}
\end{definition}

An important subclass of association schemes is that of Schurian schemes defined below, which arise by
considering orbits of the diagonal action induced from a group action on the base set. The Bose--Mesner
algebra of Schurian schemes is then precisely the algebra of matrices that are invariant under the natural
conjugation action (see \cref{fact:schurianBoseMesner} below). This makes Schurian schemes particularly useful
in the study of semi-definite programs as for one such program $P$ (see \cref{fact:invariantsolution}).

\begin{definition}[Schurian scheme]
  Let $G$ be a group acting transitively on a finite set $X$. The \emph{Schurian scheme} associated with this
  action is defined as $S\coloneqq (X,(X\times X)/G)$, where $(X\times X)/G$ is the set of orbits of the
  natural diagonal action of $G$ on $X\times X$ given by $\sigma\cdot(x,y)\coloneqq (\sigma(x),\sigma(y))$
  ($x,y\in X$, $\sigma\in G$).
\end{definition}

\begin{fact}\label{fact:schurianBoseMesner}
  Let $G$ be a group acting transitively on a finite set $X$. The Schurian scheme is an association scheme and
  its Bose--Mesner algebra is precisely the algebra of $G$-invariant matrices under the natural conjugation
  action of $G$ on $\CC^{X\times X}$ given by $A\cdot\sigma = P_\sigma^{-1} A P_\sigma$ ($A\in\CC^{X\times
    X}$, $\sigma\in G$), where $P_\sigma\in\CC^{X\times X}$ is the permutation matrix given by
  $P_\sigma[x,y]\coloneqq \One[x = \sigma(y)]$ ($x,y\in X$).
\end{fact}

\begin{definition}[Codes in an association scheme]
  A \emph{code} in an association scheme $S = (X, R)$ is a non-empty subset $C\subseteq X$. 

  The \emph{inner distribution} of the code $C$ is the function $a^C\colon R\to\RR$ given by $a^C_r\coloneqq \lvert
  C^2\cap r\rvert/\lvert C\rvert$.

  For a set $D\subseteq R$, we say that $C$ is a $D$-code if $a_r^C = 0$ for every $r\in R\setminus D$.
\end{definition}

Given an association scheme $S = (X,R)$, by letting $S'\coloneqq(X,R')$, where $R'\coloneqq\{r\cup r^\top \mid
r\in R\}$, it is straightforward to check that $S'$ is a symmetric association scheme and if $C$ is a $D$-code
in $S$, then it is also a $D'$-code in $S'$, where $D'\coloneqq\{r\cup r^\top \mid r\in D\}$. For this reason,
when working with codes, we may suppose without loss of generality that the underlying association scheme is
symmetric.

\begin{definition}[Delsarte linear program]
  Given a set $D\subseteq R$, the Delsarte linear program associated with $(S,D)$ is the
  program $\cL_S(D)$ given by
  \begin{align*}
    \max \quad
    & \sum_{r\in R} a_r
    \\
    \text{s.t.} \quad
    & a_{\cD_X} = 1
    & &
    & & (\text{Normalization})
    \\
    & a_r = 0
    & & \forall r\in R\setminus D
    & & (\text{$D$-code constraints})
    \\
    & \sum_{r\in R} q_s(r)\cdot a_r \in\RR_+
    & & \forall s\in R'
    & & (\text{MacWilliams inequalities})
    \\
    & a_r\in\RR_+
    & & \forall r\in R
    & & (\text{Non-negativity}),
  \end{align*}
  where the variables are $(a_r)_{r\in R}$.
\end{definition}

It is clear that the inner distribution $a^C$ of a $D$-code is a feasible solution of $\cL_S(D)$, so the
optimum value of $\cL_S(D)$ is an upper bound on the size of $D$-codes (since $\sum_{r\in R} a_r^C = \lvert
C\rvert$).

When the underlying scheme $S$ is symmetric, one can use instead the \emph{real Bose--Mesner algebra of $S$},
which is the $\RR$-algebra generated by $\{D_r\mid r\in R\}$. All facts above remain true with $\CC$ replaced
with $\RR$.

\begin{definition}[Dual]
  Two association schemes $S=(X,R)$ and $\widehat{S}=(X,\widehat{R})$ over the same set $X$ are said to be
  \emph{dual} to each other if there exist bijections $f\colon R\to\widehat{R}'$ and $g\colon
  R'\to\widehat{R}$ such that for every $r\in R$ and every $s\in R'$, we have
  \begin{align*}
    p_r(s) & = q_{f(r)}(g(s)), &
    q_s(r) & = p_{g(s)}(f(r)).
  \end{align*}

  In this case it is typical to identify $R$ and $R'$ with $\widehat{R}'$ and $\widehat{R}$, respectively
  through these bijections. An association scheme $S=(X,R)$ is \emph{self-dual} when it is its own dual.
\end{definition}

\begin{definition}[Translation schemes]
  A \emph{translation scheme} is an association scheme $S = (X,R)$ in which $X$ is further equipped with an
  Abelian group structure and each relation $r\in R$ is an $X$-invariant set, i.e., for every $x,y,z\in X$, we
  have $(x,y)\in r\iff (z+x,z+y)\in r$.
  
  Equivalently, an association scheme $S = (X,R)$ where $X$ has Abelian group structure is a translation
  scheme if and only if there exists a function $f_S\colon X\to R$ such that $(x,y)\in f_S(x-y)$ for every
  $x,y\in X$. This means that we can also alternatively view $R$ as a partition of $X$ rather than $X\times
  X$; the relations of the association scheme are defined by the ``first row'' of the matrix.
\end{definition}

\begin{fact}
  Every translation scheme is commutative.
\end{fact}

\begin{remark}\label{rmk:translationsymmetric}
  It is easy to see that the function $f_S$ satisfies $f_S(-x) = f_S(x)^\top$ for every $x\in X$, hence $S$ is
  symmetric if and only if the function $f_S$ is even (i.e., $f_S(-x) = f_S(x)$ for every $x\in R$).
\end{remark}

\begin{remark}\label{rmk:pqtranslation}
  For translation schemes, the $p$ and $q$-functions can be computed via Fourier analysis as follows. Let us fix
  an indexing of the characters $\chi_x\colon X\to\CC$ of $X$ by $X$ so that $\chi_x(y) = \chi_y(x)$. Define the
  functions $\phi_r\colon X\to\CC$ ($r\in R$) by
  \begin{align*}
    \phi_r(x) & \coloneqq \sum_{y\in f_S^{-1}(r)} \overline{\chi_y(x)}.
  \end{align*}
  The level sets of each $\phi_r$ induce a partition of $X$, so we let $R'\subseteq 2^X$ be the coarsest common
  refinement of them and let $f'_S\colon X\to R'$ be the unique function such that $x\in f'_S(x)$ for every
  $x\in X$. Define then the functions $\psi_s\colon X\to\CC$ ($s\in R'$) by
  \begin{align*}
    \psi_s(y) & \coloneqq \sum_{x\in s} \chi_x(y) = \sum_{x\in (f_S')^{-1}(s)} \chi_x(y).
  \end{align*}

  It is a standard fact of association scheme theory that for every $x\in X$, every $r\in R$ and every $s\in
  R'$, we have
  \begin{align*}
    \phi_r(x) & = p_r(f_S'(x)), &
    \psi_s(x) & = q_s(f_S(x)).
  \end{align*}
\end{remark}

It is also straightforward to see that $f_S'$ induces a translation scheme structure
$\widehat{S}=(X,\widehat{R})$ on $X$ where $\widehat{R}\coloneqq\{r_s \mid s\in R'\}$ for the relations
$r_s\coloneqq\{(x,y)\in X\times X \mid f_S'(x-y) = s\}$ and $S$ and $\widehat{S}$ are dual of each other, and
as such we typically identify $R'$ with $\widehat{R}$ with $s\mapsto r_s$.

\begin{definition}[Additive codes and annihilator codes]
  A code $C$ in a translation scheme $S$ is called \emph{additive} if it is a subgroup of $X$. Trivially, an
  additive code $C$ is a $D$-code if and only if $f_S(C)\subseteq D$.

  Given an additive code $C$ in $S$, the \emph{annihilator code} of $C$ is
  \begin{align*}
    C^\circ & \coloneqq \{y\in X \mid \chi_y(x) = 1\}.
  \end{align*}
\end{definition}

It is straightforward to see that $C^\circ$ is additive and $(C^\circ)^\circ = C$ (as long as $C$ is
additive). It is more natural to see the annihilator code as a code in the dual scheme $\widehat{S}$, as the
(generalized) MacWilliams identities say that the inner distribution of $C^\circ$ in $\widehat{S}$ can be
retrieved from the inner distribution of $C$ in $S$ as follows.
\begin{theorem}[Generalized MacWilliams identities]
  For a translation scheme $S = (X, R)$ with dual scheme $\widehat{S} = (X, \widehat{R})$ and an additive code
  $C$ in $S$,  we have
  \begin{align*}
    a^{C^\circ}_s & = \frac{1}{\lvert C\rvert}\sum_{r\in R} q_s(r) a^C_r,
  \end{align*}
  for all $s \in \widehat{R}$,
\end{theorem}

\begin{definition}[Tensor product schemes]
  Given two schemes $S_1 = (X_1,R_1)$ and $S_2 = (X_2,R_2)$, their tensor product is the association scheme
  $S_1\otimes S_2\coloneqq(X_1\times X_2, R_1\otimes R_2)$, where $R_1\otimes R_2\coloneqq\{r_1\otimes r_2
  \mid r_1\in R_1\land r_2\in R_2\}$ for the relations
  \begin{align*}
    r_1\otimes r_2
    & \coloneqq
    \{((x_1,x_2),(y_1,y_2))\in (X_1\times X_2)\times (X_1\times X_2) \mid
    (x_1,y_1)\in r_1\land (x_2,y_2)\in r_2\}.
  \end{align*}

  For $\ell\in\NN_+$, the \emph{$\ell$th tensor power} of the association scheme $S$ is defined as
  \begin{align*}
    S^\ell & \coloneqq \mathop{\underbrace{S\otimes\cdots\otimes S}}\limits_{\ell\text{ times}}.
  \end{align*}
\end{definition}

It is straightforward to check that $S_1\otimes S_2$ is an association scheme that inherits the properties of
$S_1$ and $S_2$ in the sense that if both $S_1$ and $S_2$ are commutative (resp., symmetric, translation),
then $S_1\otimes S_2$ is so (in the case of translation scheme, the group structure in $X_1\times X_2$ is the
direct product group). Furthermore, if $C_i$ is a $D_i$-code in $S_i$ ($i\in[2]$), then $C_1\times C_2$ is a
$D_1\otimes D_2$-code in $S_1\otimes S_2$, where
\begin{align*}
  D_1\otimes D_2
  & \coloneqq
  \{r_1\otimes r_2 \mid r_1\in D_1\land r_2\in D_2\}.
\end{align*}

\begin{definition}[Refinement of a scheme]
  A refinement of an association scheme $S = (X,R)$ is an association scheme $S_2 = (X,R_2)$ over the same
  underlying set $X$ such that each $r\in R$ is a union of elements of $R_2$.
\end{definition}
Trivially, a $D$-code in $S$ is a $D'$-code in $S'$, where
\begin{align*}
  D' & = \{r'\in R_2 \mid \exists r\in D, r'\subseteq r\}.
\end{align*}

\begin{example}[Weak Hamming scheme]\label{ex:weakHamming}
  Given a non-trivial finite Abelian group $G$ and $n\in\NN_+$, the \emph{(weak) Hamming scheme of order $n$
    over $G$} is the translation scheme $\HH_n(G)\coloneqq (G^n,R)$, where $R\coloneqq\{r_i \mid
  i\in\{0,\ldots,n\}\}$ for
  \begin{align*}
    r_i & \coloneqq \{(x,y) \mid \Delta(x,y) = i\},
  \end{align*}
  where $\Delta(x,y)\coloneqq\lvert\{j\in[n] \mid x_j\neq y_j\}\rvert$ is the Hamming distance between $x$ and
  $y$. It is easy to see that $\HH_n(G)$ is a translation scheme over the direct product group $G^n$ in which
  $f_{\HH_n(G)}(x) = r_{\Delta(x,0)}$ for every $x\in G^n$. In fact, $\HH_n(G)$ is self-dual and its $p$ and
  $q$ functions are the Krawtchouk polynomials:
  \begin{align*}
    p_{r_i}(r_j)
    =
    q_{r_i}(r_j)
    & =
    \sum_{t=0}^j
    (-1)^t (\lvert G\rvert - 1)^{i-t} \binom{j}{t}\binom{n - j}{i - t}.
  \end{align*}

  We use the notation $\HH_n\coloneqq\HH_n(\FF_2)$, when the underlying group is the field with two
  elements. Under this notation, a binary code of blocklength $n$ and distance $d$ is simply a $D_d$-code in
  $\HH_n$, where $D_d\coloneqq\{r_0,r_d,r_{d+1},\ldots,r_n\}$.

  Alternatively, the weak Hamming scheme can be seen as a Schurian scheme as follows. Consider the natural
  right action of the symmetric group $S_n$ on $n$ letters on $G^n$ given by $(x\cdot\sigma)_i\coloneqq
  x_{\sigma(i)}$ ($x\in G^n$, $\sigma\in S_n$, $i\in[n]$) and the natural left action of the symmetric group
  $S_{G^n}$ on $G^n$. These actions together induce an action of a semidirect product $S_{G^n}\rtimes S_n$
  on $G^n$ whose associated Schurian scheme is precisely $\HH_n(G)$.
\end{example}

\begin{example}[Strong Hamming scheme]\label{ex:strongHamming}
  Given a non-trivial finite Abelian group $G$ and $n\in\NN_+$, the \emph{strong Hamming scheme of order $n$
    over $G$} is the translation scheme $\HH_n^*(G)\coloneqq (G^n,R)$, where $R\coloneqq\{r_h \mid h\in
  \{0,1,\ldots,n\}^G\}\setminus\{\varnothing\}$, for
  \begin{align*}
    r_h & \coloneqq \{(x,y) \mid \forall g\in G, \lvert x-y\rvert_g = h(g)\},
  \end{align*}
  where
  \begin{align*}
    \lvert z\rvert_g & \coloneqq \lvert\{i\in[n] \mid z_i = g\}\rvert.
  \end{align*}

  It is easy to see that $\HH_n^*(G)$ is a translation scheme that is a refinement of $\HH_n(G)$ and for every
  $x\in G^n$, we have $f_{\HH_n^*(G)}(x) = r_{h_x}$ for $h_x(g)\coloneqq\lvert x\rvert_g$. In fact, $\HH_n^*(G)$ is
  self-dual and its $p$ and $q$ functions are given by
  \begin{align*}
    p_{r_{h_1}}(r_{h_2})
    =
    q_{r_{h_1}}(r_{h_2})
    & =
    \sum_{F\in\cF}
    \left(\prod_{g_1\in G} \frac{h_1(g_1)!}{\prod_{g_2\in G} F(g_1,g_2)!}\right)
    \prod_{g_1,g_2\in G} \chi_{g_1}(g_2)^{F(g_1,g_2)},
  \end{align*}
  where $\cF$ is the set of all functions $F\colon G\times G\to\{0,\ldots,n\}$ such that
  \begin{align*}
    \sum_{g'\in G} F(g,g') & = h_1(g), &
    \sum_{g'\in G} F(g',g) & = h_2(g), &
  \end{align*}
  for every $g\in G$.

  Alternatively, the strong Hamming scheme can be seen as a Schurian scheme as follows. Consider the natural
  right action of the symmetric group $S_n$ on $n$ letters on $G^n$ given by $(x\cdot\sigma)_i\coloneqq
  x_{\sigma(i)}$ ($x\in G^n$, $\sigma\in S_n$, $i\in[n]$) and the natural translation action of the product
  group $G^n$ on itself. These actions together induce an action of a semidirect product $G^n\rtimes S_n$ on
  $G^n$ whose associated Schurian scheme is precisely $\HH_n^*(G)$.
\end{example}

While for binary alphabets, the strong and weak Hamming scheme obviously coincide (i.e., $\HH_n^*(\FF_2) =
\HH_n(\FF_2)$), for larger alphabets this is not the case.

We finish this section recalling the connection of the Delsarte linear program with the modified
Lov\'{a}sz $\vartheta'$-function from graph theory (see \cref{def:vartheta'}).

\begin{definition}
  Given a commutative association scheme $S=(X,R)$ and $D\subseteq R$ with $\cD_X\in D$, the graph $G_S(D)$ is
  defined by
  \begin{align*}
    V(G_S(D))
    & \coloneqq
    X,
    \\
    E(G_S(D))
    & \coloneqq
    \left\{\{x,y\}\in\binom{X}{2}
    \;\middle\vert\;
    \exists r\in R\setminus D, (x,y)\in r\right\}.
  \end{align*}
\end{definition}

Under this definition, a $D$-code on $S$ is simply an independent set in the graph $G_S(D)$. The next theorem
by Schrijver connects the Delsarte linear program $\cL_S(D)$ to the semi-definite program $\cS(G_S(D))$.

\begin{theorem}[Schrijver~\protect{\cite{Schrijver79}}]\label{thm:vartheta'}
  Let $S=(X,R)$ be a commutative association scheme and let $D\subseteq R$ with $\cD_X\in D$. Then
  $\vartheta'(G_S(D))$ is equal to the optimum value of the Delsarte linear program $\cL_S(D)$.
\end{theorem}

\subsection{Natural Refinements of Translation Schemes}
\label{subsec:ftt}

In this section, we show how our construction generalizes nicely to translation schemes with an underlying
left module structure over some ring. This can be applied to any translation scheme by recalling that any
Abelian group is naturally a $\ZZ$-module, but sometimes it is more interesting to use a different module
structure (e.g., a vector space over a finite field).

\begin{definition}[Association scheme automorphism]  
  An \emph{automorphism} of an association scheme $S = (X,R)$ is a bijection $f\colon X\to X$ that fixes each
  $r\in R$ as a set, that is, we have $\{(f(x),f(y)) \mid (x,y)\in r\} = r$. The group of automorphisms of $S$
  is denoted $\Aut(S)$.
\end{definition}

\begin{definition}
  Let $S = (X,R)$ be a translation scheme and let $f_S\colon X\to R$ be the unique function such that
  $(x,y)\in f_S(x-y)$ for every $(x,y)\in X\times X$. Given a ring $K$, let us further assume that $X$ is
  equipped with a left $K$-module structure extending the Abelian group structure and let
  $\Aut_K(S)\leq\Aut(S)$ be the subgroup of automorphisms of the association scheme $S$ that are also left
  $K$-module automorphisms of $X$.

  A code $C$ in $S$ is called \emph{$K$-linear} if it is both additive and $K$-invariant in the sense that
  $kx\in C$ for every $k\in K$ and every $x\in C$.

  Let $\ell\in\NN_+$ and let $T\subseteq K^\ell$ be a collection of $\ell$-tuples of $K$. Define the function
  $f_{S,T}\colon X^\ell\to R^T$ by
  \begin{align*}
    f_{S,T}(x)(k) & \coloneqq f_S\left(\sum_{i=1}^\ell k_i x_i\right) \qquad (x\in X^\ell, k\in T).
  \end{align*}

  We say that \emph{$f_{S,T}$ factors through types of $S$} if for every $x,y\in X^\ell$, we have $f_{S,T}(x) =
  f_{S,T}(y)$ if and only if there exists $\sigma\in\Aut_K(S)$ such that $\sigma(x_i) = y_i$ for every
  $i\in[\ell]$.
\end{definition}

Similarly to symmetric difference configurations of \cref{def:configuration}, the function $f_{S,T}$ captures
information of the value of $f_S$ in $K$-linear combinations of $\ell$-tuples of elements of $X$ using
coefficients in $T\subseteq K^\ell$. The definition of factoring through types then requires that this
information is enough to determine the orbit\footnote{This is also the reason behind the name ``factors
  through types'': in model theory, two elements of a \emph{finite} model have the same type if and only if
  they are in the same orbit under the action of the automorphism group.} of a tuple $x\in X^\ell$ under the
natural diagonal action of $\Aut_K(S)$.

\begin{remark}\label{rmk:AutKS}
  It is straightforward to see that $f_S$ is $\Aut_K(S)$-invariant in the sense that $f_S\comp\sigma = f_S$
  for every $\sigma\in\Aut_K(S)$.

  This in particular implies that in the definition of $f_{S,T}$ factoring through types, the backward
  implication always holds: if $x,y\in X^\ell$ and $\sigma\in\Aut_K(S)$ are such that $\sigma(x_i) = y_i$ for
  every $i\in[\ell]$, then for every $k\in T$ we have
  \begin{align*}
    f_{S,T}(y)(k)
    & =
    f_S\left(\sum_{i=1}^\ell k_i \sigma(x_i)\right)
    =
    f_S\left(\sigma\left(\sum_{i=1}^\ell k_i x_i\right)\right)
    =
    f_S\left(\sum_{i=1}^\ell k_i x_i\right)
    =
    f_{S,T}(x)(k).
  \end{align*}
\end{remark}

The next definition generalizes the construction of \cref{sec:binary_warmup}.
\begin{definition}
  Let $S = (X,R)$ be a translation scheme over a left $K$-module $X$, let $\ell\in\NN_+$, let $T\subseteq
  K^\ell$ be such that $e_i\in T$ for every $i\in[\ell]$, where $(e_i)_j\coloneqq \One[i = j]$ and suppose
  $f_{S,T}$ factors through types.

  The \emph{$T$-refined $\ell$th tensor power of $S$} is the translation scheme $S^{\ell,T} = (X^\ell,R^{\ell,T})$ is
  defined by letting
  \begin{align*}
    R^{\ell,T} & \coloneqq \{r_h \mid h\in R^T\}\setminus\{\varnothing\},
  \end{align*}
  where 
  \begin{align*}
    r_h & \coloneqq \{(x,y)\in X^\ell\times X^\ell \mid f_{S,T}(x-y) = h\}
  \end{align*}
  for each function $h\colon T\to R$ and $X^\ell$ is equipped with the direct product left $K$-module
  structure. \cref{thm:ftt} below shows that $S^{\ell,T}$ is indeed a translation scheme.
\end{definition}

Before we show that $S^{\ell,T}$ is indeed a translation scheme, two particular choices of $(K,T)$ deserve
special attention.
\begin{enumerate}
\item When $K=\ZZ$ and $T=\{e_i \mid i\in[\ell]\}$, then $S^{\ell,T}$ is just the $\ell$th tensor power $S^\ell$.
\item When $T = K^\ell$, then $f_{S,T}$ encodes the \emph{complete $K$-linear configuration} of $x$ as it is
  able to determine the value of $f_S$ in any $K$-linear combination of $x_1,\ldots,x_\ell$. This will be
  particularly useful when $K$ is a (finite) field (and thus $X$ is a $K$-vector space).
\end{enumerate}

\begin{theorem}\label{thm:ftt}
  Let $S = (X,R)$ be a translation scheme over a left $K$-module $X$, let $\ell\in\NN_+$, let $T\subseteq
  K^\ell$ be such that $e_i\in T$ for every $i\in[\ell]$, where $(e_i)_j\coloneqq \One[i = j]$ and suppose
  $f_{S,T}$ factors through types.

  Then the following hold.
  \begin{enumerate}[label={\arabic*.}, ref={(\arabic*)}]
  \item $S^{\ell,T}\coloneqq (X^\ell,R^{\ell,T})$ is a translation scheme over the direct product group
    $X^\ell$ that refines the tensor power $S^\ell$.%
    \label{thm:ftt:scheme}
  \item If $S$ is symmetric, then so is $S^{\ell,T}$.%
    \label{thm:ftt:symm}
  \item If $C$ is a $K$-linear $D$-code in $S$, then $C^\ell$ is a $K$-linear $D^{\ell,T}$-code in
    $S^{\ell,T}$, where
    \begin{align*}
      D^{\ell,T} & \coloneqq \{r_h \mid h\in R^T\land \im(h)\subseteq D\}\setminus\{\varnothing\}.
    \end{align*}
    \label{thm:ftt:code}
  \end{enumerate}
\end{theorem}

\begin{proof}
  We start proving item~\ref{thm:ftt:scheme}.

  It is obvious that $R^{\ell,T}$ forms a partition of $X^\ell\times X^\ell$ into non-empty subsets.

  Note also that if $(x,y)\in X^\ell\times X^\ell$ are such that $f_{S,T}(x-y)(k) = \cD_X$ for every $k\in T$, then
  since $e_i\in T$, we get $x_i = y_i$, thus $x = y$. Since we also have $f_{S,T}(0)(k) = f_S(0) = \cD_X$ for
  every $k\in T$, it follows that for the function $T\to R$ that is constant equal to $\cD_X$ we have
  $r_{\cD_X} = \cD_{X^\ell}$.

  It is also easy to see that for $h\colon T\to R$, by letting $h^\top\colon T\to R$ be given by $h^\top(k)\coloneqq
  h(k)^\top$, we have $r_h^\top = r_{h^\top}$.

  Note further that for each $r_1,r_2,\ldots,r_\ell\in R$, we have
  \begin{align*}
    r_1\otimes\cdots\otimes r_\ell & = \bigcup\{r_h \mid h\in R^T \land \forall i\in [\ell], h(e_i)=r_i\}.
  \end{align*}

  It remains only to show that the existence of the intersection numbers for $S^{\ell,T}$.

  For every $h_1,h_2\colon T\to R$ and every $x,y\in X^\ell$, let
  \begin{align*}
    N_{h_1,h_2}(x,y) & \coloneqq \lvert\{z\in X^\ell \mid f_{S,T}(x-z) = h_1 \land f_{S,T}(z-y) = h_2\}\rvert.
  \end{align*}
  It is sufficient to show that if $x,y,x',y'\in X^\ell$ are such that $f_{S,T}(x-y) = f_{S,T}(x'-y')$, then
  $N_{h_1,h_2}(x,y) = N_{h_1,h_2}(x',y')$.

  Since $f_{S,T}$ factors through types of $S$, there exists $\sigma\in\Aut_K(S)$ such that $\sigma(x_i-y_i) =
  x'_i-y'_i$ for every $i\in[\ell]$. Then we have{\footnotesize
  \begin{align*}
    N_{h_1,h_2}(x',y')
    & =
    \lvert\{z\in X^\ell \mid f_{S,T}(x'-z) = h_1\land f_{S,T}(z-y') = h_2\}\rvert
    \\
    & =
    \left\lvert\left\{u\in X^\ell \;\middle\vert\;
    \forall k\in T,
    \left(f_S\left(\sum_{i=1}^\ell k_i(x'_i-y'_i-u_i)\right) = h_1(k)\land
    f_S\left(\sum_{i=1}^\ell k_i u_i\right) = h_2(k)
    \right)
    \right\}
    \right\rvert
    \\
    & =
    \left\lvert\left\{u\in X^\ell \;\middle\vert\;
    \forall k\in T,
    \left(f_S\left(\sum_{i=1}^\ell k_i(\sigma(x_i-y_i)-u_i)\right) = h_1(k)\land
    f_S\left(\sum_{i=1}^\ell k_i u_i\right) = h_2(k)
    \right)
    \right\}
    \right\rvert
    \\
    & =
    \left\lvert\left\{u\in X^\ell \;\middle\vert\;
    \forall k\in T,
    \left(f_S\left(\sum_{i=1}^\ell k_i(x_i-y_i-\sigma^{-1}(u_i))\right) = h_1(k)\land
    f_S\left(\sum_{i=1}^\ell k_i\cdot \sigma^{-1}(u_i)\right) = h_2(k)
    \right)
    \right\}
    \right\rvert
    \\
    & =
    \left\lvert\left\{w\in X^\ell \;\middle\vert\;
    \forall k\in T,
    \left(f_S\left(\sum_{i=1}^\ell k_i(x_i-w_i)\right) = h_1(k)\land
    f_S\left(\sum_{i=1}^\ell k_i(w_i - y_i)\right) = h_2(k)
    \right)
    \right\}
    \right\rvert
    \\
    & =
    N_{h_1,h_2}(x,y),
  \end{align*}}
  where the second equality follows from the substitution $u_i\coloneqq z_i-y'_i$, the fourth equality follows
  since $\sigma$ is a left $K$-module automorphism of $X$ and $f_S$ is $\sigma$-invariant (see
  \cref{rmk:AutKS}) and the fifth equality follows from the substitution $w_i\coloneqq \sigma^{-1}(u_i) +
  y_i$.

  \medskip

  For item~\ref{thm:ftt:symm}, since a translation scheme $S$ is symmetric if and only if the function $f_S$
  is even (see \cref{rmk:translationsymmetric}), the fact that $S$ is symmetric implies $f_S$ is even, hence
  $f_{S,T}$ is also even and thus $S^{\ell,T}$ is symmetric (as $f_{S^{\ell,T}}(x) = r_{f_{S,T}(x)}$).

  \medskip

  For item~\ref{thm:ftt:code}, it is obvious that $C^\ell$ is both a subgroup of $X^\ell$ and $K$-invariant. Let
  $x\in C^\ell$ and note that since $C$ is $K$-linear, for every $k\in T$, we have $\sum_{i=1}^\ell k_i x_i\in C$,
  so $f_{S,T}(x)(k) = f_S(\sum_{i=1}^\ell k_i x_i)\in D$ and thus
  \begin{align*}
    f_{S^{\ell,T}}(C^\ell) & = \{r_{f_{S,T}(x)} \mid x\in C^\ell\} \subseteq D,
  \end{align*}
  hence $C^\ell$ is a $K$-linear $D^{\ell,T}$-code.
\end{proof}

Note that if the underlying translation scheme $S=(X,R)$ is a Schurian scheme associated to a group action of
a semidirect product $X\rtimes G$ in $X$ such as in the (weak or strong) Hamming scheme (see
\cref{ex:weakHamming,ex:strongHamming}), then we could easily produce a Schurian translation scheme refining
the $\ell$th tensor power by considering the action of a semidirect product $X^\ell\rtimes G$ on $X^\ell$
obtained by considering the product action of $X^\ell$ and the diagonal action of $G$. However, even in the
Schurian case, the true value of \cref{thm:ftt} above lies in two facts:
\begin{enumerate}
\item The relation of $(x,y)\in X^\ell\times X^\ell$ is determined by the value of $f_{S,T}(x-y)$, that is,
  the value of $f_S$ in $K$-linear combinations of $(x-y)$ using tuples in $T$.
\item If $C$ is a $K$-linear $D$-code in $S$, we can deduce ``extra'' restrictions of the code $C^\ell$ in
  $S^{\ell,T}$ besides the ones that follow from the tensor power. More specifically, it is trivial that
  $C^\ell$ is a $D^{\otimes \ell}$-code in the tensor power $S^\ell$, which in turn implies that it is a
  $\widehat{D}$-code in $S^{\ell,T}$, where
  \begin{align*}
    \widehat{D}
    & \coloneqq
    \{r_h\in R^{\ell,T} \mid \exists r\in R^{\otimes\ell}, r_h\subseteq r\}
    \\
    & =
    \{r_h \mid h\in R^T\land \forall i\in[\ell], h(e_i)\in D\}.
  \end{align*}
  However, \cref{thm:ftt} says that we further have $h(k)\in D$ for \emph{every} $k\in T$ (not only for the
  $e_i$).
\end{enumerate}

Our next objective is to show that under mild assumptions on the structure of the left $K$-module $G$, for the
weak and strong Hamming schemes $\HH_n(G)$ and $\HH_n^*(G)$ of \cref{ex:weakHamming,ex:strongHamming}, the
functions $f_{\HH_n(G),K^\ell}$ and $f_{\HH_n^*(G),K^\ell}$ factor through types when $G^n$ is equipped with
the direct product left $K$-module structure. A particular case when all such mild assumptions hold is when $G
= K = \FF$ for some finite field $\FF$.

\begin{remark}\label{rmk:klp}
  Once we prove that the function $f_{\HH_n(\FF_2),\FF_2^\ell}$ factors through types of the weak Hamming
  scheme $\HH_n(\FF_2)$ (\cref{cor:highLPboundweak} below), the hierarchy of linear programs
  presented in \cref{sec:binary_warmup} can be retrieved as $\KLP_\Lin(n,d,\ell) =
  \cL_{S^{\ell,T}}(D_d^{\ell,T})$ for $S = \HH_n(\FF_2)$, $T = \FF_2^\ell$ and
  $D_d\coloneqq\{r_0,r_d,r_{d+1},\ldots,r_n\}$.

  Analogously, the hierarchy for non-linear codes can be retrieved as
  $\KLP(n,d,\ell)=\cL_{S^{\ell,T}}(\widehat{D}_d^\ell)$ using instead the weaker restriction set
  \begin{align*}
    \widehat{D}_d^\ell & \coloneqq \{r\in R^{\ell,T} \mid \exists r'\in D_d^{\otimes\ell}, r\subseteq r'\}.
  \end{align*}

  These can also be retrieved using the strong Hamming scheme $\HH_n^*(\FF_2)$ instead (see
  \cref{cor:highLPboundstrong} below) as for the binary case we have $\HH_n^*(\FF_2)=\HH_n(\FF_2)$.

  However, the same corollaries apply for the more general case of a (not necessarily binary) finite field
  $\FF$, in which the weak and strong Hamming schemes are different. In this case, we define
  \begin{align*}
    \KLP^\FF_\Lin(n,d,\ell)\coloneqq \cL_{S^{\ell,T}}(D_d^{\ell,T})
  \end{align*}
  using $S = \HH_n(\FF)$, $T = \FF^\ell$ and $D_d\coloneqq\{r_0,r_d,r_{d+1},\ldots,r_n\}$.

  We can also define $\KLP^\FF(n,d,\ell)$ analogously for arbitrary finite fields, but as we will see in
  \cref{prop:main:lifting:formal}, these hierarchies for non-linear codes collapse and yield the same bound as
  the usual Delsarte linear program.
\end{remark}

Before we start with the case of the strong Hamming scheme $\HH_n^*(G)$, let us prove a few lemmas.

\begin{lemma}\label{lem:nontrivialcharactersum}
  Let $K$ be a finite ring, let $G$ be a finite simple left $K$-module and let $\chi\colon G\to\CC$ be a
  non-trivial character of $G$. Then
  \begin{align*}
    \sum_{k\in K} \chi(kg)
    & =
    \lvert K\rvert\cdot\One[g = 0]
  \end{align*}
  for every $g\in G$.
\end{lemma}

\begin{proof}
  If $g = 0$, then $\chi(kg) = 1$ for every $k\in K$, thus $\sum_{k\in K} \chi(kg) = \lvert K\rvert$.

  On the other hand, if $g\neq 0$, then we must have $Kg = G$ as $Kg$ is a non-trivial left $K$-submodule of $G$
  and $G$ is simple. Note also that for $k_1,k_2\in K$, we have $\chi(k_1g) = \chi(k_2g)$ if and only if
  $k_1-k_2\in H$, where
  \begin{align*}
    H & \coloneqq \{k\in K \mid \chi(kg) = 1\}.
  \end{align*}
  Since $H$ is a subgroup of $K$, each level set of $k\mapsto\chi(kg)$ is a coset of $H$, so they must all have
  the same size, namely $\lvert K\rvert/\lvert H\rvert$, so we get
  \begin{align*}
    \sum_{k\in K} \chi(kg)
    & =
    \frac{\lvert K\rvert}{\lvert H\rvert}\cdot\frac{\lvert\im(\chi)\rvert}{\lvert G\rvert}
    \sum_{g'\in G}\chi(g')
    = 0,
  \end{align*}
  where the last equality follows since $\chi$ is a non-trivial character (so it is orthogonal to the trivial
  character).
\end{proof}

The next lemma says that the joint distribution of $\ell$ random variables with values in a finite simple left
$K$-module (for a finite ring $K$) can be recovered from the (individual) distributions of all $K$-linear
combinations of them.

\begin{lemma}\label{lem:distlincomb}
  Let $K$ be a finite ring, let $G$ be a finite simple left $K$-module and let $\chi\colon G\to\CC$ be a
  non-trivial character of $G$. Suppose further that $\rn{X}$ is a random variable with values in $G^\ell$ for
  some $\ell\in\NN_+$ and for every $k\in K^\ell$, let $\rn{Y}_k \coloneqq \sum_{i=1}^\ell k_i\rn{X}_i$. Then 
  \begin{align*}
    \PP[\rn{X} = x]
    & =
    \frac{1}{\lvert K^\ell\rvert}
    \sum_{\substack{k\in K^\ell\\y\in G}}
    \chi\left(\sum_{i=1}^\ell k_i x_i - y\right)
    \cdot\PP[\rn{Y}_k = y]
  \end{align*}
  for every $x\in G^\ell$.
\end{lemma}

\begin{proof}
  First note that for every $k\in K^\ell$ and every $y\in G$, we have
  \begin{align*}
    \PP[\rn{Y}_k = y]
    & =
    \sum_{\substack{z\in G^\ell\\ \sum_{i=1}^\ell k_i z_i = y}} \PP[\rn{X} = z],
  \end{align*}
  which implies that
  \begin{align}
    \frac{1}{\lvert K^\ell\rvert}
    \sum_{\substack{k\in K^\ell\\y\in G}}
    \chi\left(\sum_{i=1}^\ell k_i x_i - y\right)
    \cdot\PP[\rn{Y}_k = y]
    & =
    \frac{1}{\lvert K^\ell\rvert}
    \sum_{z\in G^\ell}
    \PP[\rn{X}=z]\cdot
    \sum_{\substack{k\in K^\ell\\y\in G\\ \sum_{i=1}^\ell k_i z_i = y}}
    \chi\left(\sum_{i=1}^\ell k_i x_i - y\right)
    \notag\\
    & =
    \frac{1}{\lvert K^\ell\rvert}
    \sum_{z\in G^\ell}
    \PP[\rn{X}=z]\cdot
    \sum_{k\in K^\ell}
    \chi\left(\sum_{i=1}^\ell k_i (x_i - z_i)\right)
    \notag\\
    & =
    \PP[\rn{X}=x]
    +
    \frac{1}{\lvert K^\ell\rvert}
    \sum_{\substack{z\in G^\ell\\z\neq x}}
    \PP[\rn{X}=z]\cdot
    \sum_{k\in K^\ell}
    \chi\left(\sum_{i=1}^\ell k_i (x_i - z_i)\right).
    \label{eq:chisum}
  \end{align}

  To complete the proof, it is sufficient to show that the inner sum of the second term in~\eqref{eq:chisum}
  is zero. But note that
  \begin{align*}
    \sum_{k\in K^\ell}
    \chi\left(\sum_{i=1}^\ell k_i (x_i - z_i)\right)
    & =
    \prod_{i=1}^\ell \sum_{k\in K} \chi(k(x_i-z_i))
  \end{align*}
  and since $x_i\neq z_i$ for at least one $i\in[\ell]$, from \cref{lem:nontrivialcharactersum}, the
  above is zero as desired.
\end{proof}

Let us also recall one standard fact from algebra.

\begin{lemma}\label{lem:finitefaithful}
  Let $K$ be a ring and $G$ be a left $K$-module. If $G$ is both finite and faithful, that is, the annihilator
  \begin{align*}
    \Ann_K(G) & \coloneqq \{k\in K\mid \forall g\in G, kg = 0\}
  \end{align*}
  is trivial (i.e., $\Ann_K(G) = \{0\}$). Then $K$ is finite.
\end{lemma}

\begin{proof}
  Each element $k\in K$ induces a left $K$-module endomorphism $f_k\colon G\to G$ of $G$ given by $f_k(g) =
  k\cdot g$. Note that for $k_1,k_2\in K$, we have $f_{k_1} = f_{k_2}$ if and only if $k_1-k_2\in\Ann_K(G)$,
  so since $G$ is faithful, all $f_k$ must be different. If $G$ is finite, it has only finitely many
  endomorphisms, so $K$ must also be finite.
\end{proof}

We can now show that for the strong Hamming scheme $\HH_n^*(G)$ over a finite simple left $K$-module $G$, the
function $f_{\HH_n^*(G),K^\ell}$ factors through types. We recall that when $K$ is commutative, the simplicity
condition reduces to saying that there is a maximal ideal $I$ of $K$ such that $G\cong K/I$ as a $K$-module,
that is, it is a $1$-dimensional vector space over the (necessarily finite) field $\FF\coloneqq K/I$ (note
that $K$ itself does not need to be a field, e.g., $K=\ZZ$ and $G=\ZZ_p$ for some prime $p$); in other words,
all cases when $K$ is commutative and $G$ is simple are indirectly captured by the usual case $K = G = \FF$
for some finite field $\FF$.

\begin{proposition}\label{prop:strongHammingftt}
  Let $K$ be a ring, let $G$ be a finite simple left $K$-module and let $n,\ell\in\NN_+$. Consider the strong
  Hamming scheme $\HH_n^*(G)$ of order $n$ over $G$ equipped with the direct product left $K$-module structure
  on $G^n$. Then $f_{\HH_n^*(G),K^\ell}$ factors through types of $\HH_n^*(G)$.
\end{proposition}

\begin{proof}
  First, recall that the strong Hamming scheme relations are given by
  \begin{align*}
    r_h & \coloneqq \{(x,y) \mid \forall g\in G, \lvert x-y\rvert_g = h(g)\},
  \end{align*}
  for $h\in\{0,1,\ldots,n\}^G\}$ such that $r_h$ is non-empty, where $\lvert z\rvert_g$ is the number of
  positions $i\in[n]$ such that $z_i=g$. In this proof we will abuse notation and write $f_{\HH_n^*(G)}(x) =
  h$ in place of $f_{\HH_n^*(G)}(x) = r_h$, that is, we will view $f_{\HH_n^*(G)}$ as a function with values
  in $\{0,1,\ldots,n\}^G$ rather than in $R\coloneqq\{r_h \mid
  h\in\{0,1,\ldots,n\}^G\}\setminus\{\varnothing\}$. Accordingly, we will also view $f_{\HH_n^*(G),K^\ell}$ as
  a function with values in $(\{0,1,\ldots,n\}^G)^{K^\ell}$ rather than in $R^{K^\ell}$.

  \medskip
  
  First, we claim that it is enough to prove the case when $K$ is finite. Indeed, recall that the annihilator
  $\Ann_K(G)$ of $G$ in $K$ is a two-sided ideal of $K$ and the left $K$-module structure on $G$ induces a
  natural left $K/\Ann_K(G)$-module structure given by $(k + \Ann_K(G))\cdot g \coloneqq k\cdot g + \Ann_K(G)$
  ($k\in K$, $g\in G$). Note also that for every $x\in (G^n)^\ell$ and every $k\in K^\ell$, we have
  \begin{align*}
    f_{\HH_n^*(G),(K/\Ann_K(G))^\ell}(x)\bigl((k_i + \Ann_K(G) \mid i\in[\ell])\bigr)
    & =
    f_{\HH_n^*(G)}\left(\sum_{i=1}^n k_i\cdot x_i\right)
    =
    f_{\HH_n^*(G), K^\ell}(x)(k),
  \end{align*}
  so if $f_{\HH_n^*(G),(K/\Ann_K(G))^\ell}$ factors through types, then $f_{\HH_n^*(G),K^\ell}$ also does
  so. From \cref{lem:finitefaithful}, $K/\Ann_K(G)$ must be finite as $G$ is a finite faithful left
  $K/\Ann_K(G)$-module, completing our reduction.

  \medskip

  Let us now prove the case when $K$ is finite.

  First, note that $\Aut_K(\HH_n(G))$ contains a subgroup isomorphic to the symmetric group $S_n$ on $n$
  letters\footnote{In fact, $\Aut_K(\HH_n(G))$ is precisely equal to this subgroup, but we will not need this
    fact.}; namely, the natural right action of $S_n$ on $G^n$ given by $(x\cdot\sigma)_i \coloneqq x_{\sigma(i)}$
  ($x\in G^n$, $\sigma\in S_n$, $i\in[n]$) is free and preserves the left $K$-module structure of $G^n$ and
  every relation $r_h\in R$ of $\HH_n^*(G)$ is invariant under this action and thus this action induces a
  subgroup of $\Aut_K(\HH_n^*(G))$ isomorphic to $S_n$ where $\sigma\in S_n$ corresponds to the automorphism
  $F_\sigma\colon x\mapsto x\cdot\sigma$.

  Given a point $z\in (G^n)^\ell$, let $\rn{X}^z$ be the random variable with values in $G^\ell$ defined by
  \begin{align*}
    \rn{X}^z_j & \coloneqq (z_j)_{\rn{i}} \qquad (j\in [\ell]),
  \end{align*}
  where $\rn{i}$ is picked uniformly at random in $[n]$, that is, $\rn{X}^z_j$ is the value of the $j$th word
  $z_j$ at the (uniformly at random) position $\rn{i}$. Note that we use the same $\rn{i}$ for all values of
  $j\in[\ell]$, so the coordinates of $\rn{X}$ are not necessarily independent.

  For every $k\in K^\ell$, let also $\rn{Y}^z_k\coloneqq \sum_{j=1}^\ell k_j \rn{X}^z_j$ and note that for every $g\in
  G$, we have
  \begin{align*}
    \PP[\rn{Y}^z_k = g]
    & =
    \frac{
      \left\lvert\left\{
      i\in [n]
      \;\middle\vert\;
      \sum_{j=1}^\ell k_j (z_j)_i = g
      \right\}\right\rvert
    }{
      n
    }
    =
    \frac{f_{\HH_n^*(G),K^n}(z)(k)}{n}.
  \end{align*}

  By \cref{lem:distlincomb}, the distribution of $\rn{Y}^z$ completely determines the distribution of
  $\rn{X}^z$. This means that if $x,y\in (G^n)^\ell$ are such that $f_{\HH_n^*(G),K^\ell}(x) =
  f_{\HH_n^*(G),K^\ell}(y)$, then $\rn{X}^x$ has the same distribution as $\rn{X}^y$ and thus there exists a
  permutation $\sigma\in S_n$ such that for every $i\in[n]$ and every $j\in[\ell]$, we have
  $(x_j)_{\sigma(i)} = (y_j)_i$, that is, for the automorphism $F_\sigma\in\Aut_K(\HH_n^*(G))$, we have
  $F_\sigma(x_j) = y_j$ for every $j\in[\ell]$, so $f_{\HH_n^*(G),K^\ell}$ factors through types.
\end{proof}

The next example shows that the simplicity assumption in \cref{prop:strongHammingftt} is necessary.

\begin{example}
  Consider the finite left $\FF_2$-module $\FF_2^2$, let $n\coloneqq 4$ and $\ell\coloneqq 2$ and consider the
  elements $x,y\in((\FF_2^2)^n)^\ell$ defined as
  \begin{align*}
    y_1\coloneqq x_1 & \coloneqq \bigl((0,0),(0,1),(1,0),(1,1)\bigr),\\
    x_2 & \coloneqq \bigl((0,1),(1,0),(0,0),(1,1)\bigr),\\
    y_2 & \coloneqq \bigl((1,0),(0,1),(1,1),(0,0)\bigr)
    \intertext{and note that}
    x_1 + x_2 & = \bigl((0,1),(1,1),(1,0),(0,0)\bigr),\\
    y_1 + y_2 & = \bigl((1,0),(0,0),(0,1),(1,1)\bigr),
  \end{align*}

  If we think of the function $f_{\HH_4^*(\FF_2^2)}$ as taking values in $\{0,1,2,3,4\}^{\FF_2^2}$ as we did
  in the proof of \cref{prop:strongHammingftt} and the function $f_{\HH_4^*(\FF_2^2),\FF_2^2}$ as
  taking values in $(\{0,1,2,3,4\}^{\FF_2^2})^{\FF_2^2}$, then for every $g\in\FF_2^2$, we have
  \begin{align*}
    f_{\HH_4^*(\FF_2^2),\FF_2^2}(x)(0,0)(g) = f_{\HH_4^*(\FF_2^2),\FF_2^2}(y)(0,0)(g) & = 4\cdot\One[g = 0],\\
    f_{\HH_4^*(\FF_2^2),\FF_2^2}(x)(1,0)(g) = f_{\HH_4^*(\FF_2^2),\FF_2^2}(y)(1,0)(g) & = 1,\\
    f_{\HH_4^*(\FF_2^2),\FF_2^2}(x)(0,1)(g) = f_{\HH_4^*(\FF_2^2),\FF_2^2}(y)(1,0)(g) & = 1,\\
    f_{\HH_4^*(\FF_2^2),\FF_2^2}(x)(1,1)(g) = f_{\HH_4^*(\FF_2^2),\FF_2^2}(y)(1,0)(g) & = 1.
  \end{align*}

  However, no $\sigma\in\Aut_{\FF_2}(\HH_4^*(\FF_2^2))$ satisfies $\sigma(x_1) = y_1$ and $\sigma(x_2) =
  y_2$. This is because $\sigma(x_1)=y_1$ implies that $\sigma=\id_{(\FF_2^2)^4}$ (as
  $\Aut_{\FF_2}(\HH_4^*(\FF_2^2))$ is precisely given by the natural right action of the symmetric group $S_4$
  on $(\FF_2^2)^4$ by $(g\cdot\sigma)_i\coloneqq g_{\sigma(i)}$ ($g\in(\FF_2^2)^4$, $\sigma\in S_4$ and $i\in[4]$)).
\end{example}

The next example shows that it is not enough to consider the set $\{0,1\}^\ell\subseteq K^\ell$ that captures
information only about subset sums of the tuples of words.

\begin{example}
  Consider the finite simple left $\FF_3$-module $\FF_3$, let $n\coloneqq 3$ and $\ell\coloneqq 2$ and consider the
  elements $x,y\in(\FF_3^n)^\ell$ defined as
  \begin{align*}
    y_1\coloneqq x_1 \coloneqq  x_2 & \coloneqq (0,1,2),\\
    y_2 & \coloneqq (2,0,1),
    \intertext{and note that}
    x_1 + x_2 & = (0,2,1),\\
    y_1 + y_2 & = (2,1,0).
  \end{align*}

  Again, thinking of the function $f_{\HH_3^*(\FF_3)}$ as taking values in $\{0,1,2,3\}^{\FF_3}$ as we did in
  the proof of \cref{prop:strongHammingftt} and the function $f_{\HH_3^*(\FF_3),\{0,1\}^2}$ as taking
  values in $(\{0,1,2,3\}^{\FF_3})^{\{0,1\}^2}$, then for every $g\in\FF_3$, we have
  \begin{align*}
    f_{\HH_3^*(\FF_3),\{0,1\}^2}(x)(0,0)(g) = f_{\HH_3^*(\FF_3),\{0,1\}^2}(y)(0,0)(g) & = 3\cdot\One[g = 0],\\
    f_{\HH_3^*(\FF_3),\{0,1\}^2}(x)(0,1)(g) = f_{\HH_3^*(\FF_3),\{0,1\}^2}(y)(0,1)(g) & = 1,\\
    f_{\HH_3^*(\FF_3),\{0,1\}^2}(x)(1,0)(g) = f_{\HH_3^*(\FF_3),\{0,1\}^2}(y)(1,0)(g) & = 1,\\
    f_{\HH_3^*(\FF_3),\{0,1\}^2}(x)(1,1)(g) = f_{\HH_3^*(\FF_3),\{0,1\}^2}(y)(1,1)(g) & = 1.
  \end{align*}

  However, no $\sigma\in\Aut_{\FF_3}(\HH_3^*(\FF_3))$ satisfies $\sigma(x_1) = y_1$ and $\sigma(x_2) = y_2$
  because the former implies $\sigma=\id_{\FF_3^3}$.
\end{example}

\begin{corollary}\label{cor:highLPboundstrong}
  Let $\FF$ be a finite field and let $C$ be an $\FF$-linear $D$-code in the strong Hamming scheme
  $\HH_n^*(\FF)$. Then for every $\ell\in\NN_+$, we have
  \begin{align*}
    \lvert C\rvert & \leq \val(\cL_{\HH_n^*(\FF)^{\ell,\FF^\ell}}(D^{\ell,\FF^\ell}))^{1/\ell}.
  \end{align*}
\end{corollary}

\begin{proof}
  Since $\FF$ is a simple $\FF$-module, by \cref{thm:ftt} and \cref{prop:strongHammingftt},
  $C^\ell$ is a $K$-linear $D^{\ell,\FF^\ell}$-code in $\HH_n^*(\FF)^{\ell,\FF^\ell}$ and thus we have the bound
  \begin{align*}
    \lvert C^\ell\rvert & \leq \val(\cL_{\HH_n^*(\FF)^{\ell,\FF^\ell}}(D^{\ell,\FF^\ell}))
  \end{align*}
  provided by the Delsarte linear program for $\HH_n^*(\FF)^{\ell,\FF^\ell}$.
\end{proof}

For the case of the weak Hamming scheme, annihilators of single elements will play an important role and it
will be more convenient to work with annihilator Hamming schemes defined below. We will show that the
annihilator Hamming scheme is indeed
a symmetric translation scheme in \cref{prop:annihilatorHamming} and the connection to the weak
Hamming scheme will be established in \cref{lem:annweakHamming}.

\begin{definition}
  Given a ring $K$, a finite simple left $K$-module $G$ and $n\in\NN_+$, the \emph{annihilator Hamming scheme
    of order $n$ over $G$} is the symmetric translation scheme $\HH_n^{\Ann_K}(G) \coloneqq (G^n,R)$ (with $G^n$
  equipped with the direct product left $K$-module structure), where
  \begin{align*}
    R
    & \coloneqq
    \{r_h \mid h\colon 2^K \to \{0,1,\ldots,n\}\}\setminus\{\varnothing\},
    \\
    r_h
    & \coloneqq
    \{(x,y)\in G^n\times G^n \mid \forall A\subseteq K, \lvert x-y\rvert_A = h(A)\}
    \qquad (h\colon 2^K \to \{0,1,\ldots,n\}),
    \\
    \lvert z\rvert_A & \coloneqq \lvert\{i\in[n] \mid \Ann_K(z) = A\}\rvert,
  \end{align*}
  that is, $(x,y)$ is in the relation $r_h$ if and only if for each set $A\subseteq K$, the number of
  positions $i\in[n]$ such that $\Ann_K(x_i-y_i) = A$ is exactly $h(A)$.
\end{definition}

Before we actually prove that $\HH_n^{\Ann_K}(G)$ is indeed a symmetric translation scheme, let us prove the
following small lemma that says that when $K$ is commutative, then the annihilator and the weak Hamming
schemes coincide.

\begin{lemma}\label{lem:annweakHamming}
  Let $K$ be a ring, let $G$ be a finite simple left $K$-module and let $n\in\NN_+$. If $K$ is commutative,
  then $\HH_n(G) = \HH_n^{\Ann_K}(G)$.
\end{lemma}

\begin{proof}
  Since $K$ is commutative, for every $g\in K\setminus\{0\}$ we have
  \begin{align*}
    \Ann_K(g) & = \Ann_K(Kg) = \Ann_K(G)\neq K.
  \end{align*}
  Since the relations of $\HH_n^{\Ann_K}(G)$ are based on the values of
  \begin{align*}
    \lvert z\rvert_A & \coloneqq \lvert\{i\in[n] \mid \Ann_K(z) = A\}\rvert
  \end{align*}
  we get
  \begin{align*}
    \lvert z\rvert_A
    & =
    \begin{dcases*}
      \lvert\{i\in[n] \mid z_i\neq 0\}\rvert, & if $A = \Ann_K(G)$,\\
      \lvert\{i\in[n] \mid z_i = 0\}\rvert, & if $A = K$,\\
      0, & otherwise.
    \end{dcases*}
  \end{align*}
  Thus it follows trivially that $\HH_n^{\Ann_K}(G) = \HH_n(G)$.
\end{proof}

Our next order of business is to show that $\HH_n^{\Ann_K}(G)$ is indeed a symmetric translation scheme (even
when $K$ is not necessarily commutative). To do so, we need one lemma that says that in a finite simple left
$K$-module $G$, the orbits of $G$ under the natural action of the group $\Aut_K(G)$ of $K$-module
automorphisms of $G$ are completely determined by the annihilators of its elements. We in fact will prove a
more general version over $G^\ell$ that will be needed later.

\begin{lemma}\label{lem:AnnAut}
  Let $K$ be a ring, let $G$ be a simple left $K$-module and let $\ell\in\NN_+$. Let us also equip $G^\ell$
  with the natural left $K^\ell$-module structure over the direct product ring $K^\ell$ given by
  $(kg)_i\coloneqq k_i g_i$ ($k\in K^\ell$, $g\in G^\ell$ and $i\in[\ell]$). The following are equivalent for
  $x,y\in G^\ell$.
  \begin{enumerate}
  \item We have $\Ann_{K^\ell}(x) = \Ann_{K^\ell}(y)$.
    \label{lem:AnnAut:Ann}
  \item There exists a left $K$-module automorphism $\sigma\in\Aut_K(G)$ of $G$ such that $\sigma(x_i) = y_i$
    for every $i\in[\ell]$.
    \label{lem:AnnAut:Aut}
  \end{enumerate}
\end{lemma}

\begin{proof}
  For the implication~\ref{lem:AnnAut:Aut}$\implies$\ref{lem:AnnAut:Ann}, note that for $k\in K^\ell$, we have
  the equivalences
  \begin{align*}
    k y = 0
    & \iff
    (\forall i\in[\ell], k_i y_i = 0)
    \iff
    (\forall i\in[\ell], k_i\sigma(x_i) = 0)
    \iff
    (\forall i\in[\ell], k_i x_i = 0)
    \iff
    k x = 0,
  \end{align*}
  where the third equivalence follows since $\sigma\in\Aut_K(G)$ is a left $K$-module automorphism of $G$. Thus
  $\Ann_{K^\ell}(x) = \Ann_{K^\ell}(y)$.

  \medskip

  Let us now prove the implication~\ref{lem:AnnAut:Ann}$\implies$\ref{lem:AnnAut:Aut}. If $x = 0$, then
  $K^\ell = \Ann_{K^\ell}(x) = \Ann_{K^\ell}(y)$. Since $G$ is simple, we must have $y = 0$ (as $\{z\in G \mid
  \Ann_K(z)=K\}$ is a proper left $K$-submodule of $G$, so it must be trivial) and any left $K$-module
  automorphism $\sigma\in\Aut_K(G)$ of $G$ satisfies $\sigma(x_i) = y_i$ for every $i\in[\ell]$. Suppose then
  that $x\neq 0$ and without loss of generality, suppose that $x_1\neq 0$ and thus $\Ann_K(x_1)\neq K$. Since
  \begin{align*}
    \Ann_K(x_1)
    & =
    \{k\in K \mid (k,0,\ldots,0)\in\Ann_{K^\ell}(x)\}
    =
    \{k\in K \mid (k,0,\ldots,0)\in\Ann_{K^\ell}(y)\}
    =
    \Ann_K(y_1),
  \end{align*}
  it follows that $\Ann_K(y_1)\neq K$ so $y_1\neq 0$.

  Since $G$ is simple, we have $K x_1 = G$, so we can define a function $\sigma\colon G\to G$ indirectly by
  $\sigma(k x_1)\coloneqq k y_1$ for every $k\in K$. To check that $\sigma$ is well-defined, note that if $k_1 x_1
  = k_2 x_1$, then $k_1 - k_2\in\Ann_K(x_1) = \Ann_K(y_1)$ and thus $k_1 y_1 = k_2 y_2$. It is straightforward
  to check that $\sigma$ is a left $K$-module endomorphism of $G$. Since the kernel of $\sigma$ is a left
  $K$-submodule of $G$ that does not contain $x_1$ and $G$ is simple, it follows that the kernel of $\sigma$
  must be trivial, so $\sigma$ is injective. On the other hand, the image of $\sigma$ is a left $K$-submodule
  of $G$ that contains $y_1\neq 0$, so simplicity of $G$ implies that $\sigma$ is surjective and thus $\sigma$
  is a left $K$-module automorphism of $G$.

  Let us now show that $\sigma(x_i) = y_i$ for every $i\in[\ell]$. For $i = 1$, this is obvious. For $i\geq
  2$, let $k'\in K$ be such that $k' x_1 = x_i$ so that $\sigma(x_i) = k' y_1$. We now let $k\in K^\ell$ be
  given by $k_1\coloneqq k'$, $k_i\coloneqq -1$ and $k_j\coloneqq 0$ for every
  $j\in[\ell]\setminus\{1,i\}$. Since $k'x_1 - x_i = 0$, we have $k\in\Ann_{K^\el}(x) = \Ann_{K^\ell}(y)$, so
  we get $k'y_1 - y_i = 0$, and thus $\sigma(x_i) = k'y_1 = y_i$ as desired.
\end{proof}

Let us now prove that $\HH_n^{\Ann_K}(G)$ is indeed a symmetric translation scheme. The proof uses similar
ideas to that of \cref{thm:ftt}.

\begin{proposition}\label{prop:annihilatorHamming}
  Let $K$ be a ring, let $G$ be a finite simple left $K$-module and let $n\in\NN_+$. Then the annihilator
  Hamming scheme $\HH_n^{\Ann_K}(G)$ of order $n$ over $G$ is a symmetric translation scheme over the direct
  group $G^n$.
\end{proposition}

\begin{proof}
  Recall that the relation set of $\HH_n^{\Ann_K(G)}$ is given by
  \begin{align*}
    R
    & \coloneqq
    \{r_h \mid h\colon 2^K \to \{0,1,\ldots,n\}\}\setminus\{\varnothing\},
  \end{align*}
  where  
  \begin{align*}
    r_h
    & \coloneqq
    \{(x,y)\in G^n\times G^n \mid \forall A\subseteq K, \lvert x-y\rvert_A = h(A)\}
    \qquad (h\colon 2^K \to \{0,1,\ldots,n\}),
    \\
    \lvert z\rvert_A & \coloneqq \lvert\{i\in[n] \mid \Ann_K(z) = A\}\rvert.
  \end{align*}

  The fact that $R$ forms a partition of $G^n\times G^n$ into non-empty subsets is obvious.

  Since $G$ is simple, the only element $g\in G$ with $\Ann_K(g) = K$ is $g = 0$ (as the set of such elements
  is a proper left $K$-submodule of $G$, so it must be trivial), thus for the function $h\colon
  2^K\to\{0,1,\ldots,n\}$ given by $h(A)\coloneqq n\One[A=K]$, we have $r_h = \cD_{G^n}$.

  Note further that $\Ann_K(z) = \Ann_K(-z)$ for every $z\in G$, which immediately implies that $r_h^\top =
  r_h$ for every $h\colon 2^K\to\{0,1,\ldots,n\}$.

  It is also obvious that each $r_h$ is invariant under the group action of $G^n$.

  It remains only to show the existence of the intersection numbers for $\HH_n^{\Ann_K(G)}$.

  For every $A_1,A_2\subseteq K$ and every $g_1,g_2\in G$, let
  \begin{align*}
    N_{A_1,A_2}(g_1,g_2) & \coloneqq \lvert\{z\in G \mid \Ann_K(g_1 - z) = A_1\land\Ann_K(z - g_2)=A_2\}\rvert.
  \end{align*}

  \begin{claim}\label{clm:annposition}
    If $\Ann_K(g_1-g_2) = \Ann_K(g_1'-g_2')$, then $N_{A_1,A_2}(g_1,g_2) = N_{A_1,A_2}(g_1',g_2')$.
  \end{claim}

  \begin{proof}
    By \cref{lem:AnnAut}, there exists a left $K$-module automorphism $\sigma\in\Aut_K(G)$ of $G$ such that
    $\sigma(g_1-g_2) = g_1'-g_2'$. Then we have
    \begin{align*}
      N_{A_1,A_2}(g_1'-g_2')
      & =
      \lvert\{z\in G \mid \Ann_K(g_1' - z) = A_1\land\Ann_K(z - g_2')=A_2\}\rvert
      \\
      & =
      \lvert\{u\in G \mid \Ann_K(g_1' - g_2' - u) = A_1\land\Ann_K(u)=A_2\}\rvert
      \\
      & =
      \lvert\{u\in G\mid \Ann_K(\sigma(g_1-g_2) - u)=A_1\land\Ann_K(u)=A_2\}\rvert
      \\
      & =
      \lvert\{u\in G\mid \Ann_K(g_1-g_2 - \sigma^{-1}(u))=A_1\land\Ann_K(\sigma^{-1}(u))=A_2\}\rvert
      \\
      & =
      \lvert\{w\in G\mid \Ann_K(g_1 - w)=A_1\land\Ann_K(w - g_2)=A_2\}\rvert,
    \end{align*}
    where the second equality follows from the substitution $u\coloneqq z - g_2'$, the fourth equality follows since
    $\sigma\in\Aut_K(G)$ is a left $K$-module automorphism and the fifth equality follows from the
    substitution $w\coloneqq\sigma^{-1}(u) + g_2$.
  \end{proof}

  \cref{clm:annposition} implies that for $A_1,A_2,B\subseteq K$ we can define $N_{A_1,A_2}^B\in\NN$ such that
  $N_{A_1,A_2}(g_1,g_2) = N_{A_1,A_2}^B$ whenever $\Ann_K(g_1-g_2)=B$.

  Note now that if $(x,y)\in r_h$ for some $h\colon 2^K\to\{0,1,\ldots,n\}$ and $h_1,h_2\colon
  2^K\to\{0,1,\ldots,n\}$, then we have
  \begin{align*}
    & \!\!\!\!\!\!
    \lvert\{z\in G^n \mid (x,z)\in r_{h_1}\land (z,y)\in r_{h_2}\}\rvert
    \\
    & =
    \sum_{F\in\cF} \prod_{i=1}^n \lvert\{w\in G \mid \Ann_K(x_i-w) = F(i)_1 \land \Ann_K(w-y_i)=F(i)_2\}\rvert,
  \end{align*}
  where $\cF$ is the set of functions $F\colon [n]\to 2^K\times 2^K$ such that
  \begin{align*}
    \lvert\{i\in[n] \mid F(i)_j = A\}\rvert = h_j(A) \qquad (j\in[2], A\subseteq K).
  \end{align*}

  Using the definition of the numbers $N_{A_1,A_2}^B$, we get
  \begin{align*}
    \sum_{F\in\cF} \prod_{i=1}^n \lvert\{w\in G \mid \Ann_K(x_i-w) = F(i)_1 \land \Ann_K(w-y_i)=F(i)_2\}\rvert
    & =
    \sum_{F\in\cF} \prod_{i=1}^n N_{F(i)_1,F(i)_2}^{\Ann_K(x_i-y_i)}.
  \end{align*}

  Finally, note that if $(x',y')\in r_h$, then there exists a permutation $\sigma\in S_n$ such that
  $\Ann_K(x'_i-y'_i) = \Ann_K(x_{\sigma(i)} - y_{\sigma(i)})$ for every $i\in[n]$, which implies that
  \begin{align*}
    \sum_{F\in\cF} \prod_{i=1}^n N_{F(i)_1,F(i)_2}^{\Ann_K(x'_i-y'_i)}
    & =
    \sum_{F\in\cF} \prod_{i=1}^n N_{F(\sigma^{-1}(i))_1,F(\sigma^{-1}(i))_2}^{\Ann_K(x_i-y_i)}
    =
    \sum_{F'\in\cF} \prod_{i=1}^n N_{F'(i)_1,F'(i)_2}^{\Ann_K(x_i-y_i)},
  \end{align*}
  where the last equality follows from the substitution $F'\coloneqq F\comp\sigma^{-1}$. Thus the existence of
  the intersection numbers is proved.
\end{proof}

It will be very convenient to work with the following equivalence relation that can be seen as the equivalence
relation of the ``projective space of $G^\ell$ with origin''.

\begin{definition}
  Let $K$ be a ring, let $G$ be a finite simple left $K$-module and $\ell\in\NN_+$. The equivalence relation
  $\sim_\ell$ over $G^\ell$ is defined by
  \begin{align*}
    x\sim_\ell y & \iff \exists\sigma\in\Aut_K(G),\forall i\in[\ell], \sigma(x_i) = y_i,
  \end{align*}
  that is, the equivalence classes of $\sim_\ell$ are the orbits of the natural diagonal action of $\Aut_K(G)$
  on $G^\ell$ given by $(\sigma\cdot x)_i \coloneqq \sigma(x_i)$ ($\sigma\in\Aut_K(G)$, $x\in G^\ell$,
  $i\in[\ell]$).
\end{definition}
In the definition above, if $K=\FF$ for a finite field $\FF$ (hence $G$ is a $1$-dimensional $\FF$-vector
space), then $x\sim_\ell y$ if and only if there exist $k_1,k_2\in\FF\setminus\{0\}$ such that $k_1x = y$ and
$x = k_2y$, that is, $\sim_\ell$ is the equivalence relation defining the $(\ell-1)$-dimensional projective
space $P(\FF^\ell)\cup\{0\}$ with origin.

\begin{remark}\label{rmk:AnnAut}
  Under the definition of $\sim_\ell$, we can reinterpret \cref{lem:AnnAut}, as saying that $x\sim_\ell
  y$ if and only if $\Ann_{K^\ell}(x) = \Ann_{K^\ell}(y)$.
\end{remark}

The next lemma is an analogue of \cref{lem:distlincomb} that says that the distribution of the
$\sim_\ell$-equivalence class of an $\ell$-tuple of random variables with values in a finite simple left
$K$-module (for a finite ring $K$) can be recovered from the (individual) distributions of the
$\sim_1$-equivalence classes of all $K$-linear combinations of them.

\begin{lemma}\label{lem:simelldistlincomb}
  Let $K$ be a finite ring, let $G$ be a finite simple left $K$-module, let $\rn{X}$ be a random variable with
  values in $G^\ell$ for some $\ell\in\NN_+$ and for every $k\in K^\ell$, let $\rn{Y}_k\coloneqq\sum_{i=1}^\ell
  k_i\rn{X}_i$. Then
  {\small%
    \begin{align*}
      \PP[\rn{X}\sim_\ell x]
      & =
      \begin{multlined}[t]
        \left(
        \frac{\lvert G\rvert}{
          \lvert\Stab_{\Aut_K(G)}(x)\rvert\cdot\lvert K^\ell\rvert\cdot (\lvert G\rvert-1)
        }
        \cdot
        \sum_{k\in K^\ell}
        \left\lvert
        \Stab_{\Aut_K(G)}\left(\sum_{i=1}^\ell k_i x_i\right)
        \right\rvert
        \cdot
        \PP\left[
          \rn{Y}_k \sim_1 \sum_{i=1}^\ell k_i x_i
          \right]
        \right)
        \\
        -
        \frac{\lvert O_{\Aut_K(G)}(x)\rvert}{\lvert G\rvert - 1}
      \end{multlined}
    \end{align*}%
  }%
  for every $x\in G^\ell$, where $\Stab_{\Aut_K(G)}(z)$ is the stabilizer group of $z$ under the action of
  $\Aut_K(G)$ and $O_{\Aut_K(G)}(z)$ is the orbit of $z$ under the action of $\Aut_K(G)$ (the actions of
  $\Aut_K(G)$ on $G$ and $G^\ell$ are respectively the natural action and the diagonal action).
\end{lemma}

\begin{proof}
  By \cref{lem:AnnAut} (see also \cref{rmk:AnnAut}), we know that the $\sim_\ell$-equivalence class
  of $x$ is precisely the orbit $O_{\Aut_K(G)}(x)$, so we have
  \begin{align*}
    \PP[\rn{X}\sim_\ell x]
    & =
    \frac{1}{\lvert\Stab_{\Aut_K(G)}(x)\rvert}\sum_{\sigma\in\Aut_K(G)} \PP[\rn{X} = \sigma(x)].
  \end{align*}
  Letting $\chi$ be a non-trivial character of $G$, by \cref{lem:distlincomb}, we get
  \begin{align*}
    \PP[\rn{X}\sim_\ell x]
    & =
    \frac{1}{\lvert\Stab_{\Aut_K(G)}(x)\rvert\cdot\lvert K^\ell\rvert}
    \cdot
    \sum_{\substack{\sigma\in\Aut_K(G)\\k\in K^\ell\\y\in G}}
    \chi\left(\sum_{i=1}^\ell k_i\sigma(x_i) - y\right)
    \cdot\PP[\rn{Y}_k = y].
  \end{align*}

  Recall that for every $g\in G$, if we average the value $\chi(g)$ over all non-trivial characters of $G$, then
  we get $(\lvert G\rvert\One[g = 0] - 1)/(\lvert G\rvert-1)$, thus by performing such averaging operation in
  the above, we get
  \begin{align}\label{eq:simell}
    \PP[\rn{X}\sim_\ell x]
    & =
    \begin{multlined}[t]
      \frac{1}{\lvert\Stab_{\Aut_K(G)}(x)\rvert\cdot\lvert K^\ell\rvert\cdot(\lvert G\rvert - 1)}
      \\
      \cdot
      \sum_{\substack{\sigma\in\Aut_K(G)\\k\in K^\ell\\y\in G}}
      \left(\lvert G\rvert\One\left[\sum_{i=1}^\ell k_i\sigma(x_i) = y\right] - 1\right)
      \cdot\PP[\rn{Y}_k = y].
    \end{multlined}
  \end{align}

  Note now that
  \begin{equation}\label{eq:simellfirst}
    \begin{aligned}
      \sum_{\substack{\sigma\in\Aut_K(G)\\k\in K^\ell\\y\in G}}
      \lvert G\rvert\One\left[\sum_{i=1}^\ell k_i\sigma(x_i) = y\right]
      \cdot\PP[\rn{Y}_k = y]
      & =
      \sum_{\substack{\sigma\in\Aut_K(G)\\k\in K^\ell}}
      \lvert G\rvert
      \cdot\PP\left[\rn{Y}_k = \sigma^{-1}\left(\sum_{i=1}^\ell k_i x_i\right)\right]
      \\
      & =
      \lvert G\rvert
      \sum_{k\in K^\ell}
      \left\lvert\Stab_{\Aut_K(G)}\left(\sum_{i=1}^\ell k_i x_i\right)\right\rvert
      \cdot\PP\left[\rn{Y}_k\sim_1\sum_{i=1}^\ell k_i x_i\right],
    \end{aligned}
  \end{equation}
  where the last equality follows since \cref{lem:AnnAut} and \cref{rmk:AnnAut} imply that the
  $\sim_1$-equivalence class of $\sum_{i=1}^\ell k_i x_i$ is precisely its $\Aut_K(G)$-orbit.

  On the other hand, we have
  \begin{align}\label{eq:simellsecond}
    \sum_{\substack{\sigma\in\Aut_K(G)\\k\in K^\ell\\y\in G}}
    \PP[\rn{Y}_k = y]
    & =
    \sum_{\substack{\sigma\in\Aut_K(G)\\k\in K^\ell}} 1
    =
    \lvert\Aut_K(G)\rvert\cdot\lvert K^\ell\rvert
    =
    \lvert O_{\Aut_K(G)}(x)\rvert\cdot\lvert\Stab_{\Aut_K(G)}(x)\rvert\cdot\lvert K^\ell\rvert.
  \end{align}

  The result now follows by putting together~\eqref{eq:simell}, \eqref{eq:simellfirst}
  and~\eqref{eq:simellsecond}.
\end{proof}

We can finally prove that for the annihilator Hamming scheme $\HH_n^{\Ann_K}(G)$, the associated function
$f_{\HH_n^{\Ann_K}(G),K^\ell}$ factors through types.

\begin{proposition}\label{prop:annihilatorHammingftt}
  Let $K$ be a ring, let $G$ be a finite simple left $K$-module and let $n,\ell\in\NN_+$. Consider the
  annihilator Hamming scheme $\HH_n^{\Ann_K}(G)$ of order $n$ over $G$. Then $f_{\HH_n^{\Ann_K}(G),K^\ell}$
  factors through types of $\HH_n^{\Ann_K}(G)$.
\end{proposition}

\begin{proof}
  Recall that the relation set of $\HH_n^{\Ann_K(G)}$ is given by
  \begin{align*}
    R
    & \coloneqq
    \{r_h \mid h\colon 2^K \to \{0,1,\ldots,n\}\}\setminus\{\varnothing\},
    \intertext{where}
    r_h
    & \coloneqq
    \{(x,y)\in G^n\times G^n \mid \forall A\subseteq K, \lvert x-y\rvert_A = h(A)\}
    \qquad (h\colon 2^K \to \{0,1,\ldots,n\}),
    \\
    \lvert z\rvert_A & \coloneqq \lvert\{i\in[n] \mid \Ann_K(z) = A\}\rvert.
  \end{align*}
  
  Just as in the proof of \cref{prop:strongHammingftt}, we will abuse notation and write
  $f_{\HH_n^{\Ann_K}(G)}(x) = h$ in place of $f_{\HH_n^{\Ann_K}(G)}(x) = r_h$, thus viewing
  $f_{\HH_n^{\Ann_K}}$ as a $\{0,1,\ldots,n\}^{2^K}$-valued function. Accordingly, we will also view
  $f_{\HH_n^{\Ann_K}(G),K^\ell}$ as a function with values in $(\{0,1,\ldots,n\}^{2^K})^{K^\ell}$ rather than
  in $R^{K^\ell}$.

  \medskip

  By the same argument as in the proof of \cref{prop:strongHammingftt}, it is enough to prove the
  case when $K$ is finite: the arbitrary case can be reduced to the finite case by letting $K'\coloneqq
  K/\Ann_K(G)$, considering the natural $K'$-module structure on $G$ and noting that
  $f_{\HH_n^{\Ann_K}(G),K^\ell}$ factors through types of $\HH_n^{\Ann_K}(G)$ if and only if
  $f_{\HH_n^{\Ann_{K'}}(G),(K')^\ell}$ factors through types of $\HH_n^{\Ann_{K'}}(G)$ and $K'$ is finite by
  \cref{lem:finitefaithful}.

  \medskip

  Let us then prove the case when $K$ is finite.

  First, we claim that $\Aut_K(\HH_n^{\Ann_K}(G))$ contains a subgroup isomorphic to a semidirect
  product\footnote{In fact, $\Aut_K(\HH_n^{\Ann_K}(G))$ is exactly equal to this semidirect product, but we
    will not need this result.}  $\Aut_K(G)^n\rtimes S_n$, where $S_n$ is the symmetric group on
  $[n]$. Indeed, consider the natural actions of $\Aut_K(G)$ and $S_n$ on $G^n$ given by
  \begin{align*}
    (F\cdot g)_i & \coloneqq F_i(g_i), &
    (g\cdot\sigma)_i & \coloneqq g_{\sigma(i)}
  \end{align*}
  for $F\in\Aut_K(G)^n$, $g\in G^n$, $\sigma\in S_n$ and $i\in[n]$. It is obvious that these actions are free
  and preserve the $K$-module structure of $G^n$, and thus induce subgroups of the $K$-module automorphism
  group of $G^n$ isomorphic to $\Aut_K(G)^n$ and $S_n$, respectively. Let $H$ be the product of these
  subgroups. It is straightforward to check that if $F\cdot g = g\cdot\sigma$ holds for every $g\in G^n$, then
  $F_i=\id_G^n$ for every $i\in[n]$ and $\sigma=\id_n$, so these subgroups have trivial intersection. Since
  \begin{align*}
    \Bigl(\bigl(F\cdot (g\cdot\sigma)\bigr)\cdot\sigma^{-1}\Bigr)_i
    & =
    \bigl(F\cdot (g\cdot\sigma)\bigr)_{\sigma^{-1}(i)}
    =
    F_{\sigma^{-1}(i)}\bigl((g\cdot\sigma)_{\sigma^{-1}(i)}\bigr)
    =
    F_{\sigma^{-1}(i)}(g_i),
  \end{align*}
  it follows that the subgroup isomorphic to $\Aut_K(G)^n$ is normal in $H$ and thus $H\cong\Aut_K(G)^n\rtimes
  S_n$. It remains to show that $H$ also preserves the association scheme structure. Indeed, note that for
  $F\in\Aut_K(G)^n$, $g\in G^n$, $\sigma\in S_n$ and $A\subseteq K$, we have
  \begin{align*}
    \lvert F(g)\rvert_A
    & =
    \lvert\{i\in[n] \mid \Ann_K(F_i(g_i)) = A\}\rvert
    =
    \lvert\{i\in [n] \mid \Ann_K(g_i) = A\}\rvert
    =
    \lvert g\rvert_A
    \\
    \lvert g\cdot\sigma\rvert_A
    & =
    \lvert\{i\in [n]\mid \Ann_K(g_{\sigma(i)})=A\}\rvert
    =
    \lvert\{i\in [n] \mid \Ann_K(g_i)=A\}\rvert
    =
    \lvert g\rvert_A.
  \end{align*}
  Thus $f_{\HH_n^{\Ann_K}(G)}$ is invariant under both actions of the groups $\Aut_K(G)^n$ and $S_n$, so $H$ is a
  subgroup of $\Aut_K(\HH_n^{\Ann_K}(G))$.

  \medskip

  We now define the same random variables as in \cref{prop:strongHammingftt}: given a point $z\in
  (G^n)^\ell$, let $\rn{X}^z$ be the random variable with values in $G^\ell$ defined by
  \begin{align*}
    \rn{X}^z_j & \coloneqq (z_j)_{\rn{i}} \qquad (j\in [\ell]),
  \end{align*}
  where $\rn{i}$ is picked uniformly at random in $[n]$ and for every $k\in K^\ell$, let $\rn{Y}^z\coloneqq
  \sum_{j=1}^\ell k_j\rn{X}^z_j$.

  Note that for every $g\in G$, we have
  \begin{align*}
    \PP[\rn{Y}^z_k \sim_1 g]
    & =
    \PP[\Ann_K(\rn{Y}^z_k) = \Ann_K(g)]
    \\
    & =
    \frac{
      \left\lvert\left\{
      i\in[n]
      \;\middle\vert\;
      \Ann_K\left(\sum_{j=1}^\ell k_j (z_j)_i\right) = \Ann_K(g)
      \right\}\right\rvert
    }{
      n
    }
    \\
    & =
    \frac{f_{\HH_n^{\Ann_K}(G),K^\ell}(z)(k)(\Ann_K(g))}{n},
  \end{align*}
  where the first equality follows from \cref{lem:AnnAut} and \cref{rmk:AnnAut}.

  By \cref{lem:simelldistlincomb}, the individual distributions of the $\sim_1$-equivalence classes of
  the $\rn{Y}^z_k$ ($k\in K^\ell$) completely determine the distribution of the $\sim_\ell$-equivalence
  class of $\rn{X}^z$. This means that if $x,y\in (G^n)^\ell$ are such that $f_{\HH_n^{\Ann_K}(G),K^\ell}(x) =
  f_{\HH_n^{\Ann_K}(G),K^\ell}(y)$, then for every $w\in G^\ell$, we have
  \begin{align*}
    \PP[\rn{X}^x\sim_\ell w] & = \PP[\rn{X}^y\sim_\ell w].
  \end{align*}
  Thus there exists a permutation $\sigma\in S_n$ such that for every $i\in[n]$, we have
  \begin{align*}
    \bigl((x_1)_{\sigma(i)},(x_2)_{\sigma(i)},\ldots,(x_\ell)_{\sigma(i)}\bigr)
    & \sim_\ell
    \bigl((y_1)_i,(y_2)_i,\ldots,(y_\ell)_i\bigr),
  \end{align*}
  so by the definition of $\sim_\ell$, for each $i\in[n]$, there exists a left $K$-module automorphism
  $F_i\in\Aut_K(G)$ such that
  \begin{align*}
    F_i((x_j)_{\sigma(i)}) & = (y_j)_i
  \end{align*}
  for every $j\in[\ell]$ and thus for $F = (F_1,\ldots,F_n)\in\Aut_K(G)^n$, we get
  \begin{align*}
    F\cdot (x_j\cdot\sigma) & = y_j
  \end{align*}
  for every $j\in[\ell]$, so $f_{\HH_n^{\Ann_K}(G),K^\ell}$ factors through types.
\end{proof}

From \cref{lem:annweakHamming} and \cref{prop:annihilatorHammingftt} above, we can finally
show that under mild assumptions $f_{\HH_n(G),K^\ell}$ also factors through types for the weak Hamming scheme
$\HH_n(G)$.

\begin{proposition}\label{prop:weakHammingftt}
  Let $K$ be a ring and $G$ be a finite simple left $K$-module and $n,\ell\in\NN_+$. Consider the weak Hamming
  scheme $\HH_n(G)$ of order $n$ over $G$ equipped with the product left $K$-module structure on $G^n$.

  If $K$ is commutative, then $f_{\HH_n(G),K^\ell}$ factors through types of $\HH_n(G)$.
\end{proposition}

\begin{proof}
  Follows directly from \cref{prop:annihilatorHammingftt} as \cref{lem:annweakHamming} implies
  that the weak Hamming scheme $\HH_n(G)$ coincides with the annihilator Hamming scheme $\HH_n^{\Ann_K}(G)$.
\end{proof}

\begin{corollary}\label{cor:highLPboundweak}
  Let $\FF$ be a finite field and let $C$ be an $\FF$-linear $D$-code in the weak Hamming scheme
  $\HH_n(\FF)$. Then for every $\ell\in\NN_+$, we have
  \begin{align*}
    \lvert C\rvert & \leq \val(\cL_{\HH_n(\FF)^{\ell,\FF^\ell}}(D^{\ell,\FF^\ell}))^{1/\ell}.
  \end{align*}
\end{corollary}

\begin{proof}
  By \cref{thm:ftt} and \cref{prop:weakHammingftt}, $C^\ell$ is a $K$-linear
  $D^{\ell,\FF^\ell}$-code in $\HH_n(\FF)^{\ell,\FF^\ell}$ and thus we have the bound
  \begin{align*}
    \lvert C^\ell\rvert & \leq \val(\cL_{\HH_n(\FF)^{\ell,\FF^\ell}}(D^{\ell,\FF^\ell}))
  \end{align*}
  provided by the Delsarte linear program for $\HH_n(\FF)^{\ell,\FF^\ell}$.
\end{proof}


\section{Main Properties of the Krawtchouk Hierarchies}\label{sec:main-results}

This section presents our main results on the linear programming hierarchy. The first result is the
completeness of the higher-order linear programming hierarchies for \emph{linear} codes. The second result is
the collapse of the hierarchies for \emph{general} codes.

\subsection{Completeness for Linear Codes}\label{sec:completeness}

In this section, we show the (approximate) completeness of our linear
programming hierarchy for linear codes over a finite field $\FF$.

We will show completeness at level $O(n^2)$ via a counting argument.
The intuition is that the hierarchy is likely already complete
at level $n$ (and we conjecture this to be the case).
At level $n$, the feasible region of the LP already encodes
  $A_q^\Lin(n,d)$.  That is, since at level $n$ there is a variable for each possible basis
  of a subspace of $\F_q^n$, just writing down the distance
  constraints of $\KLP_\Lin^\F(n,d,n)$ allows one to deduce the true
  value of $A_q^\Lin(n,d)$.  Of course this property is not sufficient
  to imply that the \emph{value} of $\KLP_\Lin^\F(n,d,n)$ is correct.
  At an intuitive level, the below proof shows that at level $O(n^2)$
  the large-dimensional subspaces outweigh the small-dimensional
  subspaces enough to deduce the correct value of $A_q^\Lin(n,d)$.

\begin{theorem}[Completeness]\label{theo:main:completeness:formal}
  Let $\FF$ be a finite field, let $q\coloneqq\lvert\FF\rvert$, let $\epsilon \in (0,1)$ and let $\ell \geq
  9(n^2\ln(q) + 1)/(\ln(1+\epsilon))^2$. Then for every $d\in\{0,1,\ldots,n\}$, we have
  \begin{align*}
    \val(\KLP^\FF_\Lin(n,d,\ell))^{1/\ell} & \leq (1+\epsilon) \cdot A^\Lin_q(n,d).
  \end{align*}
\end{theorem}

Before proving this theorem, note that since $\FF$-linear codes must necessarily have size of the form $q^k$
for some $k\in\NN$, by taking $\epsilon < q - 1$, we get
\begin{align*}
  A^\Lin_q(n,d)
  & =
  q^{\floor{(\log_q \val(\KLP^\FF_\Lin(n,d,\ell)))/\ell}}
\end{align*}
whenever $\ell > 9(n^2\ln(q) + 1)/(\ln(q))^2$.

\begin{proof}
  By~\cref{rmk:klp} that $\KLP^\FF_\Lin(n,d,\ell)$ is the Delsarte linear program
  $\cL_{S^{\ell,T}}(D_d^{\ell,T})$ for $S = \HH_n(\FF_2)$, $T = \FF_2^\ell$ and
  $D_d\coloneqq\{r_0,r_d,r_{d+1},\ldots,r_n\}$. By~\cref{thm:vartheta'}, the value of this linear program
  coincides with the value of $\vartheta'$ in the associated graph $G\coloneqq G_{S^{\ell,T}}(D_d^{\ell,T})$,
  i.e., the optimum value of the semi-definite program
  \begin{align*}
    \max \quad
    & \ip{J}{M}
    \\
    \text{s.t.} \quad
    & \tr M = 1
    & &
    & & (\text{Normalization})
    \\
    & M[u,v] = 0
    & & \forall \{u,v\} \in E(G)
    & & (\text{Independent set})
    \\
    & M \succeq 0
    & &
    & & (\text{PSD-ness})
    \\
    & M[u,v] \ge 0
    & & \forall u,v \in V(G)
    & & (\text{Non-negativity}),
  \end{align*}
  where the variables is $M\in\RR^{V\times V}$ symmetric and
  \begin{align*}
    E(G)
    & \coloneqq
    \left\{\{x,y\}\in \binom{(\FF^n)^\ell}{2}
    \;\middle\vert\;
    \forall k\in\FF^\ell, \left\lvert\sum_{i=1}^\ell k_i(x_i-y_i)\right\rvert\notin [d-1]
    \right\},
  \end{align*}
  where $\lvert z\rvert\coloneqq\lvert\{j\in[n]\mid z_j\neq 0\}\rvert$ is the Hamming weight of $z$.

  Let $k_0$ be the maximum dimension of an $\FF$-linear code of distance $d$ in $\FF^n$ (that is, let
  $k_0\coloneqq\log_q A^\Lin_q(n,d)$), let $M$ be a feasible solution of the program above and let
  us provide an upper bound for the objective value $\langle J,M\rangle$. Note that symmetrizing $M$ under the
  automorphism group $\Aut(G)$ of the Cayley graph $G$ does not change the objective value $\langle
  J,M\rangle$ (and preserves all restrictions), so we may suppose that $M$ is $\Aut(G)$-invariant, which in
  particular implies that all diagonal entries of $M$ are equal and since the trace of $M$ is $1$, it follows
  that all diagonal entries of $M$ are equal to $q^{-n\ell}$. On the other hand, since $M$ is positive
  semi-definite, any $2\times 2$ principal minor of $M$ is non-negative and thus all off-diagonal entries of
  $M$ have absolute value at most $q^{-n\ell}$, that is, we have $\lVert M\rVert_\infty = q^{-n\ell}$.

  Since the objective value $\langle J,M\rangle$ is simply the sum of all entries of $M$, we can provide an
  upper bound for it by simply giving an upper bound on how many entries of $M$ are allowed to be non-zero.

  Note that for an entry $M_{xy}$ indexed by $(x,y)\in(\FF^n)^\ell\times (\FF^n)^\ell$ to be non-zero, the
  difference vectors $z_1,\ldots,z_\ell\in\FF^n$ given by $z_i\coloneqq x_i - y_i$ ($i\in[\ell]$) must span an
  $\FF$-linear subspace of dimension at most $k_0$ (as any subspace of larger dimension necessarily has
  distance smaller than $d$ and thus some $k\in\FF^\ell$ will have $\lvert\sum_{i=1}^\ell
  k_i(x_i-y_i)\rvert\in [d-1]$).

  By letting $\gamma_{n,\ell,k}$ be the number of tuples $(z_1,\ldots,z_\ell)$ that
  span a subspace of dimension $k\in\{0,1,\ldots,n\}$, since each difference $(z_1,\ldots,z_\ell)$ is realized
  as $z_i=x_i-y_i$ for exactly $q^{n\ell}$ pairs $(x,y)\in(\FF^n)^\ell\times(\FF^n)^\ell$, we get
  \begin{align*}
    \langle J,M\rangle
    & \leq
    \sum_{k=0}^{k_0} \gamma_{n,\ell,k}\cdot q^{n\ell}\cdot\lVert M\rVert_\infty
    \leq
    \sum_{k=0}^{k_0} \gamma_{n,\ell,k}.
  \end{align*}

  We claim that
  \begin{align}\label{eq:gammabound}
    \gamma_{n,\ell,k}
    & \leq
    \binom{\ell}{k}
    \cdot
    \beta_{n,k}
    \cdot
    (q^k)^{\ell-k},
  \end{align}
  where
  \begin{align*}
    \beta_{n,k}
    & \coloneqq
    \prod_{j=0}^{k-1} (q^n - q^j)
  \end{align*}
  is the number of linearly independent ordered $k$-tuples in $\FF^n$. Indeed, the upper bound
  in~\eqref{eq:gammabound} follows by picking $k$ out of the $\ell$ vectors to have a linearly independent
  ordered $k$-tuple, then picking each of the other $\ell-k$ positions to be a linear combination of
  these $k$ vectors.

  Using this bound along with $\binom{\ell}{k}\leq\ell^k$ and $\beta_{n,k}\leq q^{nk}$, we get
  \begin{align*}
    \langle J,M\rangle
    & \leq
    \sum_{k=0}^{k_0} \gamma_{n,\ell,k}
    \leq
    \sum_{k=0}^{k_0}
    \ell^k q^{nk} q^{k \ell}
    \leq
    \ell^{k_0} q^{n k_0} \left(q^{k_0\ell} + \sum_{k=0}^{k_0-1} q^{k\ell}\right)
    \\
    & =
    \ell^{k_0} q^{n k_0} \left(q^{k_0\ell} + \frac{q^{k_0\ell}-1}{q^\ell-1}\right)
    \leq
    2\ell^n q^{n^2} q^{k_0\ell}.
  \end{align*}

  Taking the $\ell$th root and recalling that $q^{k_0} = A^\Lin_q(n,d)$ we conclude that
  \begin{align*}
    \val(\KLP^\FF_\Lin(n,d,\ell))^{1/\ell} & \leq (2\ell^n q^{n^2})^{1/\ell} A^\Lin_q(n,d).
  \end{align*}
  Finally, the hypothesis $\ell\geq 9(n^2\ln(q) + 1)/(\ln(1+\epsilon))^2$ implies that $(2\ell^n
  q^{n^2})^{1/\ell}\leq 1+\epsilon$, which concludes the proof (a detailed computation is included
  in~\cref{lemma:completeness_param_comp} in \cref{sec:deferred}).
\end{proof}

\begin{remark}\label{rmk:completeness}
  The same proof of Theorem~\ref{theo:main:completeness:formal} also works for $D$-codes over the
  weak Hamming scheme $\HH_n(\FF)$ yielding
  \begin{align*}
    \val(\cL_{\HH_n(\FF)}^{\ell,\FF^\ell}(D^{\ell,\FF^\ell}))^{1/\ell} & \leq (1+\epsilon)\lvert C^*\rvert,
  \end{align*}
  where $C^*$ is a $D$-code in $\HH_n(\FF)$ of maximum size. The same proof also applies to $D$-codes over
  the strong Hamming scheme $\HH_n^*(\FF)$.
\end{remark}

\subsection{Hierarchy Collapse for General Codes}\label{sec:lifting}

In this section, we show that without the additional semantic linearity constraints imposed by
$\KLP_\Lin(n,d,\ell)$, the associated hierarchy $\KLP(n,d,\ell)$ does not give any improvement over the
original Delsarte linear programming approach. The proof is in two steps: first, we show that just tensoring
the program does not change the relative value (\cref{lem:tensorlifting}). Second, we show that refining the
scheme and adding only natural (non-semantic) constraints does not change the value of the associated Delsarte
linear program (\cref{lem:refinementlifting}).

Recall that all of our linear programming hierarchies can be interpreted as a Delsarte LP of some association
scheme and some code constraints (see \cref{rmk:klp}). For these proofs, we then heavily rely on the
connection to the unsymmetrized program $\vartheta'$ in \cref{thm:vartheta'}.

\begin{lemma}\label{lem:tensorlifting}
  Let $S_1=(X_1,R_1)$ and $S_2=(X_2,R_2)$ be commutative association schemes and let $D_i\subseteq R_i$ with
  $\cD_{X_i}\in D_i$ ($i\in[2]$). Then
  \begin{align*}
    \val(\cL_{S_1\otimes S_2}(D_1\otimes D_2)) & = \val(\cL_{S_1}(D_1))\cdot\val(\cL_{S_2}(D_2)),
  \end{align*}
  where
  \begin{align*}
    D_1\otimes D_2 & \coloneqq \{r_1\otimes r_2 \mid r_1\in D_1\land r_2\in D_2\}.
  \end{align*}
\end{lemma}

\begin{proof}
  By \cref{thm:vartheta'}, for $i\in[2]$, we have
  \begin{align*}
    \val(\cL_{S_i}(D_i)) & = \vartheta'(G_{S_i}(D_i)), &
    \val(\cL_{S_1\otimes S_2}(D_1\otimes D_2)) & =\vartheta'(G_{S_1\otimes S_2}(D_1\otimes D_2)).
  \end{align*}

  Given solutions $M_1$ and $M_2$ of the primal semi-definite programs $\cS(G_{S_1}(D_1))$ and
  $\cS(G_{S_2}(D_2))$ associated with $\vartheta'(G_{S_1}(D_1))$ and $\vartheta'(G_{S_2}(D_2))$, respectively,
  note that the tensor product $M\coloneqq M_1\otimes M_2$ is a feasible solution of $\cS(G_{S_1\otimes
    S_2}(D_1\otimes D_2))$ since for if $((x_1,x_2),(y_1,y_2))\in r_1\otimes r_2$ for some $r_1\otimes r_2\in
  (R_1\otimes R_2)\setminus (D_1\otimes D_2)$, then $(x_1,y_1)\in r_1$ or $(x_2,y_2)\in r_2$, which implies
  that $M_{(x_1,x_2),(y_1,y_2)} = 0$. Since we also have
  \begin{align*}
    \langle J, M\rangle & = \langle J,M_1\rangle\cdot\langle J,M_2\rangle,
  \end{align*}
  it follows that
  \begin{align*}
    \val(\cL_{S_1\otimes S_2}(D_1\otimes D_2)) & \geq \val(\cL_{S_1}(D_1))\cdot\val(\cL_{S_2}(D_2)).
  \end{align*}

  \medskip

  For the other inequality, given solutions $(\beta_1,N^1)$ and $(\beta_2,N^2)$ of the dual semi-definite
  programs $\cS'(G_{S_1}(D_1))$ and $\cS'(G_{S_2}(D_2))$, respectively, let $\beta\coloneqq\beta_1\cdot\beta_2$
  and $N\coloneqq N^1\otimes N^2$. Since $N^i\preceq \beta_i I$ ($i\in[2]$), it follows that $N\preceq\beta
  I$. Note also that if $((x_1,x_2),(y_1,y_2))\in r_1\otimes r_2$ for some $r_1\otimes r_2\in D_1\otimes D_2$,
  then since $r_i\in D_i$, we must have
  \begin{align*}
    N_{((x_1,x_2),(y_1,y_2))} & = N^1_{x_1,y_1}\cdot N^2_{x_2,y_2} \geq 1,
  \end{align*}
  thus $(\beta,N)$ is a feasible solution of $\cS'(G_{S_1\otimes S_2}(D_1\otimes D_2))$ which implies
  \begin{align*}
    \val(\cL_{S_1\otimes S_2}(D_1\otimes D_2)) & \leq \val(\cL_{S_1}(D_1))\cdot\val(\cL_{S_2}(D_2)),
  \end{align*}
  as desired.
\end{proof}

We now proceed to refinements.

\begin{lemma}\label{lem:refinementlifting}
  Let $S'=(X,R')$ be a commutative refinement of a commutative association scheme $S=(X,R)$, let $D\subseteq
  R$ be such that $\cD_X\in D$ and let
  \begin{align*}
    D' & \coloneqq \{r\in R'\mid\exists \widehat{r}\in R, r\subseteq \widehat{r}\}.
  \end{align*}
  Then
  \begin{align*}
    \val(\cL_S(D)) & = \val(\cL_{S'}(D')).
  \end{align*}
\end{lemma}

\begin{proof}
  By \cref{thm:vartheta'}, we have
  \begin{align*}
    \val(\cL_S(D)) & = \vartheta'(G_S(D)), &
    \val(\cL_{S'}(D')) & = \vartheta'(G_{S'}(D')).
  \end{align*}
  But note that the semi-definite programs $\cS(G_S(D))$ and $\cS(G_{S'}(D'))$ corresponding to
  $\vartheta'(G_S(D))$ and $\vartheta'(G_{S'}(D'))$ are identical (i.e., have exactly the same restrictions
  and objective value), so we get $\vartheta'(G_S(D)) = \vartheta'(G_{S'}(D'))$ trivially.
\end{proof}

Composing \cref{lem:tensorlifting,lem:refinementlifting}, we conclude that if we do not add
any extra restrictions other than the natural ones, the value of the Delsarte linear program remains
unchanged.

\begin{proposition}[Lifting]\label{prop:main:lifting:formal}
  For every finite field $\FF$ and every $\ell\in\NN_+$, we have
  \begin{align*}
    \val(\KLP^\FF(n,d,\ell))^{1/\ell}  = \val(\DLP^\FF(n,d)).
  \end{align*}
\end{proposition}

\begin{proof}
  By \cref{rmk:klp}, the program $\KLP^\FF(n,d,\ell)$ can be seen as the Delsarte linear program
  $\cL_{\HH_n(\FF)^{\ell,\FF^\ell}}(\widehat{D}_d^\ell)$ of the refinement
  $\HH_n(\FF)^{\ell,\FF^\ell}$ of the tensor power $\HH_n(\FF)^\ell$ using the natural restrictions
  \begin{align*}
    \widehat{D}_d^\ell & \coloneqq \{r\in R^{\ell,\FF_2^\ell} \mid \exists r'\in R^{\otimes\ell}, r\subseteq r'\},
  \end{align*}
  so the result follows from \cref{lem:tensorlifting,lem:refinementlifting}.
\end{proof}

As a secondary corollary, we can also show that in the linear case, the logarithm of the value of the hierarchy
is subadditive. Let us note that this is also true of the hierarchy for non-linear codes for trivial reasons.

\begin{corollary}
  For every finite field $\FF$ and every $\ell_1,\ell_2\in\NN_+$, we have
  \begin{multline*}
    \val(\KLP_\Lin^\FF(n,d,\ell_1+\ell_2))
    \\
    \leq
    \val(\KLP_\Lin^\FF(n,d,\ell_1))\cdot\val(\KLP_\Lin^\FF(n,d,\ell_2)),
  \end{multline*}
\end{corollary}

\begin{proof}
  By \cref{rmk:klp}, the program $\KLP_\Lin^\FF(n,d,\ell)$ can be seen as the Delsarte linear program
  $\cL_{S_\ell}(D_d^{\ell,\FF^\ell})$ of the translation scheme $S_\ell\coloneqq\HH_n(\FF)^{\ell,\FF^\ell}$.

  Note also that for $\ell_1,\ell_2\in\NN_+$, if
  \begin{align*}
    \widehat{D}_d^{\ell_1,\ell_2}
    & \coloneqq
    \{r\in R^{\ell_1+\ell_2,\FF^{\ell_1+\ell_2}}
    \mid
    \exists r_1\in D_d^{\ell_1,\FF^{\ell_1}},\exists r_2\in D_d^{\ell_2,\FF^{\ell_2}},
    r\subseteq r_1\otimes r_2\},
  \end{align*}
  then we have
  \begin{align*}
    D_d^{\ell_1+\ell_2,\FF^{\ell_1+\ell_2}}\subseteq\widehat{D}_d^{\ell_1,\ell_2},
  \end{align*}
  and thus we get
  \begin{align*}
    & \!\!\!\!\!\!
    \val(\KLP_\Lin^\FF(n,d,\ell_1+\ell_2))
    \\
    & =
    \val(\cL_{S_{\ell_1+\ell_2}}(D_d^{\ell_1+\ell_2,\FF^{\ell_1+\ell_2}}))
    \\
    & \leq
    \val(\cL_{S_{\ell_1+\ell_2}}(\widehat{D}_d^{\ell_1,\ell_2}))
    \\
    & =
    \val(\cL_{S_{\ell_1}}(D_d^{\ell_1,\FF^{\ell_1}}))\cdot
    \val(\cL_{S_{\ell_2}}(D_d^{\ell_2,\FF^{\ell_2}}))
    \\
    & =
    \val(\KLP_\Lin^\FF(n,d,\ell_1))\cdot\val(\KLP_\Lin^\FF(n,d,\ell_2)),
  \end{align*}
  where the second equality follows from \cref{lem:tensorlifting,lem:refinementlifting}.
\end{proof}


\section{Conclusion}
\label{sec:concl}

In this paper, we presented a pair of hierarchies of linear programs $\KLP^\FF(n,d,\ell)$ and
$\KLP^\FF_\Lin(n,d,\ell)$ that provide upper bounds for the maximum size of codes and linear codes,
respectively, of distance $d$ in the weak Hamming scheme $\HH_n(\FF)$ over a finite field $\FF$. We
showed that while the first hierarchy $\KLP^\FF(n,d,\ell)$ collapses, the second hierarchy obtains the true
value of the maximum code up to rounding by level $\ell = O(n^2)$. Finally, we also showed how to extend these
hierarchy constructions to translation schemes under the mild assumption of factoring through types.

\medskip

As we mentioned in the introduction, we view the main contribution of $\KLP_\Lin$ as being a hierarchy that is
sufficiently powerful to ensure completeness while still being sufficiently simple to remain a hierarchy of linear
programs (as opposed to SDPs), and bearing enough similarities with the original Delsarte's LP to be amenable to
theoretical analysis. Thus the main open problem is to provide better upper or lower bounds to the optimum
value of $\KLP_\Lin$.

The contrast between completeness of $\KLP_\Lin$ and collapse of $\KLP$ also surfaces a very natural
question: are optimum codes very far from being linear? Along these lines, note that at level $\ell$,
$\KLP_\Lin(n,d,\ell)$ does not require full linearity of a code; namely, if $C\subseteq\FF_2^n$ satisfies
\begin{align}\label{eq:condition}
  \Delta\left(\sum_{j=1}^t x_j, \sum_{j=1}^t y_j\right) \notin [d-1],
\end{align}
for every $t\leq\ell$ and every $x_1,\ldots,x_t,y_1,\ldots,y_t\in C$, then $a^C$ is a feasible solution of the
program $\KLP_\Lin(n,d,\ell)$. For constant $\ell$, the condition~\eqref{eq:condition} is extremely mild and
much weaker than $C$ being a linear (or even affine) code. For example, if $0\in C$, then~\eqref{eq:condition}
boils down to requiring sums of at most $2\ell$ codewords from $C$ to not have Hamming weight in $[d-1]$. This
makes studying $\KLP_\Lin(n,d,\ell)$ at constant levels $\ell$ quite interesting.

In \cref{theo:main:completeness:formal}, we showed the (approximate)
completeness of $\KLP^\FF_\Lin(n,d,\ell)$ at level $O(n^2)$, via an
unusual counting argument. The hierarchy does not have the same
conceptual structure as Sum-of-Squares or Sherali--Adams, so
completeness does not follow in the same way. In an
earlier version of this manuscript, we conjectured that level $n$
would have exact completeness, and we believe we now have a proof of this
result\footnote{To
 give appropriate time for the verification of this proof, we leave
  it to a future work. We are including this note here to
  alert the interested reader that a proof might now be known.}. It is
plausible that exact completeness of $\KLP^\FF_\Lin(n,d,\ell)$ can be
attained at level $O(k_0)$, where $k_0$ is the dimension of an optimum
linear code over $\FF$ of distance $d$ and blocklength $n$.

As we mentioned in the introduction, our techniques provide a higher-order version of the linear program
responsible for the first linear programming bound in~\cite{MRRW77}. The second linear programming bound
in~\cite{MRRW77} also consists of analyzing a Delsarte LP but for the Johnson scheme instead of the Hamming
scheme. However, since the Johnson scheme is not a translation scheme, one cannot apply the theory developed
in \cref{subsec:ftt} directly. It is then natural to ask if there is a suitable generalization of this
construction that would apply to non-translation schemes such as the Johnson scheme.

In \cref{subsec:ftt}, we showed how to generalize the hierarchy constructions to translation schemes under the
assumption of factoring through types. However, in the general case it is not clear that the $p$ and
$q$-functions of $S^{\ell,T}$ can be computed efficiently even if those of $S$ can be computed
efficiently. For the particular case of the binary Hamming scheme, we obtained efficient formulas in
\cref{lem:explicitkrawtchouk,lem:recursive} (see also \cref{prop:complexity}), but one can also compute the
higher-order Krawtchouk polynomials efficiently from the usual Krawtchouk polynomials. This raises the natural
question: for a translation scheme $S$ in which $f_{S,T}$ factors through types, can the $p$ and $q$-functions
of $S^{\ell,T}$ be efficiently computed from the $p$ and $q$-functions of $S$?

For the particular case of the strong Hamming scheme $\HH_n^*(\FF)$ over an arbitrary finite field $\FF$, an
efficient formula for the higher-order $\FF$-Krawtchouk polynomials can be obtained by generalizing
\cref{lem:explicitkrawtchouk}: first, one generalizes the notion of Venn diagram configuration by saying that
$(x_1,\ldots,x_\ell)\in(\FF^n)^\ell$ has $\FF$-Venn diagram configuration $g\colon\FF^\ell\to\{0,1,\ldots,n\}$ if
\begin{align*}
  g(t) & = \lvert\{i\in[n] \mid \forall j\in[\ell], (x_j)_i = t_j\}\rvert
\end{align*}
for every $t\in\FF^\ell$. \cref{lem:distlincomb} implies that $f_{\HH_n^*(\FF),\FF^\ell}(x) =
f_{\HH_n^*(\FF),\FF^\ell}(y)$ if and only if $x$ and $y$ have the same $\FF$-Venn diagram configuration. By
indexing the $\FF$-Krawtchouk polynomials of order $\ell$ by $\FF$-Venn diagram configurations, a proof
analogous to that of \cref{lem:explicitkrawtchouk} gives
\begin{align*}
  K_h(g)
  & =
  \sum_{F\in\cF} \prod_{t\in\FF^\ell} \frac{g(t)!}{\prod_{u\in\FF^\ell} F(t,u)!} \prod_{t,u\in\FF^\ell} \chi_t(u)^{F(t,u)},
\end{align*}
for all $\FF$-Venn diagram configurations $h,g\colon\FF^\ell\to\{0,1,\ldots,n\}$, where $\cF$ is the set of
functions $F\colon \FF^\ell\times\FF^\ell\to\{0,1,\ldots,n\}$ such that
\begin{align*}
  \forall t\in\FF^\ell, \sum_{u\in\FF^\ell} F(t,u) = g(t),\\
  \forall u\in\FF^\ell, \sum_{t\in\FF^\ell} F(t,u) = h(u).
\end{align*}
One can also obtain efficient formulas for the $\FF$-Krawtchouk polynomials of order $\ell$ in the weak
Hamming scheme $\HH_n(\FF)$ with similar methods (but the formulas are considerably more complicated).


\bibliographystyle{alphaurl}
\bibliography{macros,references}

\appendix

\section{Deferred Binary Case Proofs}\label{sec:binaryproofs}

Since the proofs of \cref{lem:configcount,lem:configorbits} use \cref{lem:configconversion}, we postpone them
until after the proof of the latter.

\validVDconfigs*

\begin{proof}
  It is obvious that every Venn diagram configuration $g$ is in the set in right-hand side
  of~\eqref{eq:validVDconfigs}. On the other hand, if $g$ is in the set in the right-hand side
  of~\eqref{eq:validVDconfigs}, then the hypotheses imply that we can find a partition
  $(X_J)_{J\subseteq[\ell]}$ of $[n]$ into $2^\ell$ parts such that $\lvert X_J\rvert = g(J)$. It is easy to
  see that the words $z_1,\ldots,z_\ell\in\FF_2^n$ defined by
  \begin{align*}
    (z_j)_i & \df \One[\exists J\subseteq[\ell], (j\in J\land i\in X_J)]\qquad (i\in[n],j\in[\ell])
  \end{align*}
  have Venn diagram configuration $g$.
\end{proof}

\configconversion*

\begin{proof}
  First note that for $g\in Z_{n,\ell}$, we have
  \begin{align*}
    \sum_{J\subseteq[\ell]} V_{n,\ell}(g)(J)
    & =
    n
    +
    2^{1-\ell}\sum_{T\subseteq[\ell]} g(T) \sum_{J\subseteq[\ell]} (-1)^{\lvert T\cap J\rvert-1}
    =
    n
    -
    2^{1-\ell}\sum_{T\subseteq[\ell]} g(\varnothing) 2^\ell
    =
    n,
  \end{align*}
  so $V_{n,\ell}$ is well-defined. Since for $g\in S_{n,\ell}$, we clearly have $D_{n,\ell}(g)(\varnothing) =
  0$, it follows that $D_{n,\ell}$ is also well-defined.

  Let now $g\in S_{n,\ell}$ and note that
  \begin{equation}\label{eq:VD}
    \begin{aligned}
      V_{n,\ell}(D_{n,\ell}(g))(J)
      & =
      n\cdot\One[J = \varnothing]
      +
      2^{1-\ell}\sum_{T\subseteq[\ell]} (-1)^{\lvert T\cap J\rvert-1}
      \sum_{\substack{K\subseteq[\ell]\\\lvert K\cap T\rvert\text{ odd}}} g(K)
      \\
      & =
      n\cdot\One[J = \varnothing]
      +
      2^{1-\ell}\sum_{K\subseteq[\ell]} g(K)
      \sum_{\substack{T\subseteq[\ell]\\\\\lvert K\cap T\rvert\text{ odd}}} (-1)^{\lvert T\cap J\rvert-1}.
    \end{aligned}
  \end{equation}

  But note that
  \begin{align*}
    \sum_{\substack{T\subseteq[\ell]\\\\\lvert K\cap T\rvert\text{ odd}}} (-1)^{\lvert T\cap J\rvert-1}
    & =
    \sum_{T\subseteq[\ell]} (-1)^{\lvert T\cap J\rvert-1}\frac{1 - (-1)^{\lvert K\cap T\rvert}}{2}
    \\
    & =
    \frac{1}{2}
    \left(-\sum_{T\subseteq[\ell]} (-1)^{\lvert T\cap J \rvert}
    + \sum_{T\subseteq[\ell]} (-1)^{\lvert T\cap (J\symdiff K)\rvert}
    \right)
    \\
    & =
    2^{\ell-1}(-\One[J = \varnothing] + \One[J = K]),
  \end{align*}
  so plugging this in~\eqref{eq:VD}, we get
  \begin{align*}
    V_{n,\ell}(D_{n,\ell}(g))(J)
    & =
    n\cdot\One[J = \varnothing]
    +
    \sum_{K\subseteq[\ell]} g(K)
    (-\One[J=\varnothing] + \One[J = K])
    =
    g(J),
  \end{align*}
  where the second equality follows since $\sum_{K\subseteq[\ell]} g(K) = n$ as $g\in S_{n,\ell}$.

  Therefore $V_{n,\ell}$ is a left-inverse of $D_{n,\ell}$. But since both $S_{n,\ell}$ and $Z_{n,\ell}$ are
  $\RR$-linear subspaces of dimension $2^\ell-1$ and $V_{n,\ell}$ and $D_{n,\ell}$ are $\RR$-linear, it
  follows that $V_{n,\ell}$ and $D_{n,\ell}$ are inverses of each other.

  By \cref{lem:validVDconfigs}, we know that $\im(\config_{n,\ell}^V)\subseteq S_{n,\ell}$. On the other hand,
  if $g\in\im(\config_{n,\ell}^V)$ and $\config_{n,\ell}^V(z_1,\ldots,z_\ell)=g$, then it is straightforward
  to check that $\config_{n,\ell}^\Delta(z_1,\ldots,z_\ell) = D_{n,\ell}(g)$, thus $\config_{n,\ell}^\Delta =
  D_{n,\ell}\comp\config_{n,\ell}^V$. Applying $V_{n,\ell}$ to both sides, we get
  $V_{n,\ell}\comp\config_{n,\ell}^\Delta = \config_{n,\ell}^V$.
\end{proof}

\configcount*

\begin{proof}
  By \cref{lem:configconversion}, it is sufficient to prove that $\lvert\im(\config_{n,\ell}^V)\rvert =
  \binom{n+2^\ell-1}{2^\ell-1}$. But the number of valid Venn diagram configurations is easy to count using
  \cref{lem:validVDconfigs}: it is exactly the number of partitions of $n$ indistinguishable objects into
  $2^\ell$ distinguishable parts, which is $\binom{n + 2^\ell - 1}{2^\ell - 1}$.
\end{proof}

\configorbits*

\begin{proof}
  By \cref{lem:configconversion}, item~\ref{lem:configorbits:SD} is equivalent to:
  \begin{enumerate}[start=3]
  \item $\config_{n,\ell}^V(x_1,\ldots,x_\ell) = \config_{n,\ell}^V(y_1,\ldots,y_\ell)$.
    \label{lem:configorbits:VD}
  \end{enumerate}

  \medskip

  Let us prove that~\ref{lem:configorbits:orbit}$\implies$\ref{lem:configorbits:VD}, if
  $(y_1,\ldots,y_\ell) = (x_1,\ldots,x_\ell)\cdot\sigma$ for some $\sigma\in S_n$, then for every
  $J\subseteq[\ell]$, we have
  \begin{align*}
    \config_{n,\ell}^V(y_1,\ldots,y_\ell)
    & =
    \lvert\{i\in[n] \mid \{j\in[\ell]\mid (y_j)_i = 1\} = J\}\rvert
    \\
    & =
    \lvert\{i\in[n] \mid \{j\in[\ell]\mid (x_j)_{\sigma(i)} = 1\} = J\}\rvert
    =
    \config_{n,\ell}^V(x_1,\ldots,x_\ell).
  \end{align*}

  \medskip

  To show~\ref{lem:configorbits:VD}$\implies$\ref{lem:configorbits:orbit}, let $(X_J)_{J\subseteq[\ell]}$ and
  $(Y_J)_{J\subseteq[\ell]}$ be the partitions corresponding to $(x_1,\ldots,x_\ell)$ and
  $(y_1,\ldots,y_\ell)$, respectively, given by
  \begin{align*}
    X_J & \df \bigcap_{j\in J}\supp(x_j)\cap\bigcap_{j\in[\ell]\setminus J}([n]\setminus\supp(x_j)),\\
    Y_J & \df \bigcap_{j\in J}\supp(y_j)\cap\bigcap_{j\in[\ell]\setminus J}([n]\setminus\supp(y_j)).
  \end{align*}

  Since $\config_{n,\ell}^V(x_1,\ldots,x_\ell)=\config_{n,\ell}^V(y_1,\ldots,y_\ell)$, it follows that $\lvert
  X_J\rvert = \lvert Y_J\rvert$ for every $J\subseteq[\ell]$, so there exists a permutation $\sigma\in S_n$
  such that $\sigma(X_J) = Y_J$ for every $J\subseteq[\ell]$. Since
  \begin{align*}
    (x_j)_i & \df \One[\exists J\subseteq[\ell], j\in J\land i\in X_J],\\
    (y_j)_i & \df \One[\exists J\subseteq[\ell], j\in J\land i\in Y_J],
  \end{align*}
  for every $j\in[\ell]$ and every $i\in[n]$, it follows that $(x_1,\ldots,x_\ell)\cdot\sigma = (y_1,\ldots,y_\ell)$.
\end{proof}

\configsize*

\begin{proof}
  By \cref{lem:configconversion}, $\lvert g\rvert$ is precisely the number of $(z_1,\ldots,z_\ell)\in\FF_2^n$
  whose Venn diagram configuration is $V_{n,\ell}(g)$. But the set of such $(z_1,\ldots,z_\ell)$ is naturally in
  bijection with the set of partitions $(X_J)_{J\subseteq[\ell]}$ of $[n]$ such that $\lvert X_J\rvert = V_{n,\ell}(g)(J)$
  ($J\subseteq[\ell]$) and the number of the latter is clearly the multinomial
  \begin{align*}
    \binom{n}{V_{n,\ell}(g)}
    & =
    \frac{n!}{\prod_{J\subseteq[\ell]} V_{n,\ell}(g)(J)!}.
  \end{align*}

  Finally, from~\eqref{eq:highkrawtchouk}, we also have
  \begin{align*}
    K_g(0)
    & =
    \sum_{(y_1,\ldots,y_\ell)\in g} \prod_{j=1}^\ell \chi_{y_j}(0)
    =
    \lvert g\rvert.
    \qedhere
  \end{align*}
\end{proof}

\orthogonality*

\begin{proof}
  By \cref{rmk:highkrawtchouk}, we have
  \begin{align*}
    K_h(g)
    & =
    2^{\ell n}\cdot \widehat{\One_h}(x),
    &
    K_{h'}(g)
    & =
    2^{\ell n}\cdot \widehat{\One_{h'}}(x),
  \end{align*}
  and thus we have
  \begin{align*}
    \sum_{g\in\im(\config_{n,\ell}^\Delta)}
    \lvert g\rvert\cdot K_h(g)\cdot K_{h'}(g)
    & =
    \sum_{x\in(\FF_2^n)^\ell}
    2^{2\ell n}\cdot\widehat{\One_h}(x)\cdot\widehat{\One_{h'}}(x)
    \\
    & =
    2^{2\ell n}\cdot\ip{\One_h}{\One_{h'}}
    \\
    & =
    2^{\ell n}\cdot\lvert h\rvert\cdot\One[h = h'],
  \end{align*}
  as desired.
\end{proof}

Since the proof of \cref{lem:reflection} uses \cref{lem:explicitkrawtchouk}, we prove the latter first.

\explicitkrawtchouk*

\begin{proof}
  For an $\ell$-tuple $z=(z_1,\ldots,z_\ell)\in(\FF_2^n)^\ell$, let $P^z = (P^z_J)_{J\subseteq[\ell]}$ be the
  natural partition of $[n]$ associated with $z$ given by
  \begin{align*}
    P^z_J & \coloneqq \bigcap_{j\in J}\supp(z_j)\cap\bigcap_{j\in[\ell]\setminus J} ([n]\setminus\supp(z_j)).
  \end{align*}
  Note that by \cref{lem:configconversion}, if the symmetric difference configuration of $z$ is some function
  $f$, then it has Venn diagram configuration $V_{n,\ell}(f)$ and thus $\lvert P^z_J\rvert = V_{n,\ell}(f)(J)$
  for every $J\subseteq[\ell]$.
  
  Fix an $\ell$-tuple $x=(x_1,\ldots,x_\ell)\in g$ whose symmetric difference configuration is $g$. We now
  classify the $\ell$-tuples $y=(y_1,\ldots,y_\ell)\in h$ of symmetric difference configuration $h$ based on
  how the partitions $P^x$ and $P^y$ interact; namely, to each such $y$ we associate the function $F_y\colon
  2^{[\ell]}\times 2^{[\ell]}\to\{0,1,\ldots,n\}$ given by
  \begin{align*}
    F_y(J,K) & \coloneqq \lvert P^x_J\cap P^y_K\rvert
  \end{align*}
  for every $J,K\subseteq[\ell]$.

  By our previous observations, we know that for every $J\subseteq[\ell]$, we have
  \begin{align*}
    \sum_{K\subseteq[\ell]} F_y(J,K)
    & =
    \sum_{K\subseteq[\ell]} \lvert P^x_J\cap P^y_K\rvert
    =
    \lvert P^x_J\rvert
    =
    V_{n,\ell}(g)(J).
  \end{align*}
  Similarly, we know that for every $K\subseteq[\ell]$, we have
  \begin{align*}
    \sum_{J\subseteq[\ell]} F_y(J,K)
    & =
    \sum_{J\subseteq[\ell]} \lvert P^x_J\cap P^y_K\rvert
    =
    \lvert P^y_K\rvert
    =
    V_{n,\ell}(h)(K).
  \end{align*}
  Therefore $F_y\in\cF$.

  Note further that for each $j\in[\ell]$ we have
  \begin{align*}
    \supp(x_j)\cap\supp(y_j)
    & =
    \bigcup_{\substack{J,K\subseteq[\ell]\\j\in J\cap K}} P^x_J\cap P^y_K,
  \end{align*}
  so in the formula~\eqref{eq:highkrawtchouk}, the summand of $y\in h$ is given by
  \begin{align*}
    \prod_{j=1}^\ell \chi_{y_j}(x_j)
    & =
    \prod_{j=1}^\ell (-1)^{\supp(x_j)\cap\supp(y_j)}
    =
    \prod_{j=1}^\ell \prod_{\substack{J,K\subseteq[\ell]\\j\in J\cap K}} (-1)^{F(J,K)}.
  \end{align*}

  For each $F\in\cF$, let $n_F$ be the number of $y\in h$ such that $F_y=F$. It is easy to compute $n_F$ from
  the definition of $F_y$: since $F_y = F$ if and only if the partition $(P^x_J\cap
  P^y_K)_{J,K\subseteq[\ell]}$ satisfies $\lvert P^x_J\cap P^y_K\rvert = F(J,K)$, it follows that to get $F =
  F_y$, each part $P^x_J$ (whose size is $V_{n,\ell}(g)$) has to be partitioned into $2^\ell$ parts of sizes
  $(F(J,K))_{K\subseteq[\ell]}$ and thus $n_F$ is given by the following product of multinomials
  \begin{align*}
    n_F
    & =
    \prod_{J\subseteq[\ell]} \binom{V_{n,\ell}(g)(J)}{F(J,\place)}
    =
    \prod_{J\subseteq[\ell]} \frac{V_{n,\ell}(g)(J)!}{\prod_{K\subseteq[\ell]} F(J,K)!}.
  \end{align*}

  Putting everything together, we get
  \begin{align*}
    K_h(g)
    & =
    \sum_{F\in\cF} n_F\cdot
    \prod_{j=1}^\ell \prod_{\substack{J,K\subseteq[\ell]\\j\in J\cap K}}
    (-1)^{F(J,K)}
    \\
    & =
    \sum_{F\in\cF}
    \prod_{J\subseteq[\ell]} \frac{V_{n,\ell}(g)(J)!}{\prod_{K\subseteq[\ell]} F(J,K)!}
    \cdot
    \prod_{j=1}^\ell \prod_{\substack{J,K\subseteq[\ell]\\ j\in J\cap K}} (-1)^{F(J,K)},
  \end{align*}
  as desired.  
\end{proof}

\reflection*

\begin{proof}
  Let $V_{n,\ell}$ the function of \cref{lem:configconversion} given by~\eqref{eq:V}. By
  \cref{lem:configsize}, we have
  \begin{align*}
    \frac{\lvert g\rvert}{\lvert h\rvert}
    & =
    \frac{\binom{n}{V_{n,\ell}(g)}}{\binom{n}{V_{n,\ell}(h)}}
    =
    \prod_{J\subseteq[\ell]} \frac{V_{n,\ell}(h)(J)!}{V_{n,\ell}(g)(J)!}
  \end{align*}
  and thus by using the formula of \cref{lem:explicitkrawtchouk}, we get
  \begin{align*}
    \frac{\lvert g\rvert}{\lvert h\rvert}
    \cdot
    K_h(g)
    & =
    K_g(h),
  \end{align*}
  which gives the result.
\end{proof}

\recursive*

\begin{proof}
  We start by proving~\eqref{eq:recursive1}.
  
  First note that if $K_0\subseteq[\ell]$ is such that $V_{n,\ell}(h)(K_0) > 0$ for some symmetric difference
  configuration $h\in\im(\config_{n,\ell}^\Delta)$, then \cref{lem:validVDconfigs,lem:configconversion} imply
  that $V_{n,\ell}(h) - \One_{\{K_0\}}$ is a Venn diagram configuration in the space $\FF_2^{n-1}$ (and level
  $\ell$) and thus $h\ominus K_0 = D_{n-1,\ell}(V_{n,\ell}(h) - \One_{\{K_0\}})$ is a symmetric difference
  configuration in the space $\FF_2^{n-1}$. This also shows that $g\ominus J_0$ is a symmetric difference
  configuration in the space $\FF_2^{n-1}$.

  Let us denote by $\cF_{g,h}$ the set of functions $F\colon 2^{[\ell]}\times 2^{[\ell]}\to\{0,1,\ldots,n\}$
  such that
  \begin{align*}
    \forall J\subseteq [\ell], \sum_{K\subseteq[\ell]} F(J,K) = V_{n,\ell}(g)(J),\\
    \forall K\subseteq [\ell], \sum_{J\subseteq[\ell]} F(J,K) = V_{n,\ell}(h)(K).
  \end{align*}
  We define $\cF_{g\ominus J_0,h\ominus K_0}$ analogously (replacing $n$ with $n-1$).

  By \cref{lem:explicitkrawtchouk}, we have
  \begin{align*}
    K_h(g)
    & =
    \sum_{F\in\cF_{g,h}}
    \prod_{J\subseteq[\ell]} \frac{V_{n,\ell}(g)(J)!}{\prod_{K\subseteq[\ell]} F(J,K)!}
    \cdot
    \prod_{j=1}^\ell\prod_{\substack{J,K\subseteq[\ell]\\j\in J\cap K}} (-1)^{F(J,K)}.
  \end{align*}

  Using the multinomial identity
  \begin{align*}
    \frac{V_{n,\ell}(g)(J_0)!}{\prod_{K\subseteq[\ell]} F(J_0,K)!}
    & =
    \binom{V_{n,\ell}(g)(J_0)}{F(J_0,\place)}
    =
    \sum_{\substack{K_0\subseteq[\ell]\\ F(J_0,K_0) > 0}}
    \frac{(V_{n,\ell}(g)(J_0)-1)!}{(F(J_0,K_0)-1)!\prod_{\substack{K\subseteq[\ell]\\K\neq K_0}} F(J_0,K)!}
  \end{align*}
  and noting that $F(J_0,K_0) > 0$ implies $V(h)(K_0) > 0$, we obtain
  \begin{align*}
    K_h(g)
    & =
    \sum_{\substack{K_0\subseteq[\ell]\\ V(h)(K_0) > 0}}
    \sum_{\substack{F\in\cF_{g,h}\\ F(J_0,K_0) > 0}}
    \frac{(V_{n,\ell}(g)(J_0)-1)!}{(F(J_0,K_0)-1)!\prod_{\substack{K\subseteq[\ell]\\K\neq K_0}} F(J_0,K)!},
    \cdot
    \prod_{j=1}^\ell\prod_{\substack{J,K\subseteq[\ell]\\j\in J\cap K}} (-1)^{F(J,K)}
    \\
    & =
    \sum_{K_0\subseteq[\ell]}
    \sum_{F'\in\cF_{g\ominus J_0, h\ominus K_0}}
    \prod_{J\subseteq[\ell]} \frac{V_{n-1,\ell}(g\ominus J_0)(J)!}{\prod_{K\subseteq[\ell]} F'(J,K)!}
    \cdot
    \prod_{j=1}^\ell\prod_{\substack{J,K\subseteq[\ell]\\j\in J\cap K}} (-1)^{F'(J,K)}
    \cdot (-1)^{\lvert J_0\cap K_0\rvert}
  \end{align*}
  where the second equality follows from the substitution corresponding to the bijection
  \begin{align*}
    \{F\in\cF_{g,h} \mid F(J_0,K_0) > 0\} \to \cF_{g\ominus J_0, h\ominus K_0}
  \end{align*}
  that maps $F$ to $F'\coloneqq F - \One_{\{(J_0,K_0)\}}$.

  Equation~\eqref{eq:recursive1} now follows by applying \cref{lem:explicitkrawtchouk} again.

  \medskip

  Note now that since~$h\oplus\varnothing\ominus\varnothing = h$, equation~\eqref{eq:recursive2} is equivalent
  to
  \begin{align*}
    \sum_{\substack{K_0\subseteq[\ell]\\ V(h)(K_0) > 0}}
    K_{h\oplus\varnothing\ominus K_0}(g)
    & =
    \sum_{\substack{K_0\subseteq[\ell]\\ V(h)(K_0) > 0}}
    (-1)^{\lvert J_0\cap K_0\rvert}\cdot K_{h\oplus\varnothing\ominus K_0}(g\oplus\varnothing\ominus J_0).
  \end{align*}
  By~\eqref{eq:recursive1}, both sides of the above are equal to $K_{h\oplus\varnothing}(g\oplus\varnothing)$:
  the left-hand side using~\eqref{eq:recursive1} with $J_0 = \varnothing$ and the right-hand side using $J_0 =
  J_0$.
\end{proof}

\section{Deferred Computations}
\label{sec:deferred}

\begin{lemma}\label{lemma:completeness_param_comp}
  Let $\epsilon\in(0,1)$ and $n,q\in\NN_+$ be positive integers with $q \ge 2$. For $\ell\geq 9(n^2\ln(q) +
  1)/(\ln(1+\epsilon))^2$, we have
  \begin{align*}
    (2 \ell^n q^{n^2})^{1/\ell} \leq 1 + \epsilon.
  \end{align*}
\end{lemma}

\begin{proof}
  The statement is equivalent to
  \begin{align*}
    \ln 2 + n\ln\ell + n^2\ln q \leq \ell\ln(1+\epsilon),
  \end{align*}
  which in turn is equivalent to
  \begin{align}\label{eq:reduced}
    \frac{n^2\ln q + \ln 2}{\ln (1+\epsilon)}
    & \leq
    \ell\left(1 - \frac{\ln \ell}{\ell}\cdot\frac{n}{\ln(1+\epsilon)}\right).
  \end{align}

  We claim that it is sufficient to show that
  \begin{align}\label{eq:eightninths}
    \frac{\ln \ell}{\ell} & \leq \frac{8}{9}\cdot\frac{\ln(1+\epsilon)}{n}.
  \end{align}
  Indeed, if this is the case, then the right-hand side of~\eqref{eq:reduced} is at least $\ell/9$, which in
  turn is at least $(n^2\ln(q) + 1)/(\ln(1+\epsilon))^2$ and thus~\eqref{eq:reduced} follows from
  $\ln(1+\epsilon)\leq \ln 2\leq 1$.

  To show~\eqref{eq:eightninths}, first note that the function $f(t)\df\ln(t)/t$ is decreasing when $t\geq e$
  and since
  \begin{align*}
    \ell
    & \geq
    \frac{9(n^2\ln(q) + 1)}{(\ln(1+\epsilon))^2}
    \geq
    \frac{9n^2\ln(q)}{(\ln(1+\epsilon))^2}
    \geq
    e,
  \end{align*}
  it is sufficient to prove that
  \begin{align}\label{eq:eightninthsfunction}
     f\(\frac{9n^2\ln(q)}{(\ln(1+\epsilon))^2}\) & \leq \frac{8}{9}\cdot\frac{\ln(1+\epsilon)}{n}.
  \end{align}

  But since
  \begin{align*}
    f\left(\frac{9n^2\ln(q)}{(\ln(1+\epsilon))^2}\right)
    & =
    \frac{\displaystyle
      \ln 9 + 2\ln n + \ln\ln q + 2\ln\frac{1}{\ln(1+\epsilon)}
    }{\displaystyle
      \frac{9n^2\ln(q)}{\ln(1+\epsilon)^2}
    },
  \end{align*}
  \eqref{eq:eightninthsfunction} is equivalent to
  \begin{align*}
    \ln 9 + 2\ln n + \ln\ln q + 2\ln\frac{1}{\ln(1+\epsilon)} & \leq \frac{8 n\ln(q)}{\ln(1+\epsilon)}.
  \end{align*}
  This is clearly true by recalling that $\ln(1+\epsilon)\leq 1$ and upper bounding the terms on the left-hand
  side of the above respectively by
  \begin{align*}
    3, & &
    2n, & &
    \ln(q), & &
    \frac{2}{\ln(1+\epsilon)}.
  \end{align*}
  The last three bounds follow from $\ln x\leq x$ for $x > 0$.
\end{proof}

\end{document}
